\documentclass[a4paper,UKenglish,cleveref,autoref,thm-restate]{lipics-v2021}
\pdfoutput=1 
\hideLIPIcs  
\nocite{arnborg1987complexity}
\nocite{bodlaender2005discovering}
\nocite{bonchi2022stringI}
\nocite{chekuri2000conjunctive}
\nocite{darwiche2011sdd}
\nocite{darwiche2002knowledge}
\nocite{de2021discopy}
\nocite{di2023monoidal}
\nocite{dudek2020parallel}
\nocite{fong2019hypergraph}
\nocite{fritz2023free}
\nocite{grohe2001evaluation}
\nocite{hasegawa2008finite}
\nocite{holtzen2020scaling}
\nocite{kimmig2017algebraic}
\nocite{kissinger2011pictures}
\nocite{kissinger2014finite}
\nocite{korhonen2023single}
\nocite{lauritzen1988local}
\nocite{lopuhaa2025fault}
\nocite{lopuhaa2022efficient}
\nocite{markov2008simulating}
\nocite{mclachlan2020bayesian}
\nocite{pearl1985bayesian}
\nocite{pearl1988probabilistic}
\nocite{peterseim2025unified}
\nocite{peterseim2026fixed}
\nocite{piedeleu2025complete}
\nocite{robertson1991graph}
\nocite{rubio2024rendering}
\nocite{ruijters2015fault}
\nocite{shpilka2010arithmetic}
\nocite{stamatelatos2002fault}
\nocite{stein2024probabilistic}
\nocite{stein2021structural}
\nocite{toumi2023discopy}
\nocite{wille2022decision}
\nocite{wilson2023data}
\nocite{xie2010using}
\nocite{zanasi2017rewriting}

\usepackage{tikz}
\usetikzlibrary{backgrounds}

\tikzstyle{tikzfig}=[baseline=-0.25em,scale=0.5]

\pgfkeys{/tikz/tikzit fill/.initial=0}
\pgfkeys{/tikz/tikzit draw/.initial=0}
\pgfkeys{/tikz/tikzit shape/.initial=0}
\pgfkeys{/tikz/tikzit category/.initial=0}

\pgfdeclarelayer{edgelayer}
\pgfdeclarelayer{nodelayer}
\pgfsetlayers{background,edgelayer,nodelayer,main}

\tikzstyle{none}=[inner sep=0mm]

\newcommand{\tikzfig}[1]{{\tikzstyle{every picture}=[tikzfig]
\IfFileExists{#1.tikz}
{\input{#1.tikz}}
{\IfFileExists{./figures/#1.tikz}
{\input{./figures/#1.tikz}}
{\tikz[baseline=-0.5em]{\node[draw=red,font=\color{red},fill=red!10!white] {\textit{#1}};}}}}}

\tikzstyle{every loop}=[] {}\lstdefinestyle{sql}{
language=SQL,
basicstyle=\small\ttfamily,
keywordstyle=\color{mutedblue}\bfseries,
commentstyle=\color{grey}\itshape,
showstringspaces=false,
morekeywords={JOIN, ON},
}

\tikzstyle{dot}=[fill=black, draw=black, shape=circle, tikzit shape=circle, minimum size=3pt, inner sep=0]
\tikzstyle{not dot}=[fill=white, draw={rgb,255: red,50; green,60; blue,87}, shape=circle, inner sep=0, minimum size=4pt]
\tikzstyle{light blue dot}=[fill={rgb,255: red,225; green,230; blue,248}, draw={rgb,255: red,48; green,60; blue,89}, shape=circle]
\tikzstyle{large dot}=[fill=black, draw=black, shape=circle, inner sep=0, minimum size=5pt]
\tikzstyle{op dot}=[fill=white, draw=black, shape=circle]
\tikzstyle{op turquoise}=[fill={rgb,255: red,229; green,255; blue,250}, draw=black, shape=circle]
\tikzstyle{op red}=[fill={rgb,255: red,255; green,228; blue,229}, draw=black, shape=circle]
\tikzstyle{op orange}=[fill={rgb,255: red,255; green,225; blue,183}, draw=black, shape=circle]
\tikzstyle{op brightblue}=[fill={rgb,255: red,216; green,237; blue,255}, draw=black, shape=circle]
\tikzstyle{box}=[fill=white, draw=black, shape=rectangle]
\tikzstyle{light greengrey box}=[fill=white, draw={rgb,255: red,191; green,204; blue,187}, shape=rectangle]
\tikzstyle{magenta box}=[fill=white, draw={rgb,255: red,217; green,22; blue,133}, shape=rectangle]
\tikzstyle{brightblue box}=[fill=white, draw={rgb,255: red,5; green,123; blue,249}, shape=rectangle]
\tikzstyle{truquoise box}=[fill=white, draw={rgb,255: red,60; green,183; blue,169}, shape=rectangle]
\tikzstyle{orange box}=[fill=white, draw={rgb,255: red,229; green,146; blue,12}, shape=rectangle]
\tikzstyle{green box}=[fill=white, draw={rgb,255: red,129; green,150; blue,22}, shape=rectangle]

\tikzstyle{dark blue solid}=[-, fill={rgb,255: red,225; green,230; blue,248}, draw={rgb,255: red,50; green,60; blue,87}, tikzit draw={rgb,255: red,50; green,60; blue,87}]
\tikzstyle{black solid}=[-]
\tikzstyle{dashed dark blue}=[-, draw={rgb,255: red,50; green,60; blue,87}, dashed]
\tikzstyle{dashed arrow}=[draw={rgb,255: red,50; green,60; blue,87}, ->, dashed]
\tikzstyle{red dashed}=[-, draw={rgb,255: red,188; green,49; blue,49}, dashed, line width=0.3mm]
\tikzstyle{green dashed}=[-, draw={rgb,255: red,129; green,150; blue,22}, dashed, line width=0.3mm]
\tikzstyle{orange dashed}=[-, draw={rgb,255: red,229; green,146; blue,12}, dashed, line width=0.3mm]
\tikzstyle{purple dashed}=[-, draw={rgb,255: red,147; green,83; blue,164}, dashed, line width=0.3mm]
\tikzstyle{brightblue dashed}=[-, draw={rgb,255: red,5; green,123; blue,249}, line width=0.3mm, dashed]
\tikzstyle{magenta dashed}=[-, draw={rgb,255: red,217; green,22; blue,133}, dashed, line width=0.3mm]
\tikzstyle{turquoise dahsed}=[-, draw={rgb,255: red,60; green,183; blue,169}, line width=0.3mm, dashed]
\tikzstyle{yellow dashed}=[-, draw={rgb,255: red,236; green,222; blue,27}, dashed, line width=0.5mm]
\tikzstyle{brown dashed}=[-, draw={rgb,255: red,127; green,84; blue,42}, dashed, line width=0.3mm]
\tikzstyle{black arrow}=[->]
\tikzstyle{lightgrey line}=[-, draw={rgb,255: red,219; green,219; blue,219}]
\tikzstyle{solid bluegrey}=[-, draw={rgb,255: red,138; green,146; blue,170}, line width=0.4mm] {} {}
\newtheorem{problem}[theorem]{Problem}

\newtheorem{mainresult}[theorem]{Main result}

\crefalias{example}{example}
\crefname{example}{example}{examples}
\Crefname{example}{Example}{Examples}

\definecolor{mutedblue}{RGB}{76, 128, 143}
\definecolor{darkblue}{RGB}{50, 60, 87}
\definecolor{grey}{RGB}{150, 150, 150}
\newcommand{\ass}[2]{\mathsf{A}_{#1}(#2)}
\newcommand{\lett}{\textcolor{mutedblue}{\mathsf{let}}\;}
\newcommand{\func}[1]{\textcolor{darkblue}{\mathtt{#1}}} {}

\title{Fixed-Parameter Tractable Inference for Discrete Probabilistic Programs, via String Diagram Algebraisation}
\titlerunning{FPT Inference for Discrete Probabilistic Programs}

\author{Benedikt Peterseim}{Universiteit Twente, Enschede, The Netherlands}{benedikt.peterseim@utwente.nl}{https://orcid.org/0009-0004-3510-4325}{}
\author{Milan Lopuha\"{a}-Zwakenberg}{Universiteit Twente, Enschede, The Netherlands}{m.a.lopuhaa@utwente.nl}{https://orcid.org/0000-0001-5687-854X}{}

\authorrunning{B. Peterseim and M. Lopuha\"{a}-Zwakenberg}
\Copyright{Benedikt Peterseim and Milan Lopuha\"{a}-Zwakenberg}

\ccsdesc[500]{Mathematics of computing~Probabilistic inference problems}
\ccsdesc[400]{Theory of computation~Fixed parameter tractability}
\ccsdesc[400]{Theory of computation~Categorical semantics}

\keywords{probabilistic programming, string diagrams, treewidth, fixed-parameter tractability, probabilistic inference, parameterised complexity, hypergraph categories}

\relatedversiondetails{Published version}{https://doi.org/10.4230/LIPIcs.LICS.2026.76}

\acknowledgements{We thank the anonymous reviewers for their thoughtful, thorough feedback.}

\funding{This work is partially supported by the European Research Council under grant agreement No 101187945 (RUBICON).}

\EventEditors{Claudia Faggian and Joost-Pieter Katoen}
\EventNoEds{2}
\EventLongTitle{41st Annual Symposium on Logic in Computer Science (LICS 2026)}
\EventShortTitle{LICS 2026}
\EventAcronym{LICS}
\EventYear{2026}
\EventDate{July 20--23, 2026}
\EventLocation{Lisbon, Portugal}
\EventLogo{}
\SeriesVolume{380}
\ArticleNo{76}

\nolinenumbers

\begin{document}

\maketitle

\begin{abstract}
Discrete probabilistic programs (DPPs) provide a highly expressive formalism for compactly defining arbitrary finite probabilistic models.
This expressivity comes at a price: DPP
inference is PSPACE-hard.
In this work, we show that DPP inference only takes polynomial time for programs that are ``structurally simple''.
More precisely, inference can be performed in polynomial time when the \emph{primal graph} of each function appearing in the probabilistic program has bounded treewidth,
and the inverse \emph{acceptance probability} is at most exponential in the size of the probabilistic program. Existing  algorithms do not achieve this performance guarantee.

Our method relies on finding suitable decompositions, \emph{algebraisations}, of the string diagrams underlying DPPs, employing existing algorithms for tree decompositions. This is independent of the probabilistic setting of DPPs and has direct applications to many problems, such as evaluating queries on relational databases and cybersecurity risk assessment via attack trees.
\end{abstract}

\section{Introduction}

\subsection{Discrete probabilistic programs (DPPs)} Probabilistic programming languages are declarative formal languages for specifying probabilistic models. 
We focus on \emph{discrete} probabilistic programs (DPPs)
as introduced in~\cite{holtzen2020scaling}, with Boolean variables. 
Thus we allow exactly the following constructs:
\begin{enumerate}
\item $\lett\!\!$ expressions, or \emph{assignments};
\item the basic Boolean connectives $\land, \lor, \lnot$;
\item a biased coin flip function $\func{flip}(p)$, for each $p\in [0,1]\cap \mathbb{Q}$;
\item $\func{observe}$ expressions, conditioning on an event being true;
\item function definitions.
\end{enumerate}
Extensions to more general finite types and constructs (such as bounded iteration) can be reduced to this Boolean core~\cite{holtzen2020scaling}.
An example of a DPP is given in Figure~\ref{fig:code-and-diagram}. 
DPPs are highly expressive and include, for example, \emph{fault trees}, which are ubiquitous in reliability engineering~\cite{ruijters2015fault,stamatelatos2002fault}, as well as \emph{Bayesian networks} \cite{pearl1985bayesian,pearl1988probabilistic}, which are used for probabilistic reasoning in many fields including healthcare~\cite{mclachlan2020bayesian} and cybersecurity~\cite{xie2010using}. 
Further applications of discrete probabilistic programming, 
and its role for probabilistic programming in general, 
are described in detail in~\cite{holtzen2020scaling}.

\begin{figure}[!t]
\centering
\vspace{-6mm}
\hspace{5mm}
\begin{minipage}[t]{0.36\linewidth}
\raggedright
\small
\vspace{-7mm}
\hrule height 0.4pt
\vspace{2mm}
\(
\begin{aligned}[t]
&\quad\func{test}(z) := \\
&\quad\quad\lett t_{\top|\top} = \func{flip}(0.99);\\
&\quad\quad\lett t_{\top|\bot} = \func{flip}(0.02);\\
&\quad\quad\lett t_{\func{TP}} = z \land t_{\top|\top};\\
&\quad\quad\lett t_{\func{FP}} = \lnot z \land t_{\top|\bot};\\
&\quad\quad\lett t = t_{\func{TP}} \lor t_{\func{FP}};\\
&\quad\quad t \\
\\
&\quad\func{hasDisease}(\,) := \\
&\quad\quad\lett x = \func{flip}(10^{-4});\\
&\quad\quad\lett t_1 = \func{test}(x);\\
&\quad\quad\lett t_2 = \func{test}(x);\\
&\quad\quad\lett y = t_1 \land \lnot t_2;\\
&\quad\quad \func{observe}(y);\\
&\quad\quad x
\end{aligned}
\)
\vspace{2mm}
\hrule height 0.4pt
\end{minipage}\hspace{-0.05\linewidth}\begin{minipage}[t]{0.54\linewidth}
\centering
\scalebox{0.60}{
\begin{tikzpicture}
	\begin{pgfonlayer}{nodelayer}
		\node [style=none] (14) at (-8.25, -17.75) {};
		\node [style=none] (15) at (-7.25, -16) {};
		\node [style=none] (16) at (-6.25, -17.75) {};
		\node [style=none] (18) at (-7.25, -16) {};
		\node [style=none] (19) at (-6.75, -17.75) {};
		\node [style=none] (20) at (-7.75, -17.75) {};
		\node [style=none] (21) at (-7.75, -18.5) {};
		\node [style=none] (22) at (-6.75, -18.5) {};
		\node [style=none] (23) at (-6.75, -14.5) {};
		\node [style=none] (24) at (-5.75, -12.5) {};
		\node [style=none] (25) at (-4.75, -14.5) {};
		\node [style=none] (26) at (-5.75, -13.75) {};
		\node [style=none] (27) at (-5.75, -11.25) {};
		\node [style=none] (28) at (-5.25, -14) {};
		\node [style=none] (29) at (-6.25, -14) {};
		\node [style=none] (30) at (-6.25, -14.25) {};
		\node [style=none] (31) at (-5.25, -14.25) {};
		\node [style=none] (32) at (-5.25, -17.75) {};
		\node [style=none] (33) at (-4.25, -16) {};
		\node [style=none] (34) at (-3.25, -17.75) {};
		\node [style=none] (35) at (-4.25, -16) {};
		\node [style=none] (36) at (-3.75, -17.75) {};
		\node [style=none] (37) at (-4.75, -17.75) {};
		\node [style=none] (38) at (-4.75, -18) {};
		\node [style=none] (39) at (-3.75, -18.25) {};
		\node [style=none] (40) at (-5.5, -20.75) {};
		\node [style=none] (41) at (-4.75, -19.5) {};
		\node [style=none] (42) at (-4, -20.75) {};
		\node [style=dot] (43) at (-4.75, -18.5) {};
		\node [style=none] (44) at (-4.75, -20.75) {};
		\node [style=light blue dot] (52) at (-6.25, -22.25) {$0.99$};
		\node [style=light blue dot] (57) at (-2.75, -22) {$0.02$};
		\node [style=none] (58) at (-7.75, -22.75) {};
		\node [style=none] (60) at (-4.75, -22.75) {};
		\node [style=none] (61) at (-6.25, -25.5) {};
		\node [style=dot] (62) at (-6.25, -24) {};
		\node [style=none] (64) at (-1.75, -20.5) {{\Large $t_{\top|\bot}$}};
		\node [style=none] (67) at (-5.75, -24.5) {{\Large $z$}};
		\node [style=none] (68) at (-5.8, -19.5) {{\Large $t_{\top|\top}$}};
		\node [style=none] (69) at (-7.75, -15) {{\Large $t_{\func{TP}}$}};
		\node [style=none] (70) at (-3.75, -15) {{\Large $t_{\func{FP}}$}};
		\node [style=none] (71) at (-5.25, -12) {{\Large $t$}};
		\node [style=not dot] (72) at (-4.75, -19.25) {};
		\node [style=none] (73) at (-9.75, -10.75) {};
		\node [style=none] (74) at (-0.75, -10.75) {};
		\node [style=none] (75) at (-9.75, -26.5) {};
		\node [style=none] (76) at (-0.75, -26.5) {};
		\node [style=dot] (78) at (-7, -15.25) {};
		\node [style=dot] (79) at (-4.5, -15.25) {};
		\node [style=dot] (80) at (-2.75, -20.5) {};
		\node [style=dot] (81) at (-6.75, -19.5) {};
		\node [style=dot] (82) at (-5.75, -12) {};
		\node [style=light blue dot] (83) at (-5.5, -7.75) {$10^{-4}$};
		\node [style=none] (84) at (-7, -5) {};
		\node [style=none] (85) at (-4.25, -5.25) {};
		\node [style=dot] (86) at (-5.5, -6.25) {};
		\node [style=none] (87) at (-5.5, -5.75) {{\Large $x$}};
		\node [style=none] (88) at (-7.5, -3.25) {};
		\node [style=none] (89) at (-6.5, -3.25) {};
		\node [style=none] (90) at (-7.5, -5) {};
		\node [style=none] (91) at (-6.5, -5) {};
		\node [style=none] (92) at (-4.25, -5.25) {};
		\node [style=none] (93) at (-4.75, -3.5) {};
		\node [style=none] (94) at (-3.75, -3.5) {};
		\node [style=none] (95) at (-4.75, -5.25) {};
		\node [style=none] (96) at (-3.75, -5.25) {};
		\node [style=none] (97) at (-6.25, 0.5) {};
		\node [style=none] (98) at (-5.25, 2.25) {};
		\node [style=none] (99) at (-4.25, 0.5) {};
		\node [style=dot] (100) at (-5.25, 3) {};
		\node [style=none] (101) at (-4.75, 0.5) {};
		\node [style=none] (102) at (-5.75, 0.5) {};
		\node [style=none] (103) at (-5.75, 0.25) {};
		\node [style=none] (104) at (-4.75, 0) {};
		\node [style=none] (105) at (-5.5, -2) {};
		\node [style=none] (106) at (-4.75, -0.75) {};
		\node [style=none] (107) at (-4, -2) {};
		\node [style=dot] (108) at (-4.75, 0) {};
		\node [style=none] (109) at (-4.75, -2) {};
		\node [style=not dot] (110) at (-4.75, -0.5) {};
		\node [style=none] (111) at (-4.25, -3.5) {};
		\node [style=dot] (112) at (-4.25, -3) {};
		\node [style=dot] (113) at (-7, -2) {};
		\node [style=none] (114) at (-7, -3.25) {};
		\node [style=none] (115) at (-7.5, -2) {{\Large $t_1$}};
		\node [style=none] (116) at (-3.5, -2.75) {{\Large $t_2$}};
		\node [style=none] (117) at (-5.25, 3.5) {};
		\node [style=none] (118) at (-4.75, 3) {{\Large $y$}};
		\node [style=none] (119) at (-6.25, 4) {};
		\node [style=none] (120) at (-4.25, 4) {};
		\node [style=large dot] (121) at (-5.25, 4) {};
		\node [style=none] (122) at (-5.25, 4.5) {};
		\node [style=none] (123) at (-5.75, 4) {};
		\node [style=none] (124) at (-4.75, 4) {};
		\node [style=none] (125) at (-5.75, 4) {};
		\node [style=none] (127) at (-3.75, -4.25) {};
		\node [style=none] (128) at (-7.5, -4) {};
		\node [style=none] (130) at (-7.75, -8.75) {};
		\node [style=none] (131) at (-3.5, -9) {};
		\node [style=none] (132) at (-5.25, -10.75) {};
		\node [style=none] (133) at (-2.25, -3) {};
		\node [style=none] (134) at (-3.25, 5.25) {};
		\node [style=none] (135) at (-8.75, -11.25) {{\Large $\func{test}$}};
		\node [style=none] (137) at (-3.5, -4.5) {};
	\end{pgfonlayer}
	\begin{pgfonlayer}{edgelayer}
		\draw [style=dark blue solid] (16.center)
			 to (14.center)
			 to [in=-180, out=90] (15.center)
			 to [in=90, out=0] cycle;
		\draw [style=dark blue solid] (18.center) to (15.center);
		\draw [style=dark blue solid] (20.center) to (21.center);
		\draw [style=dark blue solid] (19.center) to (22.center);
		\draw [style=dark blue solid] (25.center)
			 to [in=-30, out=90] (24.center)
			 to [in=90, out=210] (23.center)
			 to [in=195, out=60, looseness=0.75] (26.center)
			 to [in=120, out=-15, looseness=0.75] cycle;
		\draw [style=dark blue solid] (27.center) to (24.center);
		\draw [style=dark blue solid] (29.center) to (30.center);
		\draw [style=dark blue solid] (28.center) to (31.center);
		\draw [style=dark blue solid] (32.center)
			 to [in=-180, out=90] (33.center)
			 to [in=90, out=0] (34.center)
			 to cycle;
		\draw [style=dark blue solid] (35.center) to (33.center);
		\draw [style=dark blue solid] (37.center) to (38.center);
		\draw [style=dark blue solid] (36.center) to (39.center);
		\draw [style=dark blue solid] (42.center)
			 to (40.center)
			 to (41.center)
			 to cycle;
		\draw [style=dark blue solid] (43) to (41.center);
		\draw [style=dark blue solid] (43) to (38.center);
		\draw [style=black solid] (58.center) to (21.center);
		\draw [style=black solid, in=-90, out=180] (62) to (58.center);
		\draw [style=black solid, in=-90, out=0] (62) to (60.center);
		\draw [style=dark blue solid] (62) to (61.center);
		\draw [style=dashed dark blue] (75.center) to (73.center);
		\draw [style=dashed dark blue] (74.center) to (73.center);
		\draw [style=dashed dark blue] (74.center) to (76.center);
		\draw [style=dashed dark blue] (75.center) to (76.center);
		\draw [style=black solid, in=90, out=-90, looseness=1.50] (39.center) to (80);
		\draw [style=black solid] (80) to (57);
		\draw [style=black solid] (22.center) to (81);
		\draw [style=black solid, in=90, out=-90] (81) to (52);
		\draw [style=black solid] (44.center) to (60.center);
		\draw [style=black solid, in=60, out=-90] (30.center) to (78);
		\draw [style=black solid, in=90, out=-135] (78) to (18.center);
		\draw [style=black solid, in=90, out=-45] (79) to (35.center);
		\draw [style=black solid, in=-90, out=135] (79) to (31.center);
		\draw [style=black solid, in=-90, out=180] (86) to (84.center);
		\draw [style=black solid, in=-90, out=0] (86) to (85.center);
		\draw [style=black solid] (83) to (86);
		\draw [style=dark blue solid] (97.center)
			 to [in=-180, out=90] (98.center)
			 to [in=90, out=0] (99.center)
			 to cycle;
		\draw [style=dark blue solid] (100) to (98.center);
		\draw [style=dark blue solid] (102.center) to (103.center);
		\draw [style=dark blue solid] (101.center) to (104.center);
		\draw [style=dark blue solid] (107.center)
			 to (105.center)
			 to (106.center)
			 to cycle;
		\draw [style=black solid] (111.center) to (112);
		\draw [style=black solid, in=-90, out=90] (112) to (109.center);
		\draw [style=black solid, in=-90, out=90] (113) to (103.center);
		\draw [style=black solid] (114.center) to (113);
		\draw [style=black solid] (106.center) to (110);
		\draw [style=black solid] (110) to (108);
		\draw [style=black solid] (100) to (117.center);
		\draw [style=dark blue solid] (120.center)
			 to [bend left=45, looseness=1.25] (119.center)
			 to [bend left=45, looseness=1.25] cycle;
		\draw [style=black solid, in=90, out=-180] (122.center) to (125.center);
		\draw [style=black solid, in=180, out=-90, looseness=0.75] (125.center) to (117.center);
		\draw [style=black solid, in=-90, out=0] (117.center) to (124.center);
		\draw [style=black solid, in=0, out=90] (124.center) to (122.center);
		\draw [style=dashed dark blue] (132.center)
			 to [in=-165, out=90, looseness=1.25] (131.center)
			 to [in=0, out=15, looseness=1.25] (127.center);
		\draw [style=dashed dark blue] (132.center)
			 to [in=-15, out=90, looseness=1.25] (130.center)
			 to [in=-180, out=165, looseness=1.50] (128.center);
		\draw [style=black solid, in=-90, out=0, looseness=1.25] (86) to (133.center);
		\draw [style=black solid, in=270, out=90] (133.center) to (134.center);
		\draw [style=dashed dark blue] (93.center) to (95.center);
		\draw [style=dashed dark blue] (95.center) to (96.center);
		\draw [style=dashed dark blue] (96.center) to (94.center);
		\draw [style=dashed dark blue] (93.center) to (94.center);
		\draw [style=dashed dark blue] (88.center) to (90.center);
		\draw [style=dashed dark blue] (90.center) to (91.center);
		\draw [style=dashed dark blue] (91.center) to (89.center);
		\draw [style=dashed dark blue] (88.center) to (89.center);
	\end{pgfonlayer}
\end{tikzpicture}
}
\end{minipage}
\caption{The event that a patient has an infectious disease given one positive and one negative test result, as a discrete probabilistic program $\func{hasDisease}$ (left) and a hierarchical string diagram (right). 
The disease occurs in $10^{-4}$ of the population; the test has a 99\% true positive rate (TP) and a 2\% false positive rate (FP).}
\label{fig:code-and-diagram}
\end{figure}

\subsection{The Boolean probabilistic inference problem}

The main subject of this paper is probabilistic inference: given a DPP with no inputs and a single output, what is the probability that this DPP returns \emph{true}? We cannot hope to efficiently compute this exactly: there exist programs of length $O(n)$ whose probability of returning \emph{true} is $2^{-2^{n}}$, which requires exponentially many digits to express; see \Cref{fig:counterexample}. Instead, we look at the approximate problem:

\begin{problem}\label{problem:boolean-probabilistic-inference}
Given a discrete probabilistic program $\func{f}$ with no inputs and a single Boolean output, 
and a natural number $d$, 
compute the probability that $\func{f}$ returns \emph{true}, up to $d$ fractional binary digits.
\end{problem}

\Cref{problem:boolean-probabilistic-inference} is PSPACE-hard~\cite{holtzen2020scaling}, 
even when restricting to deterministic DPPs,
showing that this specific hardness result has little to do with probability and instead stems from the expressive power of function definitions. 
However, without function definitions, probabilistic inference remains NP-hard even when restricted to fault trees~\cite{lopuhaa2025fault} and Bayesian networks. 

Nevertheless, for Bayesian networks, inference is known to be \emph{fixed-parameter tractable} \cite{lauritzen1988local}, 
i.e.~there are inference algorithms with polynomial time complexity given that a certain parameter is bounded. The relevant parameter is the \emph{treewidth} of the so-called moral graph of the Bayesian network, which measures how far away the network is from being ``tree-like''. 
Similar results are known for fault trees~\cite{lopuhaa2025fault} and the related problem of weighted model counting~\cite{darwiche2002knowledge,darwiche2011sdd,dudek2020parallel}. 
However, such results do not exist for DPPs in general.

\begin{figure}[!t]
\centering
\vspace{-4mm}

\begin{subfigure}[b]{0.38\linewidth}
\centering
\begin{minipage}[c]{0.35\linewidth}
\raggedright
\vspace{3mm}
\rule{1.5\linewidth}{0.2pt} \vspace{2mm}
\scalebox{0.7}{
\(
\begin{aligned}[t]
&\quad\func{f}_0(\,) := \\
&\quad\qquad\lett x = \func{flip}(0.5);\\
&\quad\qquad\lett y = \func{flip}(0.5);\\
&\quad\qquad\lett z = x \land y;\\
&\quad\qquad z \\
\\
&\quad\cdots \\
\\
&\quad\func{f}_{n+1}(\,) := \\
&\quad\qquad\lett x = \func{f}_{n}(\,);\\
&\quad\qquad\lett y = \func{f}_{n}(\,);\\
&\quad\qquad\lett z = x \land y;\\
&\quad\qquad z \\
\end{aligned}
\)
}
\vspace{2mm}
\rule{1.5\linewidth}{0.2pt}
\end{minipage}\hfill 
\begin{minipage}[c]{0.55\linewidth}
\centering
\scalebox{0.45}{
\begin{tikzpicture}
	\begin{pgfonlayer}{nodelayer}
		\node [style=none] (0) at (-5.5, -0.5) {};
		\node [style=none] (1) at (-4.5, 1.25) {};
		\node [style=none] (2) at (-3.5, -0.5) {};
		\node [style=none] (3) at (-4.5, 1.25) {};
		\node [style=none] (4) at (-4, -0.5) {};
		\node [style=none] (5) at (-5, -0.5) {};
		\node [style=dot] (6) at (-5.25, -1.25) {};
		\node [style=dot] (7) at (-4.5, 1.75) {};
		\node [style=none] (8) at (-3.25, -2) {};
		\node [style=dot] (9) at (-3.75, -1.25) {};
		\node [style=none] (10) at (-4.5, 2.25) {};
		\node [style=none] (11) at (-5.75, -2) {};
		\node [style=none] (12) at (-6.25, -2) {};
		\node [style=none] (13) at (-5.25, -2) {};
		\node [style=none] (14) at (-3.75, -2) {};
		\node [style=none] (15) at (-2.75, -2) {};
		\node [style=none] (16) at (-3.75, -3.5) {};
		\node [style=none] (17) at (-2.75, -3.5) {};
		\node [style=none] (18) at (-6.25, -3.5) {};
		\node [style=none] (19) at (-5.25, -3.5) {};
		\node [style=none] (20) at (-5.5, -7.75) {};
		\node [style=none] (21) at (-4.5, -6) {};
		\node [style=none] (22) at (-3.5, -7.75) {};
		\node [style=none] (23) at (-4.5, -6) {};
		\node [style=none] (24) at (-4, -7.75) {};
		\node [style=none] (25) at (-5, -7.75) {};
		\node [style=dot] (26) at (-5.25, -8.5) {};
		\node [style=dot] (27) at (-4.5, -5.5) {};
		\node [style=none] (28) at (-3.25, -9.25) {};
		\node [style=dot] (29) at (-3.75, -8.5) {};
		\node [style=none] (30) at (-4.5, -5) {};
		\node [style=none] (31) at (-5.75, -9.25) {};
		\node [style=none] (32) at (-6.25, -9.25) {};
		\node [style=none] (33) at (-5.25, -9.25) {};
		\node [style=none] (34) at (-3.75, -9.25) {};
		\node [style=none] (35) at (-2.75, -9.25) {};
		\node [style=none] (36) at (-3.75, -10.75) {};
		\node [style=none] (37) at (-2.75, -10.75) {};
		\node [style=none] (38) at (-6.25, -10.75) {};
		\node [style=none] (39) at (-5.25, -10.75) {};
		\node [style=none] (41) at (-7.25, -4.5) {};
		\node [style=none] (42) at (-1.75, -4.5) {};
		\node [style=none] (43) at (-7.25, -11.5) {};
		\node [style=none] (44) at (-1.75, -11.5) {};
		\node [style=none] (45) at (-4.5, -4.5) {};
		\node [style=none] (46) at (-5.75, -3.5) {};
		\node [style=none] (47) at (-3.25, -3.5) {};
		\node [style=none] (48) at (-5.75, -10.75) {};
		\node [style=none] (49) at (-3.25, -10.75) {};
		\node [style=none] (50) at (-4.5, -12.25) {};
		\node [style=none] (51) at (-7.25, -12.25) {};
		\node [style=none] (52) at (-1.75, -12.25) {};
		\node [style=none] (53) at (-1.75, -13) {};
		\node [style=none] (54) at (-7.25, -13) {};
		\node [style=dot] (55) at (-4.5, -13.5) {};
		\node [style=dot] (56) at (-4.5, -14.25) {};
		\node [style=dot] (57) at (-4.5, -15) {};
		\node [style=none] (58) at (-7.25, -15.5) {};
		\node [style=none] (59) at (-7.25, -16.25) {};
		\node [style=none] (60) at (-1.75, -15.5) {};
		\node [style=none] (61) at (-1.75, -16.25) {};
		\node [style=none] (62) at (-7.25, -16.25) {};
		\node [style=none] (63) at (-5.75, -15.5) {};
		\node [style=none] (64) at (-3.25, -15.5) {};
		\node [style=none] (65) at (-4.5, -17.25) {};
		\node [style=none] (66) at (-5.5, -20.5) {};
		\node [style=none] (67) at (-4.5, -18.75) {};
		\node [style=none] (68) at (-3.5, -20.5) {};
		\node [style=none] (69) at (-4.5, -18.75) {};
		\node [style=none] (70) at (-4, -20.5) {};
		\node [style=none] (71) at (-5, -20.5) {};
		\node [style=dot] (72) at (-5.25, -21.25) {};
		\node [style=dot] (73) at (-4.5, -18.25) {};
		\node [style=light blue dot] (74) at (-3.25, -22.5) {0.5};
		\node [style=dot] (75) at (-3.75, -21.25) {};
		\node [style=none] (76) at (-4.5, -17.5) {};
		\node [style=light blue dot] (77) at (-5.75, -22.5) {0.5};
		\node [style=none] (86) at (-7.25, -17.25) {};
		\node [style=none] (87) at (-1.75, -17.25) {};
		\node [style=none] (88) at (-7.25, -23.75) {};
		\node [style=none] (89) at (-1.75, -23.75) {};
		\node [style=none] (90) at (-4.5, -17.25) {};
		\node [style=none] (91) at (-3, -8.25) {{\Large $y$}};
		\node [style=none] (92) at (-5.75, -8.25) {{\Large $x$}};
		\node [style=none] (93) at (-3, -1) {};
		\node [style=none] (94) at (-6, -1) {};
		\node [style=none] (95) at (-3, -21) {{\Large $y$}};
		\node [style=none] (96) at (-6, -21) {{\Large $x$}};
		\node [style=none] (97) at (-3.75, -18.25) {{\Large $z$}};
		\node [style=none] (98) at (-3.75, -5.5) {{\Large $z$}};
		\node [style=none] (99) at (-3.75, 1.75) {{\Large $z$}};
		\node [style=none] (100) at (-6.5, -5.25) {$\func{f}_{n}$};
		\node [style=none] (101) at (-6.25, -13) {$\func{f}_{n-1}$};
		\node [style=none] (102) at (-6.5, -18) {$\func{f}_0$};
	\end{pgfonlayer}
	\begin{pgfonlayer}{edgelayer}
		\draw [style=dark blue solid] (0.center)
			 to [in=-180, out=90] (1.center)
			 to [in=90, out=0] (2.center)
			 to cycle;
		\draw [style=dark blue solid] (3.center) to (1.center);
		\draw [style=dark blue solid, in=60, out=-90] (5.center) to (6);
		\draw [style=black solid] (7) to (3.center);
		\draw [style=black solid, in=90, out=-60] (9) to (8.center);
		\draw [style=black solid, in=120, out=-90] (4.center) to (9);
		\draw [style=black solid] (7) to (10.center);
		\draw [style=black solid, in=90, out=-120, looseness=1.25] (6) to (11.center);
		\draw [style=red dashed] (12.center) to (18.center);
		\draw [style=red dashed] (12.center) to (13.center);
		\draw [style=red dashed] (13.center) to (19.center);
		\draw [style=red dashed] (18.center) to (19.center);
		\draw [style=red dashed] (14.center) to (16.center);
		\draw [style=red dashed] (14.center) to (15.center);
		\draw [style=red dashed] (15.center) to (17.center);
		\draw [style=red dashed] (17.center) to (16.center);
		\draw [style=dark blue solid] (22.center)
			 to (20.center)
			 to [in=-180, out=90] (21.center)
			 to [in=90, out=0] cycle;
		\draw [style=dark blue solid] (23.center) to (21.center);
		\draw [style=dark blue solid, in=60, out=-90] (25.center) to (26);
		\draw [style=black solid] (27) to (23.center);
		\draw [style=black solid, in=90, out=-60] (29) to (28.center);
		\draw [style=black solid, in=120, out=-90] (24.center) to (29);
		\draw [style=black solid] (27) to (30.center);
		\draw [style=black solid, in=90, out=-120, looseness=1.25] (26) to (31.center);
		\draw [style=brightblue dashed] (32.center) to (38.center);
		\draw [style=brightblue dashed] (32.center) to (33.center);
		\draw [style=brightblue dashed] (33.center) to (39.center);
		\draw [style=brightblue dashed] (38.center) to (39.center);
		\draw [style=brightblue dashed] (34.center) to (36.center);
		\draw [style=brightblue dashed] (34.center) to (35.center);
		\draw [style=brightblue dashed] (35.center) to (37.center);
		\draw [style=brightblue dashed] (37.center) to (36.center);
		\draw [style=red dashed] (41.center) to (43.center);
		\draw [style=red dashed] (41.center) to (42.center);
		\draw [style=red dashed] (42.center) to (44.center);
		\draw [style=red dashed] (44.center) to (43.center);
		\draw [style=red dashed, in=-90, out=105, looseness=1.25] (45.center) to (46.center);
		\draw [style=red dashed, in=-90, out=75, looseness=1.25] (45.center) to (47.center);
		\draw [style=brightblue dashed, in=-90, out=105, looseness=1.25] (50.center) to (48.center);
		\draw [style=brightblue dashed, in=-90, out=75, looseness=1.25] (50.center) to (49.center);
		\draw [style=brightblue dashed] (51.center) to (52.center);
		\draw [style=brightblue dashed] (51.center) to (54.center);
		\draw [style=brightblue dashed] (52.center) to (53.center);
		\draw [style=orange dashed] (58.center) to (62.center);
		\draw [style=orange dashed] (62.center) to (61.center);
		\draw [style=orange dashed] (61.center) to (60.center);
		\draw [style=dark blue solid] (68.center)
			 to (66.center)
			 to [in=-180, out=90] (67.center)
			 to [in=90, out=0] cycle;
		\draw [style=dark blue solid] (69.center) to (67.center);
		\draw [style=dark blue solid, in=60, out=-90] (71.center) to (72);
		\draw [style=black solid] (73) to (69.center);
		\draw [style=black solid, in=90, out=-60] (75) to (74);
		\draw [style=black solid, in=120, out=-90] (70.center) to (75);
		\draw [style=black solid] (73) to (76.center);
		\draw [style=black solid, in=90, out=-120, looseness=1.25] (72) to (77);
		\draw [style=purple dashed] (86.center) to (88.center);
		\draw [style=purple dashed] (86.center) to (87.center);
		\draw [style=purple dashed] (87.center) to (89.center);
		\draw [style=purple dashed] (89.center) to (88.center);
		\draw [style=purple dashed, in=105, out=-90, looseness=1.25] (63.center) to (90.center);
		\draw [style=purple dashed, in=-90, out=75, looseness=1.25] (90.center) to (64.center);
	\end{pgfonlayer}
\end{tikzpicture}
}
\end{minipage}

\vspace{5mm}
\centering(a)
\phantomcaption\label{fig:counterexample}
\end{subfigure}\begin{subfigure}[b]{0.59\linewidth}
\centering
\scalebox{0.6}{
\input{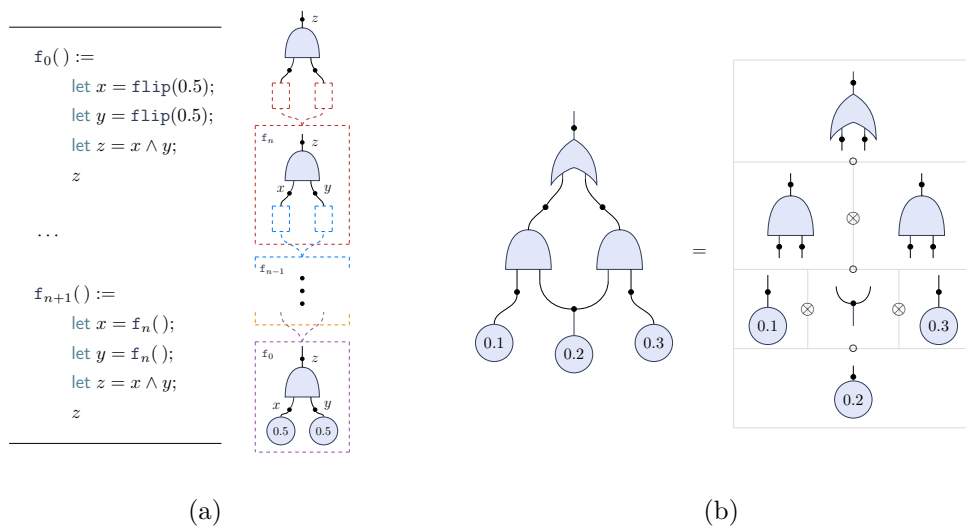}
}

\vspace{9mm} 
\centering(b)
\phantomcaption\label{fig:decomposition-example-sub}
\end{subfigure}

\caption{(a) Family $(\func{f}_n)$ of DPPs. 
The probability that $\func{f}_n$ returns \emph{true} is doubly exponentially small in $n$.
The weighted-model-counting-based inference algorithm of \texttt{Dice} requires $2^n$ 
fresh variables to compile $\func{f}_n$ to a weighted Boolean formula~\cite[Section~4.3]{holtzen2020scaling},
causing exponential blowup before any weighted model count is performed.
(b) 
Decomposition of a dot diagram via parallel ($\otimes$) and sequential ($\circ$) composition, equivalently the term $\lor \circ (\land \otimes \land) \circ (\func{flip}(0.1) \otimes \func{copy} \otimes \func{flip}(0.3)) \circ \func{flip}(0.2)$.
The operations $\circ$ and $\otimes$ correspond to matrix multiplication and Kronecker product; see \Cref{sec:semantics}.}
\label{fig:decomposition-example}
\end{figure}

\subsection{Contributions}

Our main contribution is a fixed-parameter tractable algorithm for inference in DPPs. 
The key parameter is the maximal treewidth $k$ of the so-called \emph{primal graph} of the functions in the DPP.

\begin{mainresult}[Theorem~\ref{thm:fixed-parameter-tractable-inference} paraphrased]
There exists an algorithm that computes the probability that a DPP returns \emph{true} up to $d$ digits in $O(2^{O(k)}d^2 n^3 \log^3(2p_{\mathrm{acc}}^{-1}))$ time,
where $n$ measures the overall size of the DPP,\footnote{
More precisely, $n = LMN$, 
where $L$ is (one plus) the maximum number of distinct functions called within any function,
$M$ is the number of functions in the DPP, 
and $N$ is (one plus) the maximum number of assignments in any function.  
In particular, if $l$ is the number of lines in the DPP, then $n$ is in $O(l^3)$, 
and hence, our algorithm takes polynomial time in the number of lines of the given DPP, 
for fixed $k$, $d$ and $p_{\mathrm{acc}}$. 
}
and $p_{\mathrm{acc}}$ is the \emph{acceptance probability}, 
the probability that all $\func{observe}$d values are true.\footnote{
The dependency on $p_{\mathrm{acc}}$ can be removed, for example,
if the depth of the call graph of the DPP is bounded; see \Cref{ssec:inference}.
}
\end{mainresult}

\enlargethispage{\baselineskip}
Existing methods do not achieve these performance guarantees: the state-of-the-art inference algorithm of \texttt{Dice} \cite{holtzen2020scaling} handles DPPs by compiling them into a weighted propositional formula and applying weighted model counting to the result.
Through function calls this propositional formula can grow exponentially in size;
see \Cref{fig:counterexample}. 
By contrast, we essentially compute the semantics of each called function separately and then combine the results.  
Moreover, even when restricting to DPPs corresponding to Bayesian networks or fault trees, 
our algorithm does not simply reduce to established methods.
Instead, we rely on \emph{string diagram algebraisation}.

\subsection{Idea of proof}

Our algorithm operates on the string diagram -- or more precisely, \emph{dot diagram} \cite{kissinger2014finite} -- representing the abstract syntax of a DPP. 
Dot diagrams can be decomposed into a \emph{term} consisting of atomic parts via parallel and sequential decomposition; see Figure~\ref{fig:decomposition-example}. 
Such a term is called an \emph{algebraisation}, 
and there are generally many algebraisations of a dot diagram. 
To perform inference, each sub-term is translated into a matrix with non-negative rational entries,
and $\circ$ and $\otimes$ translate to matrix composition and Kronecker product, respectively. 
A sub-term with $p$ inputs and $q$ outputs corresponds to a $2^q \times 2^p$-matrix;
thus we need algebraisations with low \emph{width}, where the width is the maximum number of inputs and outputs, counted together, over all sub-terms.

\begin{figure}[!t]
\scalebox{0.53}{
\input{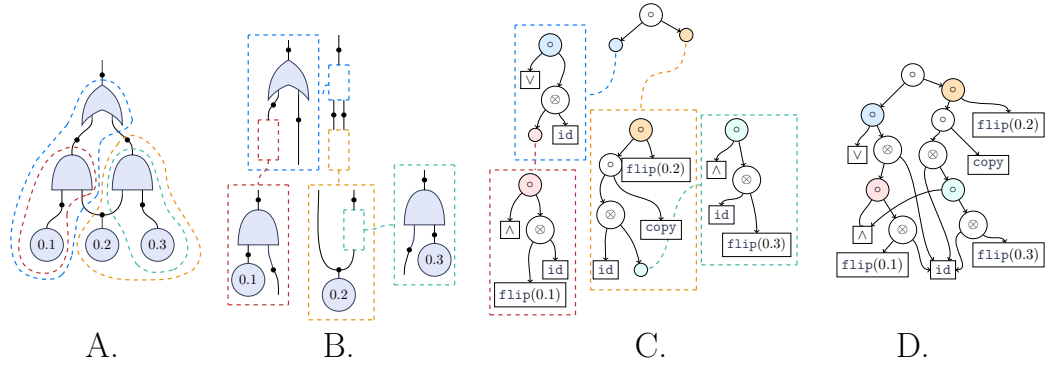}
}
\caption{The algorithm of Theorem~\ref{thm:algebraisation}:
A. Compute a branch decomposition of the dot diagram.
B. Recursively extract sub-diagrams to obtain a hierarchical dot diagram.
C. Algebrise each sub-diagram separately as a DAG.
D. Compose DAGs via substitution.
The resulting term has a \emph{width} of 4 and a DAG-size of 17.}
\label{fig:algorithm}
\end{figure}

Unfortunately, existing approaches to string diagram algebraisation are not designed to minimise algebraisation width:  
the Frobenius decomposition of~\cite{wilson2023data} is aimed towards getting a decomposition \emph{fast}.
The algorithm used by the Python toolkit \emph{DisCoPy} \cite{de2021discopy,toumi2023discopy} to compute the semantics of a string diagram also essentially uses an algebraisation algorithm,
but only implicitly, and without optimising for low width. 
A generalisation of graph-theoretic width notions to monoidal categories, \emph{monoidal width},
has been studied in~\cite{di2023monoidal}, 
but without exploiting this connection for algorithmic purposes.\footnote{
A more detailed discussion of the similarities and differences between the width of 
an algebraisation and monoidal width is given in \Cref{sec:relation-to-monoidal-width}. 
Summarised briefly, our notion of algebraisation can be understood in terms of the \emph{monoidal decompositions}
of~\cite{di2023monoidal}, but the respective notions of width differ, 
making it difficult to establish a deeper relation. 
}
To address these shortcomings we develop an algorithm that finds algebraisations with low width. Specifically, we prove:

\begin{mainresult}[Theorem~\ref{thm:algebraisation} paraphrased]
There is an algorithm that takes as input a string diagram of size $n$ over a fixed finite signature, with a bounded number of inputs and outputs, whose primal graph has treewidth $k$, and outputs an algebraisation of width $O(k)$ and size $O(nk^2)$, in $O(2^{O(k)}n)$ time.
\end{mainresult}
The idea of the proof is to take an existing algorithm for finding a so-called \emph{tree decomposition} of the primal graph, convert this decomposition to a \emph{branch decomposition}, and translate the result to an algebraisation (see Figure~\ref{fig:algorithm}).  
The treewidth of a graph is the minimal width of its tree decompositions. 
Finding the optimal tree decomposition is NP-hard~\cite{arnborg1987complexity}; 
instead, we employ approximation algorithms that find a sufficiently good decomposition in $O(2^{O(k)}n)$ time~\cite{korhonen2023single}.

\subsection{Further applications of algebraisation}

Since Theorem~\ref{thm:algebraisation} works on the level of string diagrams, it has applications beyond DPPs. 
By taking values in matrices over arbitrary semirings rather than stochastic matrices, our results can be applied to \emph{algebraic model counting} \cite{kimmig2017algebraic}, which in turn covers Boolean satisfiability, weighted model counting, and shortest path problems. 
Our results also apply to other problems that can be phrased in terms of string diagrams, 
such as evaluating conjunctive queries on relational databases~\cite{chekuri2000conjunctive}, 
or quantitative cybersecurity risk assessment via attack trees~\cite{peterseim2025unified}. 
In a different direction, the term representation that our algorithm computes 
can also be used for visualising string diagrams in a recursive manner~\cite{rubio2024rendering}.

\section{Syntax: dot diagrams}\label{sec:setup}

We first define the syntax of our DPPs, in terms of \emph{dot diagrams} \cite{kissinger2014finite}. 
On this level of generality, these are neither necessarily probabilistic nor Boolean, but instead defined over an arbitrary \emph{monoidal signature}.

\emph{Notations for lists.} 
If $X$ is a set, $X^*$ denotes the set of lists (i.e.~finite, ordered sequences) of elements of $X$.
We write a list as a tuple $\mathbf{x} = (x_1, \dots, x_n)$, write $|\mathbf{x}|$ for its length $n$, and write $\mathbf{x}\mathbf{y}$ for the concatenation of two lists.
The set of all its elements is $\mathsf{set}(\mathbf{x})$.
For a map $\phi\colon X \to Y$, we write $\phi_*\colon X^* \to Y^*$ for its elementwise extension to lists, the \emph{pushforward} of lists.

\subsection{Monoidal signatures}

\emph{Monoidal signatures} record a number of \emph{sorts} and \emph{function symbols}. 
The difference to ordinary (multi-sorted) algebraic signatures is that function symbols may have multiple outputs.

\begin{definition}
A \emph{monoidal signature} is a pair 
\[\Sigma = \left(\Sigma_0, \Sigma(\,\cdot\,,\,\cdot\,)\right)\] 
consisting of 
\begin{enumerate}
\item a set $\Sigma_0$, the set of \emph{sorts};
\item for all $\mathbf{S}, \mathbf{T}\in \Sigma_0^*$, a set $\Sigma(\mathbf{S}, \mathbf{T})$.
\end{enumerate}
The elements $f$ of $\Sigma(\mathbf{S}, \mathbf{T})$ are called \emph{function symbols} with \emph{domain} $\mathbf{S}$ and \emph{codomain} $\mathbf{T}$; this is denoted $f\colon \mathbf{S} \rightarrow \mathbf{T}$. Irrespective of domains and codomains we also write $f \in \Sigma$. 
A \emph{morphism of monoidal signatures} $F\colon\Sigma \to \Sigma'$ is a map $F\colon \Sigma_0 \to \Sigma'_0$, together with a family of maps,
\[\Sigma(\mathbf{S},\mathbf{T}) \to \Sigma'(F_*(\mathbf{S}), F_*(\mathbf{T})),\]  
each component of which we also denote by $F$.
\end{definition}

\begin{example}
The \emph{signature of discrete probabilistic programs} is the signature $\mathsf{BoolStochSig}$, where
\begin{enumerate}
\item $\mathsf{BoolStochSig}$ has a single sort, denoted by $\mathbb{B}$, so $\mathsf{BoolStochSig}_0 = \{\mathbb{B}\}$.
\item The function symbols of $\mathsf{BoolStochSig}$ are $(p\in [0,1]\cap \mathbb{Q})$:
\begin{align*}
\land\colon\: (\mathbb{B}, \mathbb{B}) &\to (\mathbb{B}), &
\lor\colon\: (\mathbb{B}, \mathbb{B}) &\to (\mathbb{B}), \\
\lnot\,\colon\: (\mathbb{B}) &\to (\mathbb{B}), &\func{observe}\colon (\mathbb{B}) &\to (\,),\\
\func{flip}(p)\colon\: (\,) &\to (\mathbb{B}). 
\end{align*}
\end{enumerate}
\end{example}

\subsection{Variable sets and assignments}

Having defined our set of function symbols, we now define variables and variable assignments.

\begin{definition}[Variable set]
Let $\Sigma$ be a monoidal signature. 
A \emph{variable set} over $\Sigma$ is a pair $(X,\mathsf{sort}_X)$ consisting of a finite set $X$ and $\mathsf{sort}_X\colon X \rightarrow \Sigma_0$.  
We write $\mathsf{sort}$ instead of $\mathsf{sort}_X$ when $X$ is clear, and we write $X$ instead of $(X,\mathsf{sort}_X)$ for the variable set. A \emph{morphism of variable sets} $X\to Y$ is a map $\phi\colon X\to Y$ such that $\mathsf{sort}_Y \circ \phi = \mathsf{sort}_X$.
\end{definition}

\begin{definition}[Assignment]\sloppy
Let $\Sigma$ be a monoidal signature and $X$ be a variable set over $\Sigma$. 
An \emph{assignment} over $X$ is a triple $(\mathbf{y}, f, \mathbf{x})$, where $\mathbf{x},\mathbf{y} \in X^*$ and $f \in \Sigma(\mathsf{sort}_*(\mathbf{x}),\mathsf{sort}_*(\mathbf{y}))$. The lists $\mathbf{x}$ and $\mathbf{y}$ are the \emph{input} and \emph{output} variables, respectively, and we will write  
\[\lett y_1, \dots, y_n = f(x_1, \dots, x_m)\]
to denote the assignment $ ((y_1, \dots, y_n), f, (x_1, \dots, x_m))$. This will be rendered as ``$f(x_1,\ldots,x_m)$'' when $n = 0$, and ``$\lett y_1,\ldots,y_n = f$'' when $m = 0$.
Moreover, we write $\ass{\Sigma}{X}$ for the set of all assignments over the variable set $X$.

The pushforward of a morphism of variable sets $\phi\colon X\to Y$ is the map
\[ 
\ass{\Sigma}{\phi}\colon\: \ass{\Sigma}{X} \to \ass{\Sigma}{Y}, \quad (\mathbf{y}, f, \mathbf{x}) \mapsto (\phi_*(\mathbf{y}), f, \phi_*(\mathbf{x})).
\]
\end{definition}

\begin{example}
Let $\Sigma := \mathsf{BoolStochSig}$,
let $X := \{x, y, z\}$ with $\mathsf{sort}_X(\,\cdot\,) := \mathbb{B}$.
Then $((z),\wedge,(x,y))$ is an assignment over $X$, representing the ``line of code'' $\lett z = x \wedge y$.
\end{example}

Note that we do not make any assumptions about 
how often a variable may appear on either side of the assignment.
We will see how a variable occurring multiple times 
on the left-hand side of the assignment can be interpreted,  
once we have defined the notion of a \emph{dot diagram} over $\Sigma$.

\subsection{Dot diagrams}

After function symbols, variables and assignments, we can now define programs, 
which for us are \emph{dot diagrams}. 
These were introduced in~\cite{kissinger2014finite} as hypergraphs whose vertices are called \emph{boxes} and whose hyperedges are called \emph{dots}. 
Our definition looks more directly like one of ``straight-line programs'' but is equivalent; 
the original definition corresponds to passing to the ``data flow graph''.

\begin{definition}[Dot diagram]
Let $\Sigma$ be a monoidal signature. 
A \emph{dot diagram} $\func{f}$ over $\Sigma$, or \emph{$\Sigma$-diagram}, is a tuple $(X_{\func{f}}, A_{\func{f}}, \mathsf{in}_{\func{f}}, \mathsf{out}_{\func{f}})$ where 
\begin{enumerate}
\item $X_{\func{f}}$ is a variable set;
\item $A_{\func{f}} \in \ass{\Sigma}{X_{\func{f}}}^*$ is a list of assignments; and
\item $\mathsf{in}_{\func{f}}, \mathsf{out}_{\func{f}} \in X_{\func{f}}^*$ are lists of \emph{input} and \emph{output} variables.
\end{enumerate}
A variable $x\in X_{\func{f}}$ is \emph{isolated} if it occurs neither in $\mathsf{in}_{\func{f}}$, nor in $\mathsf{out}_{\func{f}}$, nor in any assignment from $A_{\func{f}}$, i.e.~there is no $(\mathbf{y}, g, \mathbf{x})\in A_{\func{f}}$
such that $x\in \mathsf{set}(\mathbf{y})$ or $x\in \mathsf{set}(\mathbf{x})$. Isolated variables are allowed.
In the following, we will write
\begin{align*}
&\func{f}(x_1, \dots, x_n) :=\\
&\qquad a_1; \\
&\qquad \ldots\\
&\qquad a_k; \\
&\qquad y_1, \dots, y_m
\end{align*}
instead of $\func{f} := (X_{\func{f}}, (a_1, \dots, a_k), (x_1, \dots, x_n), (y_1, \dots, y_m))$ to define a dot diagram $\func{f}$. Dot diagrams over $\Sigma$ form a monoidal signature, which we denote by $\mathrm{Dot}(\Sigma)$, in which $\func{f}$ is a function symbol $\mathsf{sort}_*(\mathsf{in}_{\func{f}}) \to \mathsf{sort}_*(\mathsf{out}_{\func{f}})$.

Two dot diagrams $\func{f}, \mathsf{g} \in \mathrm{Dot}(\Sigma)$ with $A_{\func{f}} = (a_1, \dots, a_k)$ and $A_{\mathsf{g}} = (b_1, \dots, b_l)$ are \emph{equivalent}, written $\func{f} \equiv \mathsf{g}$, if there exist an isomorphism of variable sets $\phi\colon X_{\func{f}} \to X_\mathsf{g}$ and a bijection $\pi\colon \{1, \dots, k\} \to \{1, \dots, l\}$ such that
$\phi_*(\mathsf{in}_{\func{f}}) = \mathsf{in}_\mathsf{g}$,  $\phi_*(\mathsf{out}_{\func{f}}) = \mathsf{out}_\mathsf{g}$,
and $\ass{\Sigma}{\phi}(a_i) = b_{\pi(i)}$ for all $i\in\{1, \dots, k\}$.
In other words, $\mathsf{g}$ arises from $\func{f}$ by renaming variables and reordering assignments.
We write $\underline{\func{f}}$ for the equivalence class of $\func{f}$, and $\underline{\mathrm{Dot}}(\Sigma)$ for the set of $\Sigma$-diagrams modulo equivalence.
\end{definition}

\begin{example}\label{example:dot-diagram}
Let $\Sigma$ be the monoidal signature with  $\Sigma_0 = \{\mathbb{S}\}$ and two function symbols $\func{f}\colon (\mathbb{S}, \mathbb{S}) \to (\mathbb{S})$ and $\func{g}\colon (\mathbb{S}) \to (\mathbb{S}, \mathbb{S})$.
Let $X := \{x, y, z, v, w\}$ be the variable set over $\Sigma$ with $\mathsf{sort}_X(\,\cdot\,) := \mathbb{S}$.
Then below (left) is an example of a dot diagram $\func{h}$ over $\Sigma$ with variable set $X$, two input variables $x,y$, one output variable $w$, and three assignments. 
On the right is its graphical representation where every assignment is a \emph{box} and every variable is a \emph{dot}.\newline
\hspace{1cm}
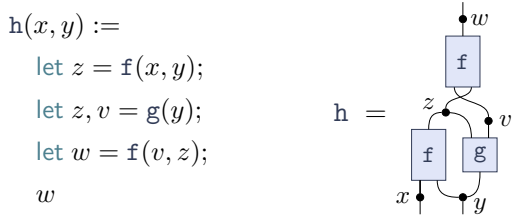
\begin{minipage}[t]{0.35\linewidth}
\raggedright
\setlength{\abovedisplayskip}{0pt}\setlength{\abovedisplayshortskip}{0pt}\begin{align*}
&\func{h}(x, y) := \\
&\quad \lett z = \func{f}(x, y); \\
&\quad \lett z, v = \func{g}(y); \\
&\quad \lett w = \func{f}(v, z); \\
&\quad w
\end{align*}
\end{minipage}\hspace{-0.06\linewidth}\begin{minipage}[t]{0.45\linewidth}
\raggedleft
\setlength{\abovedisplayskip}{0pt}\setlength{\abovedisplayshortskip}{0pt}\begin{equation*}
\scalebox{0.9}{
\begin{tikzpicture}
	\begin{pgfonlayer}{nodelayer}
		\node [style=none] (0) at (0, 0) {};
		\node [style=none] (1) at (0, -1.5) {};
		\node [style=none] (2) at (1, -1.5) {};
		\node [style=none] (3) at (1, 0) {};
		\node [style=none] (4) at (0.5, -0.75) {$\func{f}$};
		\node [style=dot] (5) at (0.25, -2) {};
		\node [style=dot] (6) at (1, 0.5) {};
		\node [style=dot] (7) at (1.5, -2) {};
		\node [style=none] (8) at (0.25, -1.5) {};
		\node [style=none] (9) at (0.75, -1.5) {};
		\node [style=none] (10) at (0.5, 0) {};
		\node [style=none] (11) at (0.25, -2.5) {};
		\node [style=none] (12) at (1.5, -2.5) {};
		\node [style=none] (13) at (1.75, 1.25) {};
		\node [style=none] (14) at (1.5, -0.25) {};
		\node [style=none] (15) at (1.5, -1.25) {};
		\node [style=none] (16) at (2.5, -1.25) {};
		\node [style=none] (17) at (2.5, -0.25) {};
		\node [style=none] (18) at (2, -0.75) {$\func{g}$};
		\node [style=dot] (20) at (2.25, 0.25) {};
		\node [style=none] (22) at (2, -1.25) {};
		\node [style=none] (23) at (1, 0.5) {};
		\node [style=none] (24) at (2.25, -0.25) {};
		\node [style=none] (25) at (1.5, -2) {};
		\node [style=none] (26) at (1.75, -0.25) {};
		\node [style=none] (27) at (1.25, 1.25) {};
		\node [style=none] (28) at (1, 2.75) {};
		\node [style=none] (29) at (1, 1.25) {};
		\node [style=none] (30) at (2, 1.25) {};
		\node [style=none] (31) at (2, 2.75) {};
		\node [style=none] (32) at (1.5, 2) {$\func{f}$};
		\node [style=dot] (34) at (1.5, 3.25) {};
		\node [style=none] (36) at (1.25, 1.25) {};
		\node [style=none] (37) at (1.75, 1.25) {};
		\node [style=none] (38) at (1.5, 2.75) {};
		\node [style=none] (39) at (1.5, 3.75) {};
		\node [style=none] (40) at (-0.25, -2) {{\large $x$}};
		\node [style=none] (41) at (2, -2.25) {{\large $y$}};
		\node [style=none] (42) at (2.75, 0.25) {{\large $v$}};
		\node [style=none] (43) at (0.5, 0.75) {{\large $z$}};
		\node [style=none] (44) at (2, 3.25) {{\large $w$}};
		\node [style=none] (45) at (-1.5, 0.5) {{\Large $\func{h}\;=$}};
	\end{pgfonlayer}
	\begin{pgfonlayer}{edgelayer}
		\draw [style=dark blue solid] (1.center)
			 to (0.center)
			 to (3.center)
			 to (2.center)
			 to cycle;
		\draw [style=black solid] (5) to (8.center);
		\draw [style=black solid, in=-90, out=-180, looseness=1.25] (7) to (9.center);
		\draw [style=black solid] (11.center) to (8.center);
		\draw [style=black solid] (5) to (8.center);
		\draw [style=black solid, in=-180, out=90, looseness=1.25] (10.center) to (6);
		\draw [style=black solid, in=-90, out=90] (6) to (13.center);
		\draw [style=black solid] (12.center) to (7);
		\draw [style=dark blue solid] (15.center)
			 to (14.center)
			 to (17.center)
			 to (16.center)
			 to cycle;
		\draw [style=black solid] (24.center) to (20);
		\draw [style=black solid, in=-90, out=0, looseness=1.25] (25.center) to (22.center);
		\draw [style=black solid] (24.center) to (20);
		\draw [style=black solid, in=-90, out=90] (20) to (27.center);
		\draw [style=black solid, in=0, out=90, looseness=1.25] (26.center) to (23.center);
		\draw [style=dark blue solid] (29.center)
			 to (28.center)
			 to (31.center)
			 to (30.center)
			 to cycle;
		\draw [style=black solid] (38.center) to (34);
		\draw [style=black solid] (34) to (39.center);
	\end{pgfonlayer}
\end{tikzpicture}
}
\end{equation*}
\end{minipage}\vspace{2mm}
\hfil

\end{example}

\begin{example}
Up to ``syntactic sugar'', the function $\func{test}$ from \Cref{fig:code-and-diagram} is a $\mathsf{BoolStochSig}$-diagram. 
Its set of variables is $X_{\func{test}} = \{z,\bar{z},t_{\top|\top},t_{\top|\bot},t_{\mathsf{TP}},t_{\mathsf{FP}},t\}$ and $\mathsf{sort}(\,\cdot\,) := \mathbb{B}$. To ``de-sugar'' it into a dot diagram, 
we need to write the assignment $\lett t_{\mathsf{FP}} = \neg z \wedge t_{\top|\bot}$ as $\lett \bar{z} = \neg z; \lett t_{\mathsf{FP}} = \bar{z} \wedge t_{\top|\bot}$.
The resulting dot diagram has one input variable $z$ and one output variable $t$.   
\end{example}

On the other hand, 
the DPP $\func{hasDisease}$ from \Cref{fig:code-and-diagram} is not simply a dot diagram, since it consists of two functions where one is called by the other.

\subsection{Hierarchical dot diagrams}

Hierarchical dot diagrams allow for ``function definitions'':
the function symbols appearing within any assignments may themselves 
be defined as a dot diagram. 
Hierarchical dot diagrams are our generalisation of DPPs.

\begin{definition}\label{definition:hierarchical-dot-diagram}
Let $\Sigma$ be a monoidal signature. 
A \emph{hierarchical dot diagram} over $\Sigma$ (short: hierarchical $\Sigma$-diagram)
is a tuple,
\[ {\func{f}} = (V({\func{f}}), E({\func{f}}), R_{{\func{f}}}, (\Sigma_v)_{v\in V({\func{f}})}, ({\func{f}}_{v})_{v\in V({\func{f}})}, (\mathrm{defn}_{\func{f}, v})_{v\in V({\func{f}})}), \]
such that  
\begin{bracketenumerate}
\item $(V({\func{f}}), E({\func{f}}))$ is a connected, rooted directed acyclic graph with root $R_{{\func{f}}}$, called the \emph{call graph} of ${\func{f}}$.
\item Each $\Sigma_v$ is a monoidal signature with the same set of sorts $\Sigma_0$ as $\Sigma$, and ${\func{f}}_{v}$ is a dot diagram over $\Sigma \sqcup_{\Sigma_0} \Sigma_v$, i.e. the signature whose set of sorts is $\Sigma_0$ and whose set of function symbols is given by the disjoint union of those of $\Sigma$ and $\Sigma_v$. 
\item $\mathrm{defn}_{\func{f}, v}\colon \Sigma_v \to \mathrm{ch}(v)$ is a bijection from the set of function symbols in $\Sigma_v$ to the set of children of $v$, such that the map $g \mapsto {\func{f}}_{\mathrm{defn}_{\func{f}, v}(g)}$ is a morphism of signatures $\Sigma_v \to \bigsqcup_{u \in \mathrm{ch}(v)} \mathrm{Dot}(\Sigma \sqcup_{\Sigma_0} \Sigma_u)$.
\item Every function symbol of $\Sigma_v$ occurs in some assignment of ${\func{f}}_v$.
\end{bracketenumerate}
A $\Sigma$-diagram $\mathsf{g}$ \emph{appears in} a hierarchical $\Sigma$-diagram ${\func{f}}$
if there exists some $v\in V({\func{f}})$ such that ${\func{f}}_v = \mathsf{g}$.
\end{definition}

Condition (4) is only necessary for more convenient complexity bounds;
it ensures that hierarchical dot diagrams do not contain ``unused'' function symbols.
In particular, the number of children of any vertex $v\in V({\func{f}})$ is at most $|A_{{\func{f}}_v}|$.

\begin{example}\label{example:hierarchical-dot-diagram}    
Let $\Sigma$ and $\func{h}$ be given as in \Cref{example:dot-diagram}.
\Cref{fig:hierarchical-dot-diagram-example} shows an example $\func{k}$
of a hierarchical dot diagram over $\Sigma$ and its call graph.
The vertices of the call graph are $V(\func{k}) = \{1, 2, 3, 4\}$ and its root is $R_{\func{k}} = 1$. 
The signatures at the vertices are 
\begin{align*}
\Sigma_1 &= \{\func{i}\colon (\mathbb{S}) \to (\mathbb{S}),\ \func{j}\colon (\mathbb{S}) \to (\mathbb{S})\}, \\
\Sigma_2 = \Sigma_3 &= \{\func{h}\colon (\mathbb{S}, \mathbb{S}) \to (\mathbb{S})\}, \quad\Sigma_4 = \varnothing.
\end{align*}
The dot diagram at the vertex $4$ is $\func{k}_4 = \func{h}$;
the dot diagrams at the other vertices are given as shown in \Cref{fig:hierarchical-dot-diagram-example}.
Finally, the functions $\mathrm{defn}_{\func{k}, v}$ are given by
\begin{align*}
\mathrm{defn}_{\func{k}, 1}(\func{i}) &= 2, 
&&\mathrm{defn}_{\func{k}, 1}(\func{j}) = 3, \\
\mathrm{defn}_{\func{k}, 2}(\func{h}) &= 4, 
&&\mathrm{defn}_{\func{k}, 3}(\func{h}) = 4.  
\end{align*}
Intuitively, the vertices in the call graph represent subroutines in $\func{k}$, 
and putting $\mathrm{defn}_{\func{k}, 2}(\func{h}) = 4$ means that the
function symbol $\func{h}$ called in subroutine $2$ is defined by subroutine $4$.
\end{example}

\begin{figure}
\vspace{-5mm}
\centering
\scalebox{0.65}{
\input{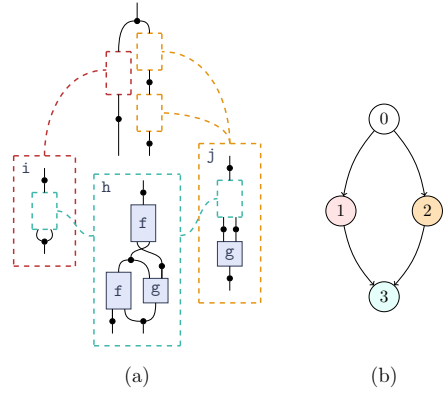}
}
\caption{(a) Example hierarchical $\Sigma$-diagram $\func{k}$ described in \Cref{example:hierarchical-dot-diagram}; (b) the call graph of $\func{k}$.}
\label{fig:hierarchical-dot-diagram-example}
\end{figure}

The discrete probabilistic programs from the introduction,
such as the program $\func{hasDisease}$ from \Cref{fig:code-and-diagram},
can also be understood as hierarchical dot diagrams,
and this is how we \emph{define} discrete probabilistic programs.
However, not every hierarchical dot diagram over $\mathsf{BoolStochSig}$ is a program in the usual sense:
in an ordinary straight-line program, a variable is either an input, or it is defined exactly once,
by an assignment that only uses variables defined before.
Dot diagrams with this property are called \emph{gs-monoidal},
after the kind of category in which they receive their semantics (see \Cref{remark:gs-monoidal}).

\begin{definition}[gs-monoidal dot diagram]\label{definition:gs-monoidal}
Let $\Sigma$ be a monoidal signature. A $\Sigma$-diagram $\func{f}$ is \emph{gs-monoidal} if
\begin{bracketenumerate}
\item every variable $x\in X_{\func{f}}$ occurs exactly once in total either in the list $\mathsf{in}_{\func{f}}$ or any of the output lists $\mathbf{y}$ of the assignments $(\mathbf{y}, g, \mathbf{x})\in A_{\func{f}}$; and
\item there are no cyclic dependencies between variables: there is no sequence of assignments $(\mathbf{y}_1, g_1, \mathbf{x}_1), \dots, (\mathbf{y}_n, g_n, \mathbf{x}_n)$ in $A_{\func{f}}$ (with $n \geq 1$) such that $\mathsf{set}(\mathbf{y}_i)\cap \mathsf{set}(\mathbf{x}_{i+1}) \neq \varnothing$ for all $i < n$ and $\mathsf{set}(\mathbf{y}_n)\cap \mathsf{set}(\mathbf{x}_1) \neq \varnothing$.
\end{bracketenumerate}
A hierarchical $\Sigma$-diagram is \emph{gs-monoidal} if every dot diagram appearing in it is gs-monoidal.
\end{definition}

Condition (1) says that the input variables are pairwise distinct, that no input variable is defined by an assignment, and that every other variable is defined by exactly one assignment, in exactly one position.
Condition (2) says that the assignments can be ordered so that every variable is defined before it is used.
A variable may be used several times, and it may occur several times in the output list.
In particular, gs-monoidal dot diagrams have no isolated variables.

\begin{remark}\label{remark:gs-monoidal}
Under the correspondence between dot diagrams and cospans of labelled hypergraphs (\Cref{ssec:dot-diagram-category}),
gs-monoidal $\Sigma$-diagrams are exactly the \emph{gs-monoidal string diagrams} of Fritz and Liang~\cite[Definition 3.6]{fritz2023free}:
Condition (1) is their \emph{left monogamy} condition, and Condition (2) is their \emph{acyclicity} condition.
The composite and the tensor product of gs-monoidal dot diagrams (as defined in \Cref{ssec:dot-diagram-category}) are again gs-monoidal \cite[Lemma 3.9 and the discussion of the tensor product in Section 3]{fritz2023free},
and gs-monoidal $\Sigma$-diagrams modulo equivalence form the free \emph{gs-monoidal category} on $\Sigma$ \cite[Theorem 4.1 and Corollary 4.7]{fritz2023free},
that is, the free symmetric monoidal category on $\Sigma$ in which every object is equipped with $\mathsf{copy}$ and $\mathsf{del}$ operations
(but not with $\mathsf{equate}$ and $\mathsf{new}$ operations) as in \Cref{sec:semantics}. This is the reason for the name.
\end{remark}

\begin{definition}[DPP]
A \emph{discrete probabilistic program (DPP)} is a hierarchical gs-monoidal dot diagram
over the monoidal signature $\mathsf{BoolStochSig}$.
\end{definition}

For example, the functions $\func{test}$ and $\func{hasDisease}$ from \Cref{fig:code-and-diagram} are gs-monoidal (after de-sugaring),
so that $\func{hasDisease}$ is a DPP.
The dot diagram $\func{h}$ from \Cref{example:dot-diagram} is not gs-monoidal, since the variable $z$ is defined twice.

\section{Semantics: hypergraph categories}\label{sec:semantics}

Dot diagrams and their hierarchical variant are purely syntactic notions.
We now discuss their semantics. For DPPs these will be matrices over the non-negative rational numbers (in fact, substochastic matrices); in the general case, these are morphisms in a \emph{strict hypergraph category}.
For the subsequent discussion, 
we assume some familiarity with basic category theory, 
and monoidal categories in particular.

\subsection{Strict hypergraph categories}

We now define strict hypergraph categories, which are strict symmetric monoidal categories with additional structure to model structural operations in DPPs and dot diagrams in general; see~\cite{fong2019hypergraph} for a general discussion of hypergraph categories and their history.

\begin{definition}
A \emph{strict hypergraph category} is a strict symmetric monoidal category $C$, together with four families of morphisms
\begin{align*}
\mathsf{copy}_c\colon c \to c \otimes c, \quad &\mathsf{del}_c\colon c \to I_C,  \\
\mathsf{equate}_c\colon c \otimes c \to c, \quad &\mathsf{new}_c\colon I_C \to c,  
\end{align*}
where $c$ ranges over the objects of $C$ and $I_C$ is the tensor unit.
These morphisms are required to satisfy the equations shown in \Cref{fig:hypergraph-categories}. 
A \emph{strict hypergraph functor} is a strict symmetric monoidal functor that preserves $\mathsf{copy}_\bullet$, $\mathsf{del}_\bullet$, $\mathsf{equate}_\bullet$, and $\mathsf{new}_\bullet$.
\end{definition}

\begin{figure}
\centering
\scalebox{0.9}{
\input{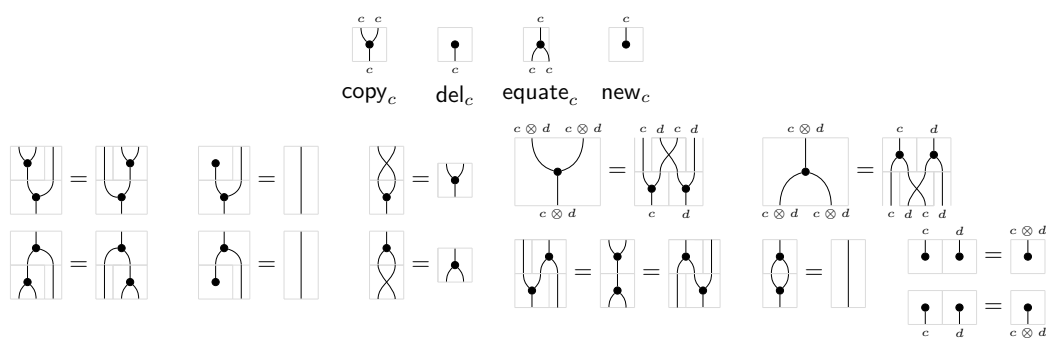}
}
\caption{The axioms of strict hypergraph categories. 
We use standard string diagram notation for symmetric monoidal categories, 
drawing light grey boxes around morphisms to avoid confusion with our representation of dot diagrams.}
\label{fig:hypergraph-categories}
\end{figure}

\subsection{Matrix hypergraph categories}

In our general framework, we will define the semantics of dot diagrams as a strict hypergraph functor. However, we are particularly interested in \emph{matrix hypergraph categories}, as the semantics of DPPs will be matrices over the non-negative rationals. We take the slightly more general viewpoint of matrices over arbitrary (commutative, unital) semirings, as these also cover further applications such as cybersecurity risk metrics in attack trees~\cite{peterseim2025unified}.

\begin{definition}
A \textit{semiring} is a set $R$  with two binary operations $+$ and $\cdot$, 
and two elements $0$ and $1$,
such that $(R,+,0)$ and $(R,\,\cdot\,,1)$ are commutative unital monoids, $\cdot$ distributes over $+$, and $0 \cdot r = 0$ for all $r \in R$.
\end{definition}

\begin{definition}
\enlargethispage{\baselineskip}
Let $R$ be a semiring.
We define the strict hypergraph category $\mathsf{Mat}_R$ as having natural numbers as objects,
and $m\times n$-matrices over $R$ as morphisms $n \to m$.
The composition of morphisms in $\mathsf{Mat}_R$ is given by matrix multiplication;
the tensor product is given on objects by the product of natural numbers, 
$n \otimes m := nm$, 
and on morphisms by the Kronecker product of matrices.
The braiding is defined on basis vectors by
\[
\mathsf{swap}_{n,m}\cdot (e_i \otimes e_j) := e_j \otimes e_i,
\]
for all $i\in \{1, \dots, n\}$, $j\in \{1,\dots, m\}$, where $e_i$ is the $i$-th standard basis vector.
Finally, the hypergraph category structure is given by,
\begin{align*}
\mathsf{copy}_{n}\cdot e_i:= e_i \otimes e_i, \quad &\mathsf{del}_{n}\cdot e_i := (1) \\
\mathsf{equate}_{n}\cdot (e_i \otimes e_j) := \delta_{ij} e_i, \quad &\mathsf{new}_{n} \cdot e_1 := \mathbf{1}_n,
\end{align*}
where $\delta_{ij} = 1$ if $i = j$ and $\delta_{ij}=0$ otherwise, 
and $\mathbf{1}_n \in R^n$ denotes the vector all of whose entries are $1$.
The axioms of hypergraph categories can easily be checked by a direct calculation.
\end{definition}
The semantic target category for DPPs will be $\mathsf{Mat}_{\mathbb{Q}_{\geq 0}}$,
where $\mathbb{Q}_{\geq 0}$ denotes the semiring of non-negative rational numbers.
We will see in \Cref{lem:dpp-substochastic} that the semantics of a DPP is in fact a \emph{substochastic} matrix:
a matrix with non-negative entries all of whose columns sum to a value of at most $1$.
Note that substochastic matrices do not form a sub-hypergraph category of $\mathsf{Mat}_{\mathbb{Q}_{\geq 0}}$,
since $\mathsf{new}_n$ is not substochastic for $n\geq 2$.

The idea is that the entry $ij$ of the semantics of a DPP records
the probability that the DPP returns the Boolean vector $\mathsf{bin}(i)$,
the binary digits of $i$, 
\emph{and} that all $\func{observe}$d events hold, 
on the input $\mathsf{bin}(j)$.
If instead we want $\func{observe}$ expressions to properly \emph{condition} on an event, 
we also need to identify two substochastic matrices
if they differ by a positive factor; 
see~\cite[Sections 6.2--6.3]{stein2024probabilistic} for a thorough discussion of the denotational semantics of discrete probabilistic programs, 
as well as~\cite{stein2021structural} for a general treatment of the categorical semantics of probabilistic programs.

\subsection{The strict hypergraph category of dot diagrams}\label{ssec:dot-diagram-category}

Dot diagrams (up to equivalence) themselves also form a strict hypergraph category $\underline{\mathrm{Dot}}(\Sigma)$;
in fact, the \emph{free} strict hypergraph category over a monoidal signature $\Sigma$ (see Figure~\ref{fig:dot-diagram-category}):
\begin{itemize}
\item Its objects are given by lists of sorts from $\Sigma$; 
\item its morphisms are (equivalence classes of) dot diagrams;
\item the (sequential) composite of two dot diagrams $\func{f}, \func{g}$ is given by gluing the output variables of $\func{f}$ to the input variables of $\func{g}$, merging the corresponding dots;
\item the tensor product is given by disjoint union;
\item $\func{copy}$, $\func{del}$, $\func{equate}$ and $\func{new}$ are as in Figure~\ref{fig:dot-diagram-category}.
\end{itemize}
Precise definitions of the composition and tensor product of dot diagrams are given in \Cref{appendix:tensor-product-and-composition-of-dot-diagrams}.
They are equivalent to those of~\cite{kissinger2014finite}, which represent dot diagrams as cospans of labelled hypergraphs.
We also introduce special dot diagrams for $\func{id}$, $\func{swap}$, and $\langle f \rangle$ for each $f \in \Sigma$, as in Figure~\ref{fig:dot-diagram-category}. 
Note that, strictly speaking, these are equivalence classes of dot diagrams.

\begin{figure}
\centering
\scalebox{0.9}{
\input{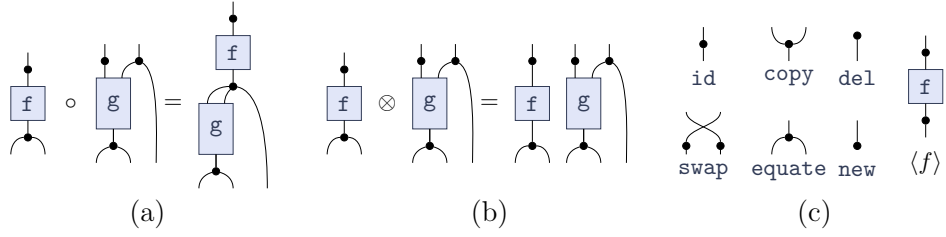}
}
\caption{(a) Composition and (b) tensor product of dot diagrams; (c) special dot diagrams.}
\label{fig:dot-diagram-category}
\end{figure}

\subsection{Interpretations and semantics}

We can now define the semantics of dot diagrams as morphisms in strict hypergraph categories. 
We first have to define how to \emph{interpret} function symbols.

\begin{definition}[Interpretation]
Let $\Sigma$ be a monoidal signature and let $C$ be a strict hypergraph category.
Then $C$ defines a monoidal signature whose sorts are the objects of 
$C$ and whose function symbols $\mathbf{S} \to \mathbf{T}$ are given by 
all morphisms $S_1 \otimes \dots \otimes S_m \to T_1 \otimes \dots \otimes T_n$ in $C$.
An \emph{interpretation} of $\Sigma$ in $C$ is a morphism of monoidal signatures $\mathcal{I}\colon \Sigma \to C$.
\end{definition}

\begin{example}
The \emph{substochastic matrix interpretation} $\mathcal{S}\colon\mathsf{BoolStochSig}\to \mathsf{Mat}_{\mathbb{Q}_{\geq 0}}$ is the interpretation of the signature of DPPs which is defined as follows: the sole sort $\mathbb{B}$ of $\mathsf{BoolStochSig}$ is mapped to the object $2$ of $\mathsf{Mat}_{\mathbb{Q}_{\geq 0}}$. On function symbols $\mathcal{S}$ is given by the following substochastic matrices:
\begin{align*}
\mathcal{S}(\land) &:= \begin{pmatrix}
1 & 1 & 1 & 0 \\
0 & 0 & 0 & 1
\end{pmatrix}
, \quad 
\mathcal{S}(\lor) := \begin{pmatrix}
1 & 0 & 0 & 0 \\
0 & 1 & 1 & 1
\end{pmatrix}
, \quad
\mathcal{S}(\lnot) := \begin{pmatrix}
0 & 1 \\
1 & 0 
\end{pmatrix}
, \\ 
&\mathcal{S}(\func{flip}(p)) := \begin{pmatrix}
1-p  \\
p   
\end{pmatrix}
, \qquad
\mathcal{S}(\func{observe}) := \begin{pmatrix}
0 & 1 
\end{pmatrix}.
\end{align*}
\end{example}

With the interpretation of function symbols in place, the semantics of dot diagrams is the unique extension of the interpretation to dot diagrams. The following result makes this precise. The proof is given in~\cite[Theorem 4.6]{kissinger2014finite}, by reduction to the corresponding statement for traced monoidal categories~\cite[Theorem 5.5.10]{kissinger2011pictures}.\footnote{
The characterisation of free traced monoidal categories was stated earlier
in~\cite{hasegawa2008finite}, 
but the proof outline given there seems to merely state \emph{what} there is to prove, 
without saying \emph{how}.
}
For a more recent, published reference, see~\cite[Corollary 4.2]{bonchi2022stringI}.\footnote{
Note, however, that a footnote to~\cite[Remark 4.3]{fritz2023free} points out two gaps in the proof of the corresponding result in~\cite{zanasi2017rewriting}.
As far as we can tell, the proof of~\cite[Corollary 4.2]{bonchi2022stringI} follows the same scheme, so that these gaps apply to it as well.
Nevertheless, the idea of the proof, 
which relies on a Frobenius-type decomposition similar to the one we will also use
in \Cref{ssec:badalg}, appears sound;
hence the gaps can presumably be closed using suitable inductive arguments similar 
to the ones in~\cite{fritz2023free}.
}

\begin{theorem}\label{thm:dot-diagrams-free-hypergraph-category}
Let $\Sigma$ be a monoidal signature, let $C$ be a strict hypergraph category, 
and let $\mathcal{I}\colon \Sigma \to C$ be an interpretation.
Then there exists a unique strict hypergraph functor
\[ [\![\,\cdot\,]\!]_{\mathcal{I}}\colon \underline{\mathrm{Dot}}(\Sigma) \to C \]
such that for all function symbols $f$ in $\Sigma$, we have
$[\![ \langle f \rangle ]\!]_{\mathcal{I}} = \mathcal{I}(f)$, 
where $\langle f \rangle$ is the dot diagram (modulo equivalence) associated to $f$. 
In other words, $\underline{\mathrm{Dot}}(\Sigma)$ is the free strict hypergraph category on $\Sigma$.
\end{theorem}

\begin{definition}[Semantics]
Let $\Sigma$ be a monoidal signature, let $C$ be a strict hypergraph category,
and let $\mathcal{I}\colon \Sigma \to C$ be an interpretation.
Moreover, let $\func{f}$ be a dot diagram.
The semantics of $\func{f}$ under $\mathcal{I}$ is $[\![ \,\underline{\func{f}} \,]\!]_{\mathcal{I}}$, also written $[\![\func{f}]\!]_{\mathcal{I}}$, where $[\![\,\cdot\,]\!]_{\mathcal{I}}$ is the unique strict hypergraph functor from \Cref{thm:dot-diagrams-free-hypergraph-category}.
\end{definition}

\emph{Hierarchical} dot diagrams will receive semantics in \Cref{definition:semantics-of-hierarchical-dot-diagrams}.

\section{Algebraisations}\label{sec:algebraisations}

From the previous section it follows (ignoring hierarchical structure for the moment) that the semantics of a DPP with 0 inputs and 1 output is a $2\times 1$-matrix over $\mathbb{Q}_{\geq 0}$ (in fact, a substochastic one, see \Cref{lem:dpp-substochastic}). In Section~\ref{ssec:inference}, we will relate this to the acceptance probability of Problem~\ref{problem:boolean-probabilistic-inference}. For now, we will turn to the problem of \emph{computing} this matrix. For this, we define \emph{hypergraph terms}, which are (non-unique) algebraic representations of dot diagrams. The process of finding a ``good'' hypergraph term is called \emph{algebraisation}, which is the topic of Section~\ref{sec:algebraisation}.

\subsection{Hypergraph terms}

Hypergraph terms are formal composites and tensor products of function symbols and constant symbols representing the special morphisms in a hypergraph category.

\begin{definition}
Let $\Sigma$ be a monoidal signature. 
The monoidal signature of \emph{hypergraph terms} over $\Sigma$ is given by
\[\mathrm{Term}(\Sigma):=(\mathrm{Term}(\Sigma)_0, \mathrm{Term}(\Sigma)(\,\cdot\,, \,\cdot\,)),\] 
where the set of sorts $\mathrm{Term}(\Sigma)_0:= \Sigma_0$ is the same as that of $\Sigma$ and $\mathrm{Term}(\Sigma)(\,\cdot\,, \,\cdot\,)$ is defined inductively as follows: 
\begin{enumerate}
\item if $f \in \Sigma(\mathbf{S}, \mathbf{T})$, then $f \in \mathrm{Term}(\Sigma)(\mathbf{S}, \mathbf{T});$
\item if  $\mathbf{S}, \mathbf{T}\in \Sigma_0^*$, then 
\begin{align*}
\mathrm{id}_{\mathbf{T}} \in \mathrm{Term}(\Sigma)(\mathbf{T}, \mathbf{T}), \qquad
&\mathrm{swap}_{\mathbf{S}, \mathbf{T}} \in \mathrm{Term}(\Sigma)(\mathbf{S}\mathbf{T}, \mathbf{T}\mathbf{S}), \\
\mathrm{copy}_{\mathbf{T}} \in \mathrm{Term}(\Sigma)(\mathbf{T}, \mathbf{T}\mathbf{T}), \qquad
&\mathrm{del}_{\mathbf{T}} \in \mathrm{Term}(\Sigma)(\mathbf{T}, ()), \\
\mathrm{equate}_{\mathbf{T}} \in \mathrm{Term}(\Sigma)(\mathbf{T}\mathbf{T}, \mathbf{T}), \qquad
&\mathrm{new}_{\mathbf{T}} \in \mathrm{Term}(\Sigma)((), \mathbf{T}); 
\end{align*}
\item if $s\in \mathrm{Term}(\Sigma)(\mathbf{T}, \mathbf{U})$ and $t\in \mathrm{Term}(\Sigma)(\mathbf{S}, \mathbf{T})$, then 
\[s \circ t \in \mathrm{Term}(\Sigma)(\mathbf{S}, \mathbf{U});\]
\item if $s\in \mathrm{Term}(\Sigma)(\mathbf{S}, \mathbf{T})$ and $t\in \mathrm{Term}(\Sigma)(\mathbf{U}, \mathbf{V})$, then 
\[s \otimes t \in \mathrm{Term}(\Sigma)(\mathbf{S}\mathbf{U}, \mathbf{T}\mathbf{V}).\]
\end{enumerate}
A \emph{hypergraph term} is any function symbol $t\in \mathrm{Term}(\Sigma)$ of this signature.

We write $s\subseteq t$ when $s$ is a sub-term of $t$; this relationship is defined in the standard inductive way. 
The same applies to the \emph{substitution} of a finite family of hypergraph terms $(s_i)_{i\in I}$ for a family of function symbols $(g_i)_{i\in I}$ (from $\Sigma$) in a hypergraph term $t$, 
written $t[g_i \mapsto s_i]$.

For computational purposes, 
we assume that hypergraph terms are represented as labelled directed acyclic graphs (DAGs) with maximal sharing, 
i.e.~the nodes of the DAG underlying a hypergraph term $t$ represent unique sub-terms of $t$; see Figure~\ref{fig:syntax-dags}, discussed below.   
\end{definition}

\begin{example}\label{example:hypergraph-terms}
Let $\Sigma$ be the monoidal signature with a
single sort $\mathbb{S}$ and two function symbols $\func{f}\colon (\mathbb{S}, \mathbb{S}) \to (\mathbb{S})$ and $\func{g}\colon (\mathbb{S}, \mathbb{S}) \to (\mathbb{S}, \mathbb{S})$.
Then
\[t = \func{f} \circ (\mathrm{equate}_{(\mathbb{S})} \otimes \mathrm{id}_{(\mathbb{S})}) \circ (\func{f} \otimes \func{g}) \circ (\mathrm{id}_{(\mathbb{S}, \mathbb{S})} \otimes \mathrm{copy}_{(\mathbb{S})}) \in \mathrm{Term}(\Sigma)((\mathbb{S}, \mathbb{S}, \mathbb{S}), (\mathbb{S}))\]
and  $s := \func{f} \circ \mathrm{swap}_{(\mathbb{S}, \mathbb{S})} \in \mathrm{Term}(\Sigma)((\mathbb{S}, \mathbb{S}), (\mathbb{S}))$ are hypergraph terms over $\Sigma$ (see Figure~\ref{fig:syntax-dags}). 
Then 
\begin{align*}
t[\func{f} \mapsto s] =\: &(\func{f} \circ \mathrm{swap}_{(\mathbb{S}, \mathbb{S})}) \circ (\mathrm{equate}_{(\mathbb{S})} \otimes \mathrm{id}_{(\mathbb{S})}) \\
&\circ ((\func{f} \circ \mathrm{swap}_{(\mathbb{S}, \mathbb{S})}) \otimes \func{g}) \circ (\mathrm{id}_{(\mathbb{S}, \mathbb{S})} \otimes \mathrm{copy}_{(\mathbb{S})}),
\end{align*}
as shown in \Cref{fig:syntax-dags} (b).
\end{example}

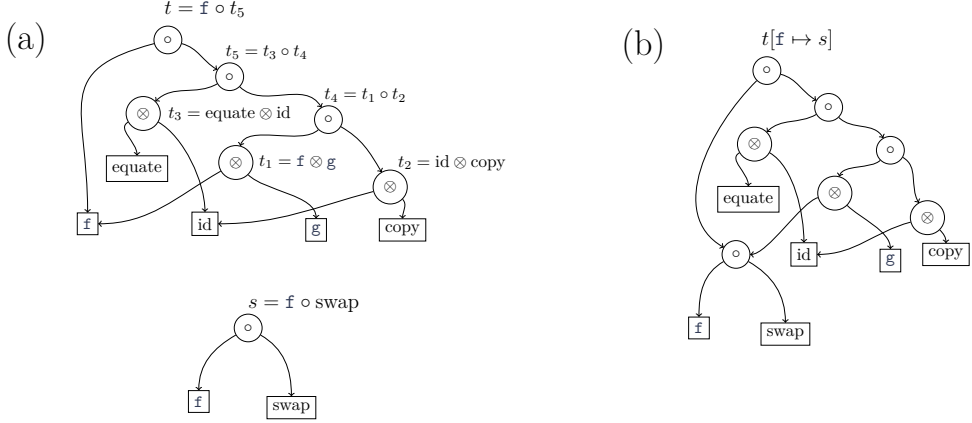
\begin{figure}
\centering
\vspace{-2mm}
\scalebox{0.65}{
\begin{tikzpicture}
	\begin{pgfonlayer}{nodelayer}
		\node [style=op dot] (0) at (-13.75, 8.5) {$\circ$};
		\node [style=op dot] (1) at (-11.25, 7) {$\circ$};
		\node [style=op dot] (2) at (-7.25, 5.25) {$\circ$};
		\node [style=op dot] (3) at (-14.75, 5.5) {$\otimes$};
		\node [style=none] (4) at (-9.75, 8) {$t_5 = t_3 \circ t_4$};
		\node [style=op dot] (5) at (-11, 3.5) {$\otimes$};
		\node [style=op dot] (6) at (-4.75, 2.5) {$\otimes$};
		\node [style=none] (7) at (-2.25, 3.5) {$t_2 = \mathrm{id} \otimes \mathrm{copy}$};
		\node [style=none] (8) at (-8.5, 3.5) {$t_1 = \func{f} \otimes \func{g}$};
		\node [style=none] (9) at (-5.75, 6.25) {$t_4 = t_1 \circ t_2$};
		\node [style=none] (10) at (-11.25, 5.5) {$t_3 = \mathrm{equate} \otimes \mathrm{id}$};
		\node [style=none] (11) at (-12.25, 9.75) {{\Large $t = \func{f} \circ t_5$}};
		\node [style=box] (12) at (-17, 1) {$\func{f}$};
		\node [style=box] (13) at (-12.25, 1) {$\mathrm{id}$};
		\node [style=box] (14) at (-7.75, 0.75) {$\func{g}$};
		\node [style=box] (15) at (-4.25, 0.75) {$\mathrm{copy}$};
		\node [style=box] (16) at (-15, 3.25) {$\mathrm{equate}$};
		\node [style=op dot] (17) at (-10.5, -3.25) {$\circ$};
		\node [style=box] (18) at (-12.5, -6.25) {$\func{f}$};
		\node [style=box] (19) at (-8.75, -6.5) {$\mathrm{swap}$};
		\node [style=none] (20) at (-8.25, -2.25) {{\Large $s = \func{f} \circ \mathrm{swap}$}};
		\node [style=op dot] (21) at (10.5, 7.25) {$\circ$};
		\node [style=op dot] (22) at (13, 5.75) {$\circ$};
		\node [style=op dot] (23) at (15.5, 4) {$\circ$};
		\node [style=op dot] (24) at (10, 4.25) {$\otimes$};
		\node [style=op dot] (26) at (13.25, 2.25) {$\otimes$};
		\node [style=op dot] (27) at (17, 1.25) {$\otimes$};
		\node [style=none] (32) at (11.75, 8.5) {{\Large $t[\func{f} \mapsto s]$}};
		\node [style=box] (34) at (12, -0.25) {$\mathrm{id}$};
		\node [style=box] (35) at (15.5, -0.5) {$\func{g}$};
		\node [style=box] (36) at (17.75, -0.25) {$\mathrm{copy}$};
		\node [style=box] (37) at (9.75, 2) {$\mathrm{equate}$};
		\node [style=op dot] (38) at (9.25, -0.25) {$\circ$};
		\node [style=box] (39) at (7.75, -3.25) {$\func{f}$};
		\node [style=box] (40) at (11.25, -3.5) {$\mathrm{swap}$};
		\node [style=none] (43) at (-19.5, 8.5) {{\huge (a)}};
		\node [style=none] (44) at (5.5, 8.25) {{\huge (b)}};
	\end{pgfonlayer}
	\begin{pgfonlayer}{edgelayer}
		\draw [style=black arrow, in=90, out=-165, looseness=1.50] (0) to (12);
		\draw [style=black arrow, in=0, out=-150] (5) to (12);
		\draw [style=black arrow, in=90, out=-150] (3) to (16);
		\draw [style=black arrow, in=150, out=-15] (0) to (1);
		\draw [style=black arrow, in=45, out=-150] (1) to (3);
		\draw [style=black arrow, in=120, out=-45] (1) to (2);
		\draw [style=black arrow, in=75, out=-135, looseness=0.75] (2) to (5);
		\draw [style=black arrow, in=120, out=-30] (2) to (6);
		\draw [style=black arrow, in=90, out=-30, looseness=0.75] (3) to (13);
		\draw [style=black arrow, in=0, out=-165] (6) to (13);
		\draw [style=black arrow, in=90, out=-45, looseness=0.75] (5) to (14);
		\draw [style=black arrow, in=90, out=-45] (6) to (15);
		\draw [style=black arrow, in=90, out=-150] (17) to (18);
		\draw [style=black arrow, in=90, out=-30] (17) to (19);
		\draw [style=black arrow, in=90, out=-150] (24) to (37);
		\draw [style=black arrow, in=150, out=-15] (21) to (22);
		\draw [style=black arrow, in=45, out=-150] (22) to (24);
		\draw [style=black arrow, in=120, out=-45] (22) to (23);
		\draw [style=black arrow, in=75, out=-135, looseness=0.75] (23) to (26);
		\draw [style=black arrow, in=120, out=-30] (23) to (27);
		\draw [style=black arrow, in=90, out=-30, looseness=0.75] (24) to (34);
		\draw [style=black arrow, in=0, out=-165] (27) to (34);
		\draw [style=black arrow, in=90, out=-45, looseness=0.75] (26) to (35);
		\draw [style=black arrow, in=90, out=-45] (27) to (36);
		\draw [style=black arrow, in=90, out=-150] (38) to (39);
		\draw [style=black arrow, in=90, out=-30] (38) to (40);
		\draw [style=black arrow, in=150, out=-135] (21) to (38);
		\draw [style=black arrow, in=0, out=-165, looseness=0.50] (26) to (38);
	\end{pgfonlayer}
\end{tikzpicture}
}
\caption{(a) DAGs representing the terms $t$ and $s$ from \Cref{example:hypergraph-terms}; 
(b) the substitution of $s$ for $\func{f}$ in $t$.}
\label{fig:syntax-dags}
\end{figure}

\subsection{Semantics of hypergraph terms}

The semantics of hypergraph terms are defined in the obvious inductive way.

\begin{definition}\label{def:hypergraph-terms}
Let $\Sigma$ be a monoidal signature, let $C$ be a strict hypergraph category, and let $\mathcal{I}\colon\:\Sigma \to C$ be an interpretation.
The \emph{semantics} $[\![ t ]\!]_{\mathcal{I}}$ of a hypergraph term $t$ under $\mathcal{I}$ is defined recursively as follows,
where $\mathcal{I}(\mathbf{T}) := \mathcal{I}(T_1) \otimes \dots \otimes \mathcal{I}(T_n)$ for a list $\mathbf{T} = (T_1, \dots, T_n)$ of sorts:
\begin{align*}
[\![ f ]\!]_{\mathcal{I}} := \mathcal{I}(f), \quad
[\![ \mathrm{id}_{\mathbf{T}} ]\!]_{\mathcal{I}} := \mathsf{id}_{\mathcal{I}(\mathbf{T})}, \quad
&[\![ s \circ t ]\!]_{\mathcal{I}} := [\![ s ]\!]_{\mathcal{I}} \circ [\![ t ]\!]_{\mathcal{I}}\\
[\![ \mathrm{swap}_{\mathbf{S}, \mathbf{T}} ]\!]_{\mathcal{I}} := \mathsf{swap}_{\mathcal{I}(\mathbf{S}),\, \mathcal{I}(\mathbf{T})},\qquad
&[\![ s \otimes t ]\!]_{\mathcal{I}} := [\![ s ]\!]_{\mathcal{I}} \otimes [\![ t ]\!]_{\mathcal{I}}, \\
[\![ \mathrm{copy}_{\mathbf{T}} ]\!]_{\mathcal{I}} := \mathsf{copy}_{\mathcal{I}(\mathbf{T})}, \qquad
&[\![ \mathrm{del}_{\mathbf{T}} ]\!]_{\mathcal{I}} := \mathsf{del}_{\mathcal{I}(\mathbf{T})}, \\
[\![ \mathrm{equate}_{\mathbf{T}} ]\!]_{\mathcal{I}} := \mathsf{equate}_{\mathcal{I}(\mathbf{T})}, \qquad
&[\![ \mathrm{new}_{\mathbf{T}} ]\!]_{\mathcal{I}} := \mathsf{new}_{\mathcal{I}(\mathbf{T})}.
\end{align*}
Two hypergraph terms $s,t$ are \emph{equivalent}
if for every strict hypergraph category $C$ and every interpretation $\mathcal{I}\colon\:\Sigma \to C$, 
we have $[\![ s ]\!]_{\mathcal{I}} = [\![ t ]\!]_{\mathcal{I}}$. 
We write $\underline{t}$ for the equivalence class of a hypergraph term $t$ under this equivalence relation. 
Finally, we denote the set of hypergraph terms $\mathbf{S} \to \mathbf{T}$ modulo equivalence 
by $\underline{\mathrm{Term}}(\Sigma)(\mathbf{S}, \mathbf{T})$, 
and accordingly, we write $\underline{\mathrm{Term}}(\Sigma)$ for the resulting strict hypergraph category.
\end{definition}

As for $\underline{\mathrm{Dot}}(\Sigma)$, the strict hypergraph category $\underline{\mathrm{Term}}(\Sigma)$ is the free strict hypergraph category on $\Sigma$.
In more detail, this means the following.

\begin{theorem}\label{thm:terms-free-hypergraph-category}
Let $\Sigma$ be a monoidal signature, let $C$ be a strict hypergraph category,
and let $\mathcal{I}: \Sigma \to C$ be an interpretation. Then
\[[\![ \,\cdot\, ]\!]_{\mathcal{I}}\colon\: \underline{\mathrm{Term}}(\Sigma) \to C\]
is the unique strict hypergraph functor such that for all function symbols $f \in \Sigma$, 
we have $[\![ f ]\!]_{\mathcal{I}} = \mathcal{I}(f)$.
\end{theorem}

\subsection{From terms to dot diagrams}\label{ssec:terms-to-dot-diagrams}

By the essential uniqueness of the free strict hypergraph category, 
we have an isomorphism of strict hypergraph categories $\underline{\mathrm{Dot}}(\Sigma) \cong \underline{\mathrm{Term}}(\Sigma)$.
The isomorphism in the direction from hypergraph terms to dot diagrams can easily be described explicitly, as follows.

\begin{definition}
Let $\Sigma$ be a monoidal signature. Let 
$\mathcal{I}_\mathsf{dot}\colon \Sigma \to \underline{\mathrm{Dot}}(\Sigma)$
be the interpretation that is the identity on sorts
and assigns to each function symbol $f$ in $\Sigma$ its corresponding dot diagram $\langle f \rangle$.
The \emph{associated dot diagram} (modulo equivalence) of a hypergraph term $t$ is 
\[ \underline{\mathsf{dot}}(t) := [\![ t]\!]_{\mathcal{I}_\mathsf{dot}}, \]
i.e.~its semantics in $\underline{\mathrm{Dot}}(\Sigma)$ under $\mathcal{I}_\mathsf{dot}$. 
\end{definition}

\subsection{Algebraisations of dot diagrams}\label{ssec:algebraisations}

Going in the other direction, from dot diagrams to hypergraph terms, 
is what we call \emph{algebraisation};
the resulting term that equivalently presents a morphism in the free strict hypergraph category is called \emph{an} algebraisation.

\begin{definition}
An \emph{algebraisation} of a dot diagram ${\func{f}}$ 
is a hypergraph term $t$ such that $\underline{\mathsf{dot}}(t) = \underline{{\func{f}}}$. 
\end{definition}

If $t$ is an algebraisation of ${\func{f}}$, then $[\![{\func{f}}]\!]_{\mathcal{I}} = [\![t]\!]_{\mathcal{I}}$ for every interpretation $\mathcal{I}$ of $\Sigma$ in a strict hypergraph category:
both $t \mapsto [\![\underline{\mathsf{dot}}(t)]\!]_{\mathcal{I}}$ and $t \mapsto [\![t]\!]_{\mathcal{I}}$ are strict hypergraph functors on $\underline{\mathrm{Term}}(\Sigma)$ that map every function symbol $g$ to $\mathcal{I}(g)$,
and by \Cref{thm:terms-free-hypergraph-category} there is only one such functor.
This is how the semantics of dot diagrams is computed in \Cref{sec:applications}.

\subsection{Quantifying the complexity of terms}

We consider two ways to quantify the complexity of a hypergraph term $t$. 
The first one, the \emph{DAG-size}, corresponds to the size of the directed acyclic graph representing $t$. 
The second one, the \emph{width} of a hypergraph term, corresponds to the largest number of in- and outputs of any intermediate result reached when evaluating $t$ bottom-up.
The width of $t$ can be used, for example, to bound the time complexity of evaluating the semantics of $t$ in a matrix category, where composition is given by matrix multiplication and the tensor product is given by the Kronecker product.

\begin{definition}
Let $\Sigma$ be a monoidal signature and let $t: (S_1, \dots, S_n) \to (T_1, \dots, T_m)$ be a hypergraph term over $\Sigma$.
The \emph{DAG-size} $|t|$ of $t$ is the number of distinct sub-terms of $t$. 
The \emph{width} of $t$ is the maximum,
\[ \mathsf{width}(t) := \max_{s\subseteq t} |\mathrm{dom}(s)| + |\mathrm{codom}(s)|, \]
over the number of inputs plus the number of outputs of sub-terms of $t$.
\end{definition}

Note that $\mathsf{width}(t)$ is \emph{not} easily read from the DAG representations of Figure~\ref{fig:syntax-dags}: 
one also needs to take the number of inputs and outputs at each vertex into account.

\section{An algebraisation algorithm}\label{sec:algebraisation}

The goal of this section is to describe an algorithm 
that converts a dot diagram over a monoidal signature
into a hypergraph term of low width and small DAG-size.
The main high-level steps of this algorithm are shown in \Cref{fig:algorithm}.
The width and DAG-size of the algebraisation that our algorithm finds 
depend on the \emph{treewidth} of the so-called \emph{primal graph} of ${\func{f}}$ (and, to a lesser extent, on the maximal arity of ${\func{f}}$). 
In the following, a ``graph'' is always a finite, undirected graph which may contain loops or multiple edges.

\begin{definition}[Primal graph]
Let ${\func{f}}$ be a $\Sigma$-diagram over a monoidal signature $\Sigma$.
The \emph{primal graph} ${\func{f}}'$ of ${\func{f}}$
is the graph whose set of vertices is the set $X_{\func{f}}$ of variables of ${\func{f}}$, with exactly one edge between any two distinct variables that both appear either a) in some common assignment, b) as inputs, or c) as outputs.
\end{definition}

Besides the treewidth of the primal graph, our bounds involve the maximal arity of ${\func{f}}$.

\begin{definition}[Maximal arity]\label{definition:maximal-arity}
The \emph{arity} of a function symbol $g\colon \mathbf{S} \to \mathbf{T}$ is $|\mathbf{S}| + |\mathbf{T}|$, its number of in- and outputs, counted together.
Let ${\func{f}}$ be a $\Sigma$-diagram.
The \emph{maximal arity} $\mathsf{maxar}({\func{f}})$ of ${\func{f}}$ is the maximum of $|\mathsf{in}_{\func{f}}| + |\mathsf{out}_{\func{f}}|$, the arity of ${\func{f}}$ as a function symbol of $\mathrm{Dot}(\Sigma)$, 
and of the arities of the function symbols occurring in the assignments of ${\func{f}}$.
\end{definition}

Every function symbol occurring in ${\func{f}}$ is a sub-term of every algebraisation $t$ of ${\func{f}}$, and $t$ is a sub-term of itself with $|\mathsf{in}_{\func{f}}| + |\mathsf{out}_{\func{f}}|$ in- and outputs;
hence $\mathsf{width}(t) \geq \mathsf{maxar}({\func{f}})$.

The bounds of this section, and \Cref{thm:semantics-in-matrix-categories}, are stated for dot diagrams without isolated variables
(and, in the hierarchical case, for hierarchical dot diagrams in which no dot diagram appearing has isolated variables),
since the parameters $n$, $k$ and $a$ of \Cref{thm:algebraisation} do not bound the number of isolated variables.
This is no loss of generality:
by the definition of the strict hypergraph category $\underline{\mathrm{Dot}}(\Sigma)$,
an isolated variable of sort $S$ is a tensor factor $\mathsf{del}_S \circ \mathsf{new}_S$ of the dot diagram,
which can be split off in linear time.
Hence an algebraisation of a dot diagram with isolated variables is obtained from an algebraisation of the dot diagram with these variables removed
by tensoring with one term $\mathrm{del}_{(S)}\circ \mathrm{new}_{(S)}$ per isolated variable,
which does not change the width (unless it is $0$) and adds a number of sub-terms linear in the number of isolated variables.
To exclude degenerate cases, the bounds moreover assume that the dot diagram has an input, an output, or an assignment whose function symbol has an input or an output,
so that $a \geq 1$ (in the hierarchical case: that this holds for at least one dot diagram appearing in it).
Otherwise, a dot diagram without isolated variables has no variables at all,
and the tensor product of its function symbols, in the order of its assignments, is an algebraisation of width $0$ that is written down in linear time.

The following definition is standard; see~\cite{bodlaender2005discovering} for an exposition of tree decompositions and treewidth, including their history and applications.

\begin{definition}
Let $G$ be a graph. 
A \emph{tree decomposition} of $G$ is a tree $\mathcal{T}$
together with a family $(X_t)_{t\in V(\mathcal{T})}$ 
of subsets $X_t \subseteq V(G)$, called \emph{bags}, such that:
\begin{enumerate}
\item The bags of $\mathcal{T}$ cover $G$:
\[\bigcup_{t\in V(\mathcal{T})} X_t = V(G). \]
\item For all adjacent $v, w\in V(G)$, there exists some $t\in V(\mathcal{T})$ such that $v, w \in X_t$.
\item For all $v\in V(G)$, the set $\{t\in V(\mathcal{T}) \mid v \in X_t\}$ induces a connected subtree of $\mathcal{T}$. 
\end{enumerate}
The \emph{width} of a tree decomposition $\mathcal{T}$ is the cardinality of its largest bag, minus one,
\[ \mathsf{width}(\mathcal{T}) := \max_{t\in V(\mathcal{T})} |X_t| - 1. \]
The \emph{treewidth} $\mathsf{tw}(G)$ of $G$ is the minimum width of any tree decomposition of $G$.
\end{definition}

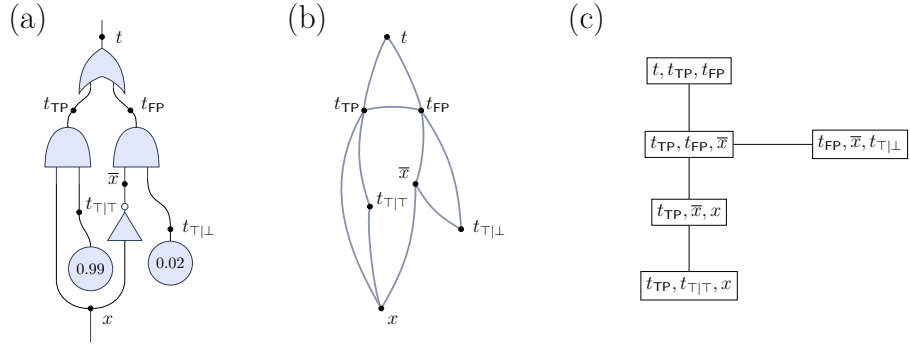
\begin{figure}
\centering
\scalebox{0.6}{
\begin{tikzpicture}
	\begin{pgfonlayer}{nodelayer}
		\node [style=dot] (76) at (1.75, -0.25) {};
		\node [style=dot] (83) at (0.25, -5.75) {};
		\node [style=none] (84) at (5, -2.25) {{\Large $t_{\top|\bot}$}};
		\node [style=none] (85) at (0.75, -6.25) {{\Large $x$}};
		\node [style=none] (86) at (0.75, -1) {{\Large $t_{\top | \top}$}};
		\node [style=none] (87) at (-1.25, 3.25) {{\Large $t_{\mathsf{TP}}$}};
		\node [style=none] (88) at (2.75, 3.25) {{\Large $t_{\mathsf{FP}}$}};
		\node [style=none] (89) at (1.25, 6.25) {{\Large $t$}};
		\node [style=dot] (91) at (-0.5, 3) {};
		\node [style=dot] (92) at (2, 3) {};
		\node [style=dot] (93) at (3.75, -2.25) {};
		\node [style=dot] (94) at (-0.25, -1.25) {};
		\node [style=box] (96) at (13.75, 4.75) {{\Large $t, t_{\mathsf{TP}}, t_{\mathsf{FP}}$}};
		\node [style=box] (97) at (13.75, 1.5) {{\Large $t_{\mathsf{TP}},  t_{\mathsf{FP}}, \overline{x}$}};
		\node [style=box] (98) at (21.25, 1.5) {{\Large $t_{\mathsf{FP}}, \overline{x}, t_{\top | \bot}$}};
		\node [style=box] (99) at (13.75, -1.5) {{\Large $t_{\mathsf{TP}},  \overline{x}, x$}};
		\node [style=box] (100) at (13.75, -4.75) {{\Large $t_{\mathsf{TP}},  t_{\top | \top}, x$}};
		\node [style=dot] (102) at (0.5, 6.25) {};
		\node [style=none] (103) at (1.25, 0.25) {{\Large $\overline{x}$}};
		\node [style=none] (149) at (-3, 7) {{\huge (b)}};
		\node [style=none] (150) at (9.25, 7) {{\huge (c)}};
		\node [style=none] (151) at (-14.5, 0.5) {};
		\node [style=none] (152) at (-13.5, 2.25) {};
		\node [style=none] (153) at (-12.5, 0.5) {};
		\node [style=none] (154) at (-13.5, 2.25) {};
		\node [style=none] (155) at (-13, 0.5) {};
		\node [style=none] (156) at (-14, 0.5) {};
		\node [style=none] (157) at (-14, -0.25) {};
		\node [style=none] (158) at (-13, -0.25) {};
		\node [style=none] (159) at (-13, 3.75) {};
		\node [style=none] (160) at (-12, 5.75) {};
		\node [style=none] (161) at (-11, 3.75) {};
		\node [style=none] (162) at (-12, 4.5) {};
		\node [style=none] (163) at (-12, 7) {};
		\node [style=none] (164) at (-11.5, 4.25) {};
		\node [style=none] (165) at (-12.5, 4.25) {};
		\node [style=none] (166) at (-12.5, 4) {};
		\node [style=none] (167) at (-11.5, 4) {};
		\node [style=none] (168) at (-11.5, 0.5) {};
		\node [style=none] (169) at (-10.5, 2.25) {};
		\node [style=none] (170) at (-9.5, 0.5) {};
		\node [style=none] (171) at (-10.5, 2.25) {};
		\node [style=none] (172) at (-10, 0.5) {};
		\node [style=none] (173) at (-11, 0.5) {};
		\node [style=none] (174) at (-11, 0.25) {};
		\node [style=none] (175) at (-10, 0) {};
		\node [style=none] (176) at (-11.75, -2.75) {};
		\node [style=none] (177) at (-11, -1.5) {};
		\node [style=none] (178) at (-10.25, -2.75) {};
		\node [style=dot] (179) at (-11, -0.25) {};
		\node [style=none] (180) at (-11, -2.75) {};
		\node [style=light blue dot] (181) at (-12.5, -4) {$0.99$};
		\node [style=light blue dot] (182) at (-9, -3.75) {$0.02$};
		\node [style=none] (183) at (-14, -4.5) {};
		\node [style=none] (184) at (-11, -4.5) {};
		\node [style=none] (185) at (-12.5, -7.25) {};
		\node [style=dot] (186) at (-12.5, -5.75) {};
		\node [style=not dot] (187) at (-11, -1.25) {};
		\node [style=dot] (188) at (-13.25, 3) {};
		\node [style=dot] (189) at (-10.75, 3) {};
		\node [style=dot] (190) at (-9, -2.25) {};
		\node [style=dot] (191) at (-13, -1.5) {};
		\node [style=dot] (192) at (-12, 6.25) {};
		\node [style=none] (207) at (-10.75, 0) {};
		\node [style=none] (229) at (-15.25, 7) {{\huge (a)}};
		\node [style=none] (230) at (-11.25, 6.25) {{\Large $t$}};
		\node [style=none] (231) at (-14, 3.25) {{\Large $t_{\mathsf{TP}}$}};
		\node [style=none] (232) at (-9.75, 3.25) {{\Large $t_{\mathsf{FP}}$}};
		\node [style=none] (233) at (-11.5, 0) {{\Large $\overline{x}$}};
		\node [style=none] (234) at (-12, -1.25) {{\Large $t_{\top | \top}$}};
		\node [style=none] (235) at (-7.75, -2.25) {{\Large $t_{\top|\bot}$}};
		\node [style=none] (236) at (-11.75, -6.25) {{\Large $x$}};
	\end{pgfonlayer}
	\begin{pgfonlayer}{edgelayer}
		\draw [style=solid bluegrey, bend left=15, looseness=0.75] (91) to (92);
		\draw [style=solid bluegrey, bend right=15] (91) to (94);
		\draw [style=solid bluegrey, bend right=15, looseness=0.50] (94) to (83);
		\draw [style=solid bluegrey, in=-120, out=120] (83) to (91);
		\draw [style=solid bluegrey, bend right=15] (83) to (76);
		\draw [style=solid bluegrey, bend right=15, looseness=0.75] (76) to (92);
		\draw [style=solid bluegrey, bend right=15] (93) to (92);
		\draw [style=solid bluegrey, bend right=15] (76) to (93);
		\draw [style=black solid] (96) to (97);
		\draw [style=black solid] (97) to (98);
		\draw [style=black solid] (97) to (99);
		\draw [style=black solid] (99) to (100);
		\draw [style=solid bluegrey, bend left=15, looseness=0.50] (91) to (102);
		\draw [style=solid bluegrey, bend left=15, looseness=0.50] (102) to (92);
		\draw [style=dark blue solid] (153.center)
			 to (151.center)
			 to [in=-180, out=90] (152.center)
			 to [in=90, out=0] cycle;
		\draw [style=dark blue solid] (154.center) to (152.center);
		\draw [style=dark blue solid] (156.center) to (157.center);
		\draw [style=dark blue solid] (155.center) to (158.center);
		\draw [style=dark blue solid] (161.center)
			 to [in=-30, out=90] (160.center)
			 to [in=90, out=210] (159.center)
			 to [in=195, out=60, looseness=0.75] (162.center)
			 to [in=120, out=-15, looseness=0.75] cycle;
		\draw [style=dark blue solid] (163.center) to (160.center);
		\draw [style=dark blue solid] (165.center) to (166.center);
		\draw [style=dark blue solid] (164.center) to (167.center);
		\draw [style=dark blue solid] (168.center)
			 to [in=-180, out=90] (169.center)
			 to [in=90, out=0] (170.center)
			 to cycle;
		\draw [style=dark blue solid] (171.center) to (169.center);
		\draw [style=dark blue solid] (173.center) to (174.center);
		\draw [style=dark blue solid] (172.center) to (175.center);
		\draw [style=dark blue solid] (176.center)
			 to (177.center)
			 to (178.center)
			 to cycle;
		\draw [style=dark blue solid] (179) to (177.center);
		\draw [style=dark blue solid] (179) to (174.center);
		\draw [style=black solid] (183.center) to (157.center);
		\draw [style=black solid, in=-90, out=180] (186) to (183.center);
		\draw [style=black solid, in=-90, out=0] (186) to (184.center);
		\draw [style=dark blue solid] (186) to (185.center);
		\draw [style=black solid, in=90, out=-90, looseness=1.50] (175.center) to (190);
		\draw [style=black solid] (190) to (182);
		\draw [style=black solid] (158.center) to (191);
		\draw [style=black solid, in=90, out=-90] (191) to (181);
		\draw [style=black solid] (180.center) to (184.center);
		\draw [style=black solid, in=60, out=-90] (166.center) to (188);
		\draw [style=black solid, in=90, out=-135] (188) to (154.center);
		\draw [style=black solid, in=90, out=-45] (189) to (171.center);
		\draw [style=black solid, in=-90, out=135] (189) to (167.center);
	\end{pgfonlayer}
\end{tikzpicture}
}
\caption{(a) Dot diagram $\func{test}$ from \Cref{fig:code-and-diagram};
(b) its primal graph $\func{test}'$;
(c) tree decomposition of $\func{test}'$ of width $2$.}
\label{fig:primal-graph-example}
\end{figure}

The main result of this section shows that the treewidth of the primal graph of a $\Sigma$-diagram $\func{f}$ is a crucial parameter for both the width and the DAG-size of the algebraisation of $\func{f}$, as well as for the time it takes to find this algebraisation. This version of our algebraisation algorithm employs a specific algorithm for low-width tree decompositions~\cite{korhonen2023single} 
for the purpose of a clean theoretical description.
However, any algorithm for computing tree decompositions 
(or, more directly, \emph{branch decompositions of hypergraphs}) may be used.

\begin{theorem}\label{thm:algebraisation}
Let ${\func{f}}$ be a $\Sigma$-diagram without isolated variables over some monoidal signature $\Sigma$,
let $k := \mathsf{tw}({\func{f}}')$ be the treewidth of the primal graph of ${\func{f}}$, 
let $a := \mathsf{maxar}({\func{f}})$ be the maximal arity of ${\func{f}}$ (\Cref{definition:maximal-arity}), 
and let $n := 1 + |A_{\func{f}}|$ be one plus the number of assignments in ${\func{f}}$.

Then ${\func{f}}$ has an algebraisation $t$ with $\mathsf{width}(t) = O(k + a)$ and DAG-size $|t| = O(n(k+a)^2)$.
Moreover, this algebraisation can be computed in $O(2^{O(k)} n a^2)$ time.
\end{theorem}

\begin{remark}\label{remark:arity}
For a signature of bounded arity, in particular for a fixed finite signature, and for dot diagrams with a bounded number of in- and outputs, 
the parameter $a$ is bounded by a constant, and \Cref{thm:algebraisation} yields an algebraisation of width $O(k)$ and DAG-size $O(nk^2)$ in $O(2^{O(k)}n)$ time, as stated in the introduction.
In general, however, $a$ cannot be dropped: 
for the dot diagram with the single assignment $\lett y = g(x, \dots, x)$, with $m$ occurrences of $x$, 
the primal graph has two vertices, but every algebraisation contains $g$ as a sub-term and hence has width at least $m + 1$.
Note that $a \leq 2k + 2$ whenever no variable occurs more than once in the same assignment or in the same list of inputs or outputs, 
since the variables of such a list are pairwise adjacent in ${\func{f}}'$ and hence form a clique. 
In the hierarchical case (\Cref{thm:algebraisation-general}), the arities of the function symbols of the signatures $\Sigma_v$ are not bounded by $\Sigma$; 
for DPPs, however, they are bounded in terms of the treewidth (\Cref{lem:dpp-width}).
\end{remark}

The rest of this section is devoted to describing the algebraisation algorithm in detail, and showing its correctness. 
In particular, we will conclude with the proof of \Cref{thm:algebraisation} 
and its generalisation to hierarchical dot diagrams, \Cref{thm:algebraisation-general}. 
The proof consists of the following steps:
\begin{description}
\item[Sec. \ref{ssec:purewiring}] First, we analyse algebraisations of diagrams without any assignments;
\item[Sec. \ref{ssec:badalg}] We then consider ``inefficient'' algebraisations of general diagrams (disregarding width for the moment);
\item[Sec. \ref{ssec:unfold}] In preparation for the next step, 
we introduce \emph{unfolding} hierarchical dot diagrams into single large dot diagrams;
\item[Sec. \ref{ssec:branchdec}] We refactor large dot diagrams into hierarchical dot diagrams, with ``small'' component dot diagrams;
\item[Sec. \ref{ssec:proof}] Finally, we put everything together in a divide-and-conquer fashion, by using the ``inefficient'' algebraisation on each small component.
\end{description}
By the time we are done, we get the analogous result for hierarchical dot diagrams (Theorem~\ref{thm:algebraisation-general}) practically for free.
All omitted proofs can be found in \Cref{appendix:proofs-algebraisation}.

\subsection{Algebrising pure wiring} \label{ssec:purewiring}

We first algebrise \emph{pure wirings}: dot diagrams with no assignments at all. 
We distinguish three special types:
\emph{permutation-type}, \emph{copy-and-delete-type}, and \emph{equate-and-select-type} dot diagrams.

\begin{definition}
Let $\Sigma$ be a monoidal signature. 
A $\Sigma$-diagram ${\func{f}}$ is a \emph{pure wiring} if $|A_{\func{f}}| = 0$.
If, moreover, $\mathsf{in}_{\func{f}}$ and $\mathsf{out}_{\func{f}}$ both do not contain any duplicates, are of the same length and $\mathsf{set}(\mathsf{in}_{\func{f}}) = \mathsf{set}(\mathsf{out}_{\func{f}})$, 
then ${\func{f}}$ is of \emph{permutation type}.

A \emph{copy-and-delete-type} $\Sigma$-diagram is a pure wiring 
such that both $\mathsf{in}_{\func{f}}$ and $\mathsf{out}_{\func{f}}$
are sorted according to some (common, arbitrary) order on $X_{\func{f}}$, 
every variable in $\mathsf{out}_{\func{f}}$ is also in $\mathsf{in}_{\func{f}}$,
and $\mathsf{in}_{\func{f}}$ contains no duplicates.

Similarly, an \emph{equate-and-select-type} $\Sigma$-diagram is a pure wiring  
such that $\mathsf{in}_{\func{f}}$ and $\mathsf{out}_{\func{f}}$
are sorted according to some (common, arbitrary) order on $X_{\func{f}}$, 
and $\mathsf{out}_{\func{f}}$ contains no duplicates.
\end{definition}

Copy-and-delete-type dot diagrams are trivial to algebrise: 
they are a tensor product of $i$-fold copying dot diagrams, where each $i$ is the multiplicity of the corresponding variable in the list of outputs.
Algebraisations for equate-and-select-type dot diagrams are obtained in the same way. 
The case of algebrising permutation-type wirings is less obvious. 
In essence, it amounts to a parallel sorting algorithm that uses only adjacent swaps,
leading to the following lemma.

\begin{lemma}\label{lem:algebrise-permutation}
Let $\Sigma$ be a monoidal signature, 
let $\sigma$ be a permutation-type $\Sigma$-diagram without isolated variables,
and let $k := |\mathsf{in}_\sigma| = |\mathsf{out}_\sigma|$ be the number of inputs (or equivalently outputs) of $\sigma$.
Then there exists an algebraisation $t$ of $\sigma$ with
$|t| = O(k^2)$ and $\mathsf{width}(t) \leq 2k$.
Moreover, such an algebraisation can be computed in time $O(k^2)$.
\end{lemma}

To algebrise general pure wirings,
we need the following lemma and its corollary,
allowing us to write any pure wiring as a composition of
permutation-type, equate-and-select-type, and copy-and-delete-type wirings.
The lemma is proved by an explicit construction:
first, all occurrences of the same variable in the list of inputs are merged
(and variables occurring in the list of outputs only are created),
then every variable is copied as often as it occurs in the list of outputs.

\begin{lemma}\label{lem:decompose-pure-wiring}
Let ${\func{f}}$ be a pure wiring over $\Sigma$.
Then there exist
permutation-type $\Sigma$-diagrams $\sigma_1, \sigma_2, \sigma_3$,
an equate-and-select-type $\Sigma$-diagram $\mathsf{e}$,
and a copy-and-delete-type $\Sigma$-diagram $\mathsf{c}$
such that ${\func{f}} \equiv \sigma_3 \circ \mathsf{c} \circ \sigma_2 \circ \mathsf{e} \circ \sigma_1$.
\end{lemma}

\begin{corollary}\label{cor:algbrise-pure-wiring}
Let ${\func{f}}$ be a pure wiring over $\Sigma$ without isolated variables,
and let $k := \max \{|\mathsf{in}_{\func{f}}|, |\mathsf{out}_{\func{f}}|\}$.
Then there exists an algebraisation $t$ of ${\func{f}}$ with $|t| = O(k^2)$ and $\mathsf{width}(t) \leq 2k$.
In addition, such an algebraisation can be computed in time $O(k^2)$.
\end{corollary}

\subsection{Inefficient algebraisations} \label{ssec:badalg}

We give a method for obtaining algebraisations of general dot diagrams that are cleanly described, 
at the cost of resulting in a large width. 
Later, we minimise this damage by only applying this to small dot diagrams as part of our divide-and-conquer approach.

The width-inefficient algorithm is based on the observation that any dot diagram can be decomposed into a pure wiring, followed by a tensor product of function symbols and identities, followed by a pure wiring;
see \Cref{fig:frobenius-decomposition}.
This is the subject of the next lemma, which is essentially a consequence of~\cite[Proposition 5.4]{wilson2023data}, where an equivalent decomposition is called the \emph{Frobenius decomposition}.
The following are proven in \Cref{sec:proof-of-lem-wiring-tensor-factorisation,proof-of-cor-algebrise-small-decls}.

\begin{lemma}[Frobenius decomposition]\label{lem:wiring-tensor-factorisation}
Let $\Sigma$ be a monoidal signature 
and let ${\func{f}}$ be a $\Sigma$-diagram. 
For each $a\in A_{\func{f}}$, write $a = (\lett \mathbf{y}_a = g_a(\mathbf{x}_a))$, 
and let $\mathbf{S}_a,\mathbf{T}_a \in \Sigma_0^*$ be the domain and codomain of $g_a\colon \mathbf{S}_a \to \mathbf{T}_a$.
Moreover, let $\mathbf{z}$ be a list of the variables of ${\func{f}}$ that occur in no assignment, each listed once,
and let $\mathbf{U} := \mathsf{sort}_*(\mathbf{z})$ be its list of sorts.
Then there exist pure wirings $\mathsf{w}_1, \mathsf{w}_2$ such that
\[ {\func{f}} \equiv \mathsf{w}_1 \circ \Big(\mathsf{id}_{\mathbf{U}} \otimes \bigotimes_{a \in A_{\func{f}}} (g_a \otimes \mathsf{id}_{\mathbf{S}_a} \otimes \mathsf{id}_{\mathbf{T}_a})\Big) \circ  \mathsf{w}_2, \]
where $\equiv$ denotes equivalence of dot diagrams.
Moreover, if ${\func{f}}$ has no isolated variables, the wirings $\mathsf{w}_1, \mathsf{w}_2$ can be computed in $O((|A_{\func{f}}| + 1)\cdot \mathsf{maxar}({\func{f}}))$ time.
\end{lemma}

\begin{corollary}\label{cor:algebrise-small-decls}
Let ${\func{f}}$ be a $\Sigma$-diagram without isolated variables.
Let $a := \mathsf{maxar}({\func{f}})$ be the maximal arity of ${\func{f}}$ and let $n:= |A_{\func{f}}|\geq 1$.
Then there exists an algebraisation $t\in \mathrm{Term}(\Sigma)$ of ${\func{f}}$
such that $\mathsf{width}(t)\leq 4na + 2a$ and $|t| = O(n^2a^2)$.
Moreover, such a $t$ can be computed in $O(n^2a^2)$ time.
\end{corollary}

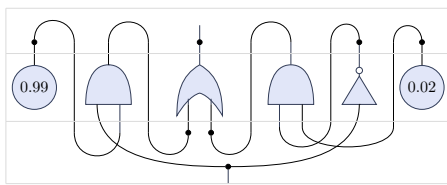
\begin{figure}
\centering
\scalebox{0.6}{
\begin{tikzpicture}
	\begin{pgfonlayer}{nodelayer}
		\node [style=none] (1) at (-4.5, 3.75) {};
		\node [style=none] (2) at (-3.5, 5.5) {};
		\node [style=none] (3) at (-2.5, 3.75) {};
		\node [style=none] (4) at (-3.5, 5.5) {};
		\node [style=none] (5) at (-3, 3.75) {};
		\node [style=none] (6) at (-4, 3.75) {};
		\node [style=none] (8) at (-3, 2.5) {};
		\node [style=none] (9) at (-0.5, 3.25) {};
		\node [style=none] (10) at (0.5, 5.5) {};
		\node [style=none] (11) at (1.5, 3.25) {};
		\node [style=none] (12) at (0.5, 4.25) {};
		\node [style=none] (13) at (0.5, 7.25) {};
		\node [style=none] (14) at (1, 4) {};
		\node [style=none] (15) at (0, 4) {};
		\node [style=none] (16) at (0, 3.75) {};
		\node [style=none] (17) at (1, 3.75) {};
		\node [style=none] (18) at (3.5, 3.75) {};
		\node [style=none] (19) at (4.5, 5.5) {};
		\node [style=none] (20) at (5.5, 3.75) {};
		\node [style=none] (21) at (4.5, 6.5) {};
		\node [style=none] (22) at (5, 3.75) {};
		\node [style=none] (23) at (4, 3.75) {};
		\node [style=none] (24) at (4, 2.75) {};
		\node [style=none] (25) at (9, 6.5) {};
		\node [style=none] (26) at (6.75, 3.75) {};
		\node [style=none] (27) at (7.5, 5) {};
		\node [style=none] (28) at (8.25, 3.75) {};
		\node [style=dot] (29) at (7.5, 6.5) {};
		\node [style=none] (30) at (7.5, 3.75) {};
		\node [style=light blue dot] (31) at (-6.75, 4.5) {$0.99$};
		\node [style=light blue dot] (32) at (10.25, 4.5) {$0.02$};
		\node [style=none] (33) at (-4, 3.75) {};
		\node [style=none] (34) at (7.5, 3.75) {};
		\node [style=none] (35) at (1.75, 0.25) {};
		\node [style=dot] (36) at (1.75, 1) {};
		\node [style=not dot] (43) at (7.5, 5.25) {};
		\node [style=dot] (44) at (0, 2.5) {};
		\node [style=dot] (45) at (1, 2.5) {};
		\node [style=dot] (46) at (10.25, 6.5) {};
		\node [style=dot] (47) at (-6.75, 6.5) {};
		\node [style=dot] (48) at (0.5, 6.5) {};
		\node [style=none] (49) at (-5, 6.5) {};
		\node [style=none] (50) at (-5, 2.5) {};
		\node [style=none] (51) at (-3, 2.5) {};
		\node [style=none] (53) at (-3.5, 6.5) {};
		\node [style=none] (54) at (-1.75, 6.5) {};
		\node [style=none] (55) at (-1.75, 2.5) {};
		\node [style=none] (56) at (2.75, 2.5) {};
		\node [style=none] (57) at (2.75, 6.5) {};
		\node [style=none] (58) at (5, 2.75) {};
		\node [style=none] (59) at (4, 2.75) {};
		\node [style=none] (60) at (5, 2.75) {};
		\node [style=none] (61) at (6.25, 3) {};
		\node [style=none] (62) at (6.25, 6.5) {};
		\node [style=none] (63) at (9, 2.75) {};
		\node [style=none] (64) at (10.25, 6.5) {};
		\node [style=none] (65) at (-8, 8) {};
		\node [style=none] (66) at (11.5, 8) {};
		\node [style=none] (67) at (11.5, 6) {};
		\node [style=none] (68) at (-8, 6) {};
		\node [style=none] (69) at (-8, 3) {};
		\node [style=none] (70) at (11.5, 3) {};
		\node [style=none] (71) at (-8, 0.25) {};
		\node [style=none] (72) at (11.5, 0.25) {};
	\end{pgfonlayer}
	\begin{pgfonlayer}{edgelayer}
		\draw [style=dark blue solid] (3.center)
			 to (1.center)
			 to [in=-180, out=90] (2.center)
			 to [in=90, out=0] cycle;
		\draw [style=dark blue solid] (4.center) to (2.center);
		\draw [style=dark blue solid] (5.center) to (8.center);
		\draw [style=dark blue solid] (11.center)
			 to [in=-30, out=90] (10.center)
			 to [in=90, out=210] (9.center)
			 to [in=195, out=60, looseness=0.75] (12.center)
			 to [in=120, out=-15, looseness=0.75] cycle;
		\draw [style=dark blue solid] (13.center) to (10.center);
		\draw [style=dark blue solid] (15.center) to (16.center);
		\draw [style=dark blue solid] (14.center) to (17.center);
		\draw [style=dark blue solid] (18.center)
			 to [in=-180, out=90] (19.center)
			 to [in=90, out=0] (20.center)
			 to cycle;
		\draw [style=dark blue solid] (21.center) to (19.center);
		\draw [style=dark blue solid] (23.center) to (24.center);
		\draw [style=dark blue solid] (28.center)
			 to (26.center)
			 to (27.center)
			 to cycle;
		\draw [style=dark blue solid] (29) to (27.center);
		\draw [style=black solid, in=-90, out=180] (36) to (33.center);
		\draw [style=black solid, in=-90, out=0] (36) to (34.center);
		\draw [style=dark blue solid] (36) to (35.center);
		\draw [style=black solid] (46) to (32);
		\draw [style=black solid, in=90, out=-90] (47) to (31);
		\draw [style=black solid] (30.center) to (34.center);
		\draw [style=black solid] (16.center) to (44);
		\draw [style=black solid] (45) to (17.center);
		\draw [in=90, out=90, looseness=1.75] (47) to (49.center);
		\draw (49.center) to (50.center);
		\draw [in=-90, out=-90, looseness=1.75] (50.center) to (51.center);
		\draw (4.center) to (53.center);
		\draw [in=90, out=90, looseness=1.75] (53.center) to (54.center);
		\draw (54.center) to (55.center);
		\draw [in=-90, out=-90, looseness=1.75] (55.center) to (44);
		\draw [in=-90, out=-90, looseness=1.75] (45) to (56.center);
		\draw (56.center) to (57.center);
		\draw [in=90, out=90, looseness=1.75] (57.center) to (21.center);
		\draw [in=90, out=-90, looseness=0.75] (24.center) to (59.center);
		\draw [in=-90, out=-90, looseness=0.75] (58.center) to (63.center);
		\draw [in=-90, out=-90, looseness=1.50] (59.center) to (61.center);
		\draw (61.center) to (62.center);
		\draw [in=90, out=90, looseness=2.00] (62.center) to (29);
		\draw (63.center) to (25.center);
		\draw [in=90, out=90, looseness=2.00] (25.center) to (64.center);
		\draw (22.center) to (60.center);
		\draw [style=lightgrey line] (66.center) to (67.center);
		\draw [style=lightgrey line] (67.center) to (70.center);
		\draw [style=lightgrey line] (70.center) to (72.center);
		\draw [style=lightgrey line] (72.center) to (71.center);
		\draw [style=lightgrey line] (71.center) to (69.center);
		\draw [style=lightgrey line] (69.center) to (70.center);
		\draw [style=lightgrey line] (65.center) to (68.center);
		\draw [style=lightgrey line] (68.center) to (67.center);
		\draw [style=lightgrey line] (65.center) to (66.center);
		\draw [style=lightgrey line] (68.center) to (69.center);
	\end{pgfonlayer}
\end{tikzpicture}
}
\caption{Frobenius decomposition of the dot diagram $\func{test}$
from \Cref{fig:code-and-diagram}: 
pure wiring,
followed by a tensor product of function symbols and identities,
followed by pure wiring.}
\label{fig:frobenius-decomposition}
\end{figure}

If every function symbol has at least one in- or output, the width resulting from the Frobenius decomposition is at least $n$ for a dot diagram with $n$ assignments, 
which is clearly far from optimal.
We will therefore only use it on small dot diagrams.

\subsection{Unfolding hierarchical dot diagrams} \label{ssec:unfold}

The next step of our divide-and-conquer approach is to refactor arbitrary dot diagrams into hierarchical dot diagrams, in which each dot diagram is small. To explain this, we first need the opposite notion: \emph{unfolding} a hierarchical dot diagram into a single large dot diagram via recursive substitution.

\begin{definition}
Let $\Sigma$ be a monoidal signature and let ${\func{f}}$ be a $\Sigma$-diagram.
Let $(g_i\colon \mathbf{S}_i \to \mathbf{T}_i)_{i \in I}\in \Sigma^I$ be a finite family of function symbols, and let $({\func{g}}_i\colon \mathbf{S}_i \to \mathbf{T}_i)_{i\in I}$ be a finite family of $\Sigma$-diagrams over the same index set.
The \emph{substitution} 
of $({\func{g}}_i)_{i\in I}$ for $(g_i)_{i \in I}\in \Sigma^I$ in ${\func{f}}$, written
\[ {\func{f}}[g_i \mapsto {\func{g}}_i]_{i\in I},\]
is defined as follows.
The set of variables of ${\func{f}}[g_i \mapsto {\func{g}}_i]_{i\in I}$ is the following pushout,
\[ X_{{\func{f}}\,[g_i \mapsto {\func{g}}_i]_{i\in I}} := \left(X_{{\func{f}}} \sqcup \bigsqcup_{i\in I} \bigsqcup_{\lett \mathbf{y} = g_i(\mathbf{x})}X_{{\func{g}}_i}\right)/\sim, \]
where $\sqcup$ denotes disjoint union, and $\sim$ is the equivalence relation that identifies the input (resp.~output) variables of each (copy of) ${\func{g}}_i$ with the input (resp.~output) variables of the corresponding assignment in which $g_i$ occurs.
The lists of in- and output variables of ${\func{f}}[g_i \mapsto {\func{g}}_i]_{i\in I}$ are the pushforwards of those of ${\func{f}}$ under the inclusion $X_{\func{f}}\to X_{{\func{f}}\,[g_i \mapsto {\func{g}}_i]_{i\in I}}$. 
Finally, the list of assignments of ${\func{f}}[g_i \mapsto {\func{g}}_i]_{i\in I}$ is given by
replacing each assignment of the form 
\[ \lett \mathbf{y} = g_i(\mathbf{x}) \]
in $A_{\func{f}}$ by the list $A_{{\func{g}}_i}$, with each variable in $A_{\func{f}}$ or $A_{{\func{g}}_i}$ replaced by its image in $X_{{\func{f}}[g_i \mapsto {\func{g}}_i]_{i\in I}}$.
\end{definition}

Unfolding a hierarchical $\Sigma$-diagram can now be defined by recursively substituting the definitions of all dot diagrams appearing therein along its call graph.

\begin{definition}
Let ${\func{f}}$ be a hierarchical $\Sigma$-diagram. 
Define 
\[
\mathsf{unfold}'({\func{f}}, v) := \begin{cases}
{\func{f}}_v \qquad\qquad\qquad\qquad\qquad\qquad\quad\;\;\;\text{if } v \text{ is a leaf,} \\
{\func{f}}_v[\mathrm{defn}_{{\func{f}}, v}^{-1}(u) \mapsto \mathsf{unfold}'({\func{f}}, u)]_{u \in \mathrm{ch}(v)} \;\text{ otherwise.}
\end{cases}
\]
\emph{Unfolding} ${\func{f}}$ is defined as $\mathsf{unfold}({\func{f}}) := \mathsf{unfold}'({\func{f}}, R_{\func{f}})$.
\end{definition}

To extend the algebraisation algorithm of \Cref{cor:algebrise-small-decls} to hierarchical dot diagrams, we need the fact that algebraisation commutes with substitution.

\begin{lemma}\label{lem:algebraisation-commutes-with-substitution}
Let ${\func{f}}$ be a $\Sigma$-diagram, 
let $(g_i:\mathbf{S}_i \to \mathbf{T}_i)_{i\in I}$ be a finite family of function symbols in $\Sigma$, 
and let $({\func{g}}_i: \mathbf{S}_i \to \mathbf{T}_i)_{i\in I}$ be a family of $\Sigma$-diagrams over the same index set.
Moreover, suppose that $t$ is an algebraisation of ${\func{f}}$ and that $s_i$ is an algebraisation of ${\func{g}}_i$, for each $i\in I$.
Then 
$ t[g_i \mapsto s_i]_{i\in I}$
is an algebraisation of ${\func{f}}[g_i \mapsto {\func{g}}_i]_{i\in I}$.
\end{lemma}
\begin{proof}
This can be shown by induction on the structure of $t$. 
In the base case, when $t$ consists of a single symbol, the claim is trivial.
If, on the other hand, $t= t_1 \circ t_2$ or $t= t_1 \otimes t_2$, this follows quickly using the induction hypothesis and the definitions of the composition and tensor product of $\Sigma$-diagrams.
\end{proof}

\begin{lemma}\label{lem:algebrise-hierarchical-function-declarations}
Let $\Sigma$ be a monoidal signature
and let ${\func{f}}$ be a hierarchical $\Sigma$-diagram such that no dot diagram appearing in ${\func{f}}$ has isolated variables.
Let $N \geq 1$ be such that every dot diagram appearing in ${\func{f}}$ has at most $N$ assignments,
and let $K := \max_{v\in V({\func{f}})} \mathsf{maxar}({\func{f}}_v)$ be the maximum of the maximal arities of the dot diagrams appearing in ${\func{f}}$.

Then there exists an algebraisation $t\in \mathrm{Term}(\Sigma)$
of $\mathsf{unfold}({\func{f}})$
with $|t| = O(|V({\func{f}})|N^2K^2)$ and  $\mathsf{width}(t)\leq 4NK + 2K$. Moreover, it can be computed in $O(|V({\func{f}})|N^2K^2)$ time.
\end{lemma}

\subsection{Algebraisation via branch decomposition} \label{ssec:branchdec}

To turn a dot diagram into a hierarchical dot diagram in which only 
``small'' auxiliary dot diagrams appear, 
we rely on the notion of a \emph{branch decomposition} of a hypergraph.
By a \emph{hypergraph} we mean an undirected 
hypergraph allowing for loops and multiple edges:

\begin{definition}
A \emph{hypergraph} $G$ consists of a set $V(G)$ of \emph{vertices}, 
a set $E(G)$ of \emph{hyperedges}, 
and a relation $I(G) \subseteq V(G) \times E(G)$. 
A vertex $v$ is \emph{incident} with a hyperedge $e$ if $(v, e) \in I(G)$.
\end{definition}

Dot diagrams have an associated hypergraph, the \emph{dependency hypergraph}, whose vertices are the variables of the dot diagram.

\begin{definition}
Let $\Sigma$ be a monoidal signature and let ${\func{f}}$ be a $\Sigma$-diagram. 
The \emph{dependency hypergraph} $\mathrm{Dep}({\func{f}})$ of ${\func{f}}$ is defined by
\[V(\mathrm{Dep}({\func{f}})) = X_{\func{f}}, \]
\[E(\mathrm{Dep}({\func{f}})) = \{1, \dots, |A_{\func{f}}|\} \sqcup \{\mathsf{in}, \mathsf{out}\},\]
i.e.~the hyperedges are the assignments of ${\func{f}}$ (counted with multiplicity), the list of inputs, and the list of outputs;
a vertex $x$ is incident with $e$ in $\mathrm{Dep}({\func{f}})$
if and only if either
\begin{enumerate}
\item $e = i$ and $x$ occurs in the $i$-th assignment $(\lett \mathbf{y} = g(\mathbf{x}))$ of ${\func{f}}$, i.e.~$x\in \mathsf{set}(\mathbf{x})$ or $x\in \mathsf{set}(\mathbf{y})$, or
\item $e = \mathsf{in}$ and $x\in \mathsf{set}(\mathsf{in}_{\func{f}})$, or
\item $e = \mathsf{out}$ and $x\in \mathsf{set}(\mathsf{out}_{\func{f}})$.
\end{enumerate}
\end{definition}

By an \emph{unrooted binary tree}, we mean a tree (i.e.~a connected acyclic undirected graph) in which each vertex has degree either 1 or 3. 
The notion of a branch decomposition of a hypergraph is due to Robertson and Seymour~\cite{robertson1991graph}.

\begin{definition}
Let $G$ be a hypergraph. 
A \emph{branch decomposition} of $G$ is an unrooted binary tree $\mathcal{B}$
together with a bijection $\eta$ from the leaves of $\mathcal{B}$ to the hyperedges of $G$.

Let $e$ be an edge of $\mathcal{B}$. 
Then removing $e$ from $\mathcal{B}$ partitions $\mathcal{B}$ into two components $\mathcal{B}_1$, $\mathcal{B}_2$.
A vertex $v$ of $G$ is \emph{at the interface of $e$} if there is a leaf $l_1$ of $\mathcal{B}_1$ and a leaf $l_2$ of $\mathcal{B}_2$ such that $v$ is incident with both $\eta(l_1)$ and $\eta(l_2)$.
The \emph{order} of an edge $e$ of $\mathcal{B}$ is the number of vertices in $G$
at the interface of $e$.

The \emph{width} $\mathsf{width}(\mathcal{B})$ of a branch decomposition $\mathcal{B}$ of $G$ is the maximum order of any edge in $\mathcal{B}$.
The \emph{branch-width} of $G$ is the minimum width of any branch decomposition of $G$.
\end{definition}

\begin{figure}
\centering
\scalebox{0.6}{
\input{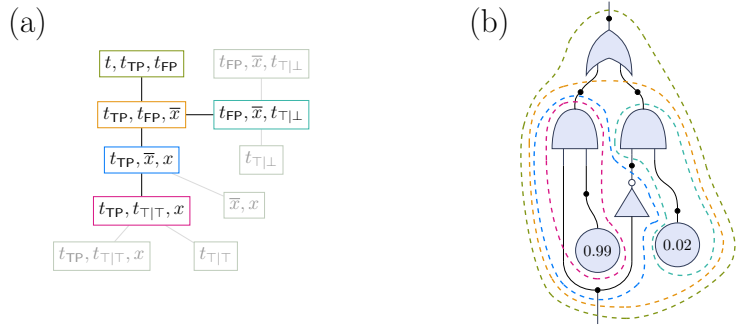}
}
\caption{(a) A tree decomposition of $\func{test}'$ extended by redundant (light grey) bags to turn it into a branch decomposition of $\mathrm{Dep}(\func{test})$;
(b) branch decomposition visualised as a hierarchical clustering of the boxes of $\func{test}$.}
\label{fig:branch-decomposition-example}
\end{figure}

A branch decomposition of the dependency hypergraph of a dot diagram can be visualised as a hierarchical clustering of its assignments (``boxes'');
see \Cref{fig:branch-decomposition-example}.
The next lemma turns a branch decomposition of the dependency hypergraph
into a hierarchical dot diagram with the desired properties.

\begin{lemma}\label{lem:decomp-to-hierarchical-decl}
Let ${\func{f}}$ be a $\Sigma$-diagram without isolated variables
and let $\mathcal{B}$ be a branch decomposition of the dependency hypergraph of ${\func{f}}$.
Then there exists a hierarchical $\Sigma$-diagram ${\func{f}}^*$ such that:
\begin{bracketenumerate}
\item $\mathsf{unfold}({\func{f}}^*)\equiv {\func{f}}$.
\item Every dot diagram appearing in ${\func{f}}^*$ has maximal arity at most $2\,\mathsf{width}(\mathcal{B}) + \mathsf{maxar}({\func{f}})$.
\item Every dot diagram appearing in ${\func{f}}^*$ has at most two assignments and no isolated variables.
\item $|V({\func{f}}^*)| \leq |A_{\func{f}}| + 1$.
\end{bracketenumerate}
Moreover, such an ${\func{f}}^*$ can be computed in $O(|V(\mathcal{B})|\cdot (\mathsf{width}(\mathcal{B}) + \mathsf{maxar}({\func{f}})))$ time.
\end{lemma}

Once we have converted a dot diagram into a hierarchical dot diagram in which only ``small'' dot diagrams appear,
as provided by \Cref{lem:decomp-to-hierarchical-decl},
we can apply \Cref{lem:algebrise-hierarchical-function-declarations} to obtain an algebraisation of small width.
Hence, taking both lemmas together, with $N := 2$, directly implies:

\begin{theorem}\label{thm:algebraisation-via-branch-decomp}
Let ${\func{f}}$ be a $\Sigma$-diagram without isolated variables
and let $\mathcal{B}$ be a branch decomposition of the dependency hypergraph of ${\func{f}}$. 
Let $K := 2\,\mathsf{width}(\mathcal{B}) + \mathsf{maxar}({\func{f}})$.
Then there exists an algebraisation $t$ of ${\func{f}}$ with $\mathsf{width}(t) \leq 10 K$
and $|t| = O((|A_{\func{f}}| + 1)\cdot K^2)$.
Moreover, given the branch decomposition $\mathcal{B}$,
such an algebraisation $t$ can be computed in time $O((|A_{\func{f}}| + 1)\cdot K^2)$.
\end{theorem}

\subsection{Proof of \Cref{thm:algebraisation}} \label{ssec:proof}

\begin{proof}
We now show \Cref{thm:algebraisation} by reducing it to \Cref{thm:algebraisation-via-branch-decomp}.
Since ${\func{f}}$ has no isolated variables, every variable of ${\func{f}}$ occurs in $\mathsf{in}_{\func{f}}$, in $\mathsf{out}_{\func{f}}$ or in some assignment, so that the primal graph ${\func{f}}'$ has at most $na$ vertices.
Hence, by a theorem of Korhonen~\cite[Theorem 1.1]{korhonen2023single}, one can compute a tree decomposition $\mathcal{T}$ of width at most $2 k + 1$
of ${\func{f}}'$ in time $O(na\cdot 2^{O(k)})$.
Moreover, we may assume without loss of generality that
$|V(\mathcal{T})|$ is at most the number of vertices of ${\func{f}}'$.
One can convert the tree decomposition $\mathcal{T}$ of ${\func{f}}'$ into a branch decomposition $\mathcal{B}$ of the dependency hypergraph $\mathrm{Dep}({\func{f}})$
of width $\mathsf{width}(\mathcal{B}) \leq \mathsf{width}(\mathcal{T}) + 1$, in $O(n(k + a))$ time.
This construction is carried out in the proof of~\cite[(5.1)]{robertson1991graph}.
The notion of treewidth used there is equivalent to the treewidth of the primal graph;
see also \Cref{fig:branch-decomposition-example}.
Using \Cref{thm:algebraisation-via-branch-decomp}
we find an algebraisation $t$ of width
\[ \mathsf{width}(t) \leq 10\cdot(2\,\mathsf{width}(\mathcal{B}) + \mathsf{maxar}({\func{f}})) \leq 10\cdot(4k + 4 + a) \in O(k + a),\]
and $|t| = O(n (k + a)^2)$ in time $O(n(k+a)^2)$.
In total, this takes $O(na\cdot 2^{O(k)} + n(k+a)^2) \subseteq O(2^{O(k)} n a^2)$ time.
\end{proof}

\subsection{Generalisation to the hierarchical case}

\Cref{thm:algebraisation} can now be easily generalised to hierarchical dot diagrams ${\func{f}}$:
when the primal graph of each dot diagram appearing in ${\func{f}}$ has small treewidth (and the maximal arities are small), 
then an algebraisation of the unfolded dot diagram $\mathsf{unfold}({\func{f}})$ of small width and DAG-size can be computed efficiently.
This is the subject of the following theorem,
which is proved in \Cref{sec:proof-of-thm-algebraisation-general} by combining \Cref{thm:algebraisation} with \Cref{lem:substitute-algebraisations}, which combines the algebraisations of the dot diagrams appearing in ${\func{f}}$ into an algebraisation of $\mathsf{unfold}({\func{f}})$.

\begin{theorem}\label{thm:algebraisation-general}
Let ${\func{f}}$ be a hierarchical dot diagram over a monoidal signature $\Sigma$ such that no dot diagram appearing in ${\func{f}}$ has isolated variables.
Moreover, let 
\begin{enumerate}
\item $k := \max_{v \in V({\func{f}})} \mathsf{tw}({\func{f}}_v')$ be the maximum treewidth of the primal graph of any ${\func{f}}_v$,
\item $a := \max_{v \in V({\func{f}})} \mathsf{maxar}({\func{f}}_v)$ be the maximal arity of any ${\func{f}}_v$,
\item $L:= 1 + \max_{v\in V({\func{f}})}|\mathrm{ch}(v)|$
be one plus the maximum number of children of any vertex in the call graph of ${\func{f}}$,
\item $M := |V({\func{f}})|$, and let
\item $N := 1 + \max_{v\in V({\func{f}})} |A_{{\func{f}}_v}|$ be one plus the maximum number of assignments in any dot diagram appearing in ${\func{f}}$.
\end{enumerate}
Then an algebraisation of $\mathsf{unfold}({\func{f}})$ of DAG-size in 
$O(M N (k+a)^2)$ and width in $O(k + a)$ can be computed in $O(L M N \cdot 2^{O(k)} a^2)$ time.
\end{theorem}

\section{Applications of algebraisation}\label{sec:applications}

With nice algebraisations in place, we now return to the question of computing semantics; this includes DPP inference. 
Of course, the precise complexity depends highly on the strict hypergraph category corresponding to the semantics. 
We first discuss general properties of semantics in matrix categories, and then apply them to Problem~\ref{problem:boolean-probabilistic-inference}. Finally, we discuss other existing problems that can be viewed as the computation of dot diagram semantics.

\subsection{Computing matrix semantics}

To cover many of the applications that follow in a unified manner, 
we first consider the problem of computing the semantics of a hierarchical dot diagram in a matrix category over a semiring.
The semantics of a hierarchical dot diagram are defined as the semantics of the dot diagram that results from unfolding it.

\begin{definition}\label{definition:semantics-of-hierarchical-dot-diagrams}
Let ${\func{f}}$ be a hierarchical dot diagram over some monoidal signature $\Sigma$, 
and let $\mathcal{I}\colon \Sigma \to C$ be an interpretation of $\Sigma$ in some strict hypergraph category $C$.
The \emph{semantics} of ${\func{f}}$ under $\mathcal{I}$ are $[\![{\func{f}}]\!]_{\mathcal{I}} := [\![\mathsf{unfold}({\func{f}})]\!]_{\mathcal{I}}$.
\end{definition}

When working over general semirings, 
a natural way to express the complexity of an algorithm is the number of additions and multiplications it takes. 
This can be made precise using the notion of an \emph{arithmetic circuit} over a semiring: 
a representation of an arithmetic expression as a directed acyclic graph.
The problem then becomes to \emph{compile} a hierarchical dot diagram ${\func{f}}$ to an equivalent arithmetic circuit. 
\Cref{thm:semantics-in-matrix-categories} shows that when, for each dot diagram appearing in ${\func{f}}$, 
the treewidth of its primal graph and its maximal arity are bounded, 
this problem can be solved efficiently.

\begin{definition}
Let $R$ be a semiring. 
An \emph{arithmetic circuit} over $R$ is a tuple $C = (V(C), E(C), \left(\mathrm{op}^C_v\right)_{v \in V(C)})$
such that:
\begin{enumerate}
\item $(V(C), E(C))$ is a directed acyclic graph.
\item For each vertex $v\in V(C)$, 
the \emph{label} $\mathrm{op}^C_v$ of $v$ is either an element of $R$, or one of the two symbols $\{+,\,\cdot\,\}$.
\item A vertex $v\in V(C)$ is a leaf in $(V(C), E(C))$ if and only if it is labelled by an element of $R$, i.e.~$\mathrm{op}^C_v \in R$.
\item All non-leaf vertices have exactly two children.
\end{enumerate}
The \emph{semantics} of $C$ at $v\in V(C)$ are given by 
\[
[\![C]\!]_v := \begin{cases}
\mathrm{op}^C_v &\text{ if } v \text{ is a leaf,} \\
\sum_{u\in \mathrm{ch}(v)} [\![C]\!]_u &\text{ if } \mathrm{op}^C_v = +, \\
\prod_{u\in \mathrm{ch}(v)} [\![C]\!]_u &\text{ otherwise.} 
\end{cases}
\]
\end{definition}

Note that in contrast to other notions of \emph{arithmetic circuit} in the literature~\cite{shpilka2010arithmetic}, 
the arithmetic circuits we consider may have multiple outputs (i.e.~roots).
Moreover, as defined above, they only represent \emph{closed} arithmetic expressions (i.e.~expressions which do not contain variables).
The following theorem characterises a general algorithm for computing the semantics of dot diagrams in matrix hypergraph categories; its proof is given in \Cref{sec:proof-of-thm-semantics-in-matrix-categories}.

\begin{theorem}\label{thm:semantics-in-matrix-categories}
Let $\Sigma$ be a monoidal signature, 
let $R$ be a semiring, 
and let $\mathcal{I}\colon \Sigma \to \mathsf{Mat}_R$ 
be an interpretation. 
Moreover, let ${\func{f}}$ be a hierarchical dot diagram over $\Sigma$ such that no dot diagram appearing in ${\func{f}}$ has isolated variables,
and let $k$, $a$, $L$, $M$ and $N$ be defined as in \Cref{thm:algebraisation-general}.
Finally, let $m, n\in \mathbb{N}$ be the dimensions of the matrix 
$[\![{\func{f}}]\!]_{\mathcal{I}} \in R^{m\times n}$.

Then there is an arithmetic circuit $C$ over $R$ with $O(2^{O(k + a)} M N)$ vertices 
and a bijection
$\nu$ from $\{1, \dots, m\} \times \{1, \dots, n\}$ to the set of roots of $C$
such that
\[ ([\![{\func{f}}]\!]_{\mathcal{I}})_{ij} = [\![C]\!]_{\nu(i, j)} \]
for all $i\in \{1, \dots, m\}$ and $j \in \{1, \dots, n\}$.
Moreover, this arithmetic circuit $C$ can be computed in $O(2^{O(k + a)}LMN)$ time.
The constants hidden in the $O$-notation depend only on the dimensions $\mathcal{I}(S)$ of the sorts $S\in\Sigma_0$.

In particular, if the arithmetic operations in $R$ can be computed in constant time and the parameters $k$ and $a$ are bounded, 
then the semantics of ${\func{f}}$ can be computed in $O(LMN)$ time.
\end{theorem}

\subsection{Probabilistic inference} \label{ssec:inference}

We now consider DPPs specifically.
Recall that a matrix over $\mathbb{Q}_{\geq 0}$ is substochastic if each of its columns sums to at most $1$.

\begin{lemma}\label{lem:dpp-substochastic}
Let ${\func{f}}$ be a discrete probabilistic program.
Then its semantics $[\![{\func{f}}]\!]_{\mathcal{S}}$ under the substochastic matrix interpretation is a substochastic matrix.
\end{lemma}

The proof, given in \Cref{sec:proof-of-lem-dpp-substochastic}, rests on the fact that every gs-monoidal dot diagram has an algebraisation without $\mathrm{equate}$ and $\mathrm{new}$,
which follows from the universal property of Fritz and Liang mentioned in \Cref{remark:gs-monoidal}.
Hence unfolding a hierarchical gs-monoidal dot diagram yields a gs-monoidal dot diagram,
and the semantics of a gs-monoidal dot diagram under an interpretation by substochastic matrices is substochastic, since substochastic matrices are closed under products and Kronecker products.
In particular, the acceptance probability defined next is at most $1$.

\begin{definition}\label{def:boolean-probabilistic-inference}
Let ${\func{f}}\colon (\,) \to (\mathbb{B})$ be a discrete probabilistic program
with no inputs and a single output.
Let $\begin{pmatrix} q & p\end{pmatrix}^\intercal := [\![ {\func{f}}]\!]_{\mathcal{S}}$ be the semantics of ${\func{f}}$ under the substochastic interpretation.
The \emph{acceptance probability} of ${\func{f}}$ is $p_{\mathrm{acc}}:= p + q$.
The \emph{probability that ${\func{f}}$ outputs true} is $p_{\func{f}} := p/p_{\mathrm{acc}}$ if $p_{\mathrm{acc}} > 0$ and $0$ otherwise.
\end{definition}

This leads to the following rephrasing of Problem~\ref{problem:boolean-probabilistic-inference}:

\begin{problem}
Given a discrete probabilistic program $\func{f}$ and a natural number $d>0$, 
approximate $p_{\func{f}}$ with an absolute error of at most $2^{-(d+1)}$.
\end{problem}

Next is our main result concerning inference in DPPs.
Its proof is mainly an application of \Cref{thm:semantics-in-matrix-categories},
together with the observation that all intermediate results of the resulting computation lie in $[0,1]$,
since every sub-term of an algebraisation of a gs-monoidal dot diagram arises from a gs-monoidal dot diagram by bending some of its wires,
and standard arguments for bounding the error of truncated binary arithmetic; see \Cref{sec:proof-of-thm-fixed-parameter-tractable-inference}.

\begin{theorem}\label{thm:fixed-parameter-tractable-inference}
Let ${\func{f}}\colon (\,) \to (\mathbb{B})$ be a discrete probabilistic program
with no inputs and a single output,
let $p_{\mathrm{acc}} > 0$ be the acceptance probability of ${\func{f}}$,
and let $d>0$ be an integer.
Moreover, let $L$, $M$, and $N$ be defined as in \Cref{thm:algebraisation-general},
and let $k := \max_{v\in V({\func{f}})} \mathsf{tw}({\func{f}}_v')$ be the maximum treewidth of the primal graph of any dot diagram appearing in ${\func{f}}$.

Then the probability $p_{\func{f}}$ that ${\func{f}}$ outputs \emph{true} can be approximated with absolute error at most $2^{-(d+1)}$
in time
\[
O((LMN)^3 \cdot  d^2\cdot \log_2^3(2p_{\mathrm{acc}}^{-1})\cdot 2^{O(k)}).
\]

In particular, if $k$ is bounded and the inverse acceptance probability $p_{\mathrm{acc}}^{-1}$ is at most exponential in $LMN$, then Boolean probabilistic inference can be performed in polynomial time.
\end{theorem}

A problem is said to be \emph{fixed-parameter tractable}  
if on every class of instances for which a given parameter is bounded, it can be solved in polynomial time. 
Hence, according to \Cref{thm:fixed-parameter-tractable-inference}, 
Boolean probabilistic inference is fixed-parameter tractable 
for the parameter $(k, p_{\mathrm{acc}}^{-1})$. 
The dependency on $p_{\mathrm{acc}}^{-1}$ can be removed
if we assume that the call graph of the given probabilistic program has bounded depth $D$:
in this case, $\log_2(p_{\mathrm{acc}}^{-1}) \leq N^{D+1}\cdot \log_2(p_{\min}^{-1})$,
where $p_{\min}$ is the smallest of the numbers $p$ and $1-p$ over all $\func{flip}(p)$ with $0 < p < 1$ occurring in ${\func{f}}$
(see \Cref{lem:acceptance-probability-bound}).
Since $\log_2(p_{\min}^{-1})$ is at most the number of bits needed to write down the probabilities occurring in ${\func{f}}$,
Boolean probabilistic inference takes polynomial time in $d$ and the size of the program (including its probabilities) when $k$ and $D$ are bounded.
This also shows that \Cref{thm:fixed-parameter-tractable-inference} can be seen as a proper generalisation of the fixed-parameter tractability of Bayesian network inference.

\begin{remark}
Even though \Cref{thm:fixed-parameter-tractable-inference} generalises the fixed-parameter tractability of Bayesian network inference, 
our inference algorithm does \emph{not} reduce to the classical junction tree algorithm for Bayesian networks. 
One significant difference is that the junction tree algorithm operates directly 
on a tree decomposition of the moral graph (which is the same as the primal graph of a corresponding DPP).
In contrast, our algorithm first converts this tree decomposition into a branch decomposition, 
making our algorithm more similar to (but still distinct from) tensor network contraction algorithms. 
\end{remark}

\enlargethispage{1.6\baselineskip}
\subsection{Further applications}

Being phrased in terms of dot diagrams rather than DPPs specifically, \Cref{thm:algebraisation} can also be applied to derive a number of fixed-parameter tractability results in a unified manner.

\paragraph*{Tensor networks and quantum circuit simulation} Tensor networks can be viewed as dot diagrams over $\mathsf{Mat}_R$ (viewed as a monoidal signature), 
and computing the semantics of such a dot diagram corresponds precisely to tensor network contraction~\cite[Theorem 5.1]{kissinger2014finite}.
Hence, the known result that tensor network contraction can be performed using polynomially many arithmetic operations for networks of bounded treewidth~\cite{markov2008simulating}, e.g.~for the purpose of quantum circuit simulation, can be viewed as a special case of \Cref{thm:semantics-in-matrix-categories}.  

\paragraph*{Computing attack tree metrics} As was shown in~\cite{peterseim2025unified},
the problem of computing \emph{attack tree metrics} \cite{lopuhaa2022efficient} can naturally 
be formulated as the problem of computing the semantics of a string diagram, 
the given \emph{attack tree},
in $\mathsf{Mat}_R$ for suitable semirings $R$. 
In particular, the metric associated to the tropical semiring $(\mathbb{N}, \min, +)$, 
the \emph{min.~cost metric},
quantifies the minimum cost for an attacker to compromise the system modelled by the attack tree.
A direct application of \Cref{thm:semantics-in-matrix-categories} now gives that attack tree metrics can be computed in polynomial time on attack trees of bounded treewidth, 
as long as the size of any intermediate result of evaluating the corresponding arithmetic circuit is bounded.
The latter condition applies, for instance, in the case of the min.~cost metric. 
For this problem, no fixed-parameter tractable algorithms have previously been described in the literature.

\paragraph*{Evaluating database queries} In the context of relational databases, it is known that conjunctive query evaluation is fixed-parameter tractable~\cite{chekuri2000conjunctive,grohe2001evaluation}. 
One can view this as a consequence of \Cref{thm:algebraisation}, as follows.
A conjunctive query (i.e.~a select-project-join query) on a database can be viewed as a dot diagram over the strict hypergraph category $\mathsf{FinRel}$ of finite relations (viewed as a monoidal signature).\footnote{This is explained in more detail in \Cref{sec:queries-as-dot-diagrams}.}
Here, $\mathsf{FinRel}$ can be defined as the matrix hypergraph category $\mathsf{Mat}_{\mathbb{B}}$, 
where $(\mathbb{B}, \lor, \land)$ is the Boolean semiring. 
In the context of databases, however, one should think of the morphisms of $\mathsf{FinRel}$ not as matrices with Boolean entries, but as tables that directly list the elements of a relation.
The problem of evaluating a conjunctive query is then equivalent to computing the semantics of the corresponding dot diagram in $\mathsf{FinRel}$.
The primal graph of the dot diagram corresponds to the \emph{query graph} of the conjunctive query.
Therefore, it follows from \Cref{thm:algebraisation} that a conjunctive query whose query graph has bounded treewidth (and whose relation symbols have bounded arity) can be evaluated in polynomial time in the size of the query and the largest table in the database.
In contrast to the case of tensor networks and attack trees, 
this result does not use the compilation of string diagrams to arithmetic circuits, 
but directly uses the cartesian product and composition of relations in the evaluation step.
This shows that the additional layer of abstraction provided by dot diagrams and hypergraph categories can be necessary for some applications. {}
\section{Conclusion}

\Cref{thm:fixed-parameter-tractable-inference} shows that despite the likely inherent intractability of inference for \emph{general} discrete probabilistic programs, 
inference can nevertheless be fast on sufficiently structurally simple instances.
Moreover, existing inference engines for discrete probabilistic programs do not satisfy our parameterised complexity bounds,
and for deeply nested programs our algorithm yields exponential \mbox{asymptotic} improvements (see \Cref{fig:counterexample}).

\enlargethispage{\baselineskip}
What remains open is whether, and to what extent, our method also yields benefits in practical implementations.
This next step will also be an opportunity to investigate potential advantages of our approach that are difficult to analyse on the level of asymptotic complexity. 
One such potential advantage is that the algebraic representation of probabilistic programs enables the use of traditional methods 
for term rewriting and program optimisation to simplify terms before evaluating them, 
making use of the sound and complete equational theory for discrete probabilistic programs from~\cite{piedeleu2025complete}.
Another possible advantage is that our abstract formulation allows dense matrices to be easily replaced by more efficient representations of substochastic matrices, 
similar to decision diagrams used in quantum circuit simulation~\cite{wille2022decision}, 
without affecting the stated performance guarantees.

Finally, we have shown that our method, \emph{string diagram algebraisation}, is applicable beyond probabilistic programming, 
enabling a unified view on a diverse variety of known fixed-parameter tractability results. {}

\bibliography{p076-Peterseim}

\appendix

\section{The monoidal category of dot diagrams}\label{appendix:tensor-product-and-composition-of-dot-diagrams}

In this section,
we provide precise definitions of the composite and tensor product of dot diagrams,
omitted from the main text.
Up to equivalence of dot diagrams, these can easily be seen to be equivalent to the definitions given in \cite[Definition 4.5]{kissinger2014finite},
where dot diagrams are defined in terms of cospans of certain kinds of labelled hypergraphs.
In both definitions that follow, let $\Sigma$ be a monoidal signature,
and let $\func{f}: \mathbf{S} \to \mathbf{T}$
and $\func{g}: \mathbf{T} \to \mathbf{U}$ be $\Sigma$-diagrams.

\begin{definition}\label{definition:composite-of-dot-diagrams}
    The \emph{composite} of $\func{f}$ and $\func{g}$ is the dot diagram
    $\func{g} \circ \func{f}$
    over $\Sigma$, where:
    \begin{enumerate}
        \item The set of variables of $\func{g} \circ \func{f}$ is
        $$X_{\func{g} \circ \func{f}} := (X_{\func{f}} \sqcup X_{\func{g}}) / \sim,$$
        where $\sqcup$ denotes disjoint union, and $\sim$ is the equivalence relation generated by identifying the $i$-th output variable of $\func{f}$ with the $i$-th input variable of $\func{g}$, for all $i \in \{1, \dots, |\mathbf{T}|\}$.
        \item The list of assignments of $\func{g} \circ \func{f}$ is
        \[
        A_{\func{g} \circ \func{f}} := (\underline{\mathsf{inl}})_*(A_{\func{f}})\,(\underline{\mathsf{inr}})_*(A_{\func{g}}),
        \]
        where $\underline{\mathsf{inl}} := p \circ \mathsf{inl}$ and $\underline{\mathsf{inr}} := p \circ \mathsf{inr}$,
        with
        $\mathsf{inl}\colon X_{\func{f}} \to X_{\func{f}} \sqcup X_{\func{g}}$ and $\mathsf{inr}\colon X_{\func{g}} \to X_{\func{f}} \sqcup X_{\func{g}}$
        given by the canonical inclusion maps,
        and $p: X_{\func{f}} \sqcup X_{\func{g}} \to X_{\func{g} \circ \func{f}}$ given by the canonical projection.
        \item The list of inputs is $(\underline{\mathsf{inl}})_*(\mathsf{in}_{\func{f}})$,
        and the list of outputs is $(\underline{\mathsf{inr}})_*(\mathsf{out}_{\func{g}})$.
    \end{enumerate}
\end{definition}

\begin{definition}\label{definition:tensor-product-of-dot-diagrams}
    The \emph{tensor product} of $\func{f}$ and $\func{g}$ is the $\Sigma$-diagram
    $\func{f} \otimes \func{g}$, where:
    \begin{enumerate}
        \item The set of variables of $\func{f} \otimes \func{g}$ is
        $$X_{\func{f} \otimes \func{g}} := X_{\func{f}} \sqcup X_{\func{g}}.$$
        \item The list of assignments of $\func{f} \otimes \func{g}$ is
        \[
        A_{\func{f} \otimes \func{g}} := \mathsf{inl}_*(A_{\func{f}})\,\mathsf{inr}_*(A_{\func{g}}).
        \]
        \item The list of inputs is $\mathsf{inl}_*(\mathsf{in}_{\func{f}})\, \mathsf{inr}_*(\mathsf{in}_{\func{g}})$;
        the list of outputs is $\mathsf{inl}_*(\mathsf{out}_{\func{f}})\, \mathsf{inr}_*(\mathsf{out}_{\func{g}})$.
    \end{enumerate}
\end{definition}

Both operations are compatible with the equivalence of dot diagrams,
and hence descend to $\underline{\mathrm{Dot}}(\Sigma)$.

\section{Proofs for \Cref{sec:algebraisation}}\label{appendix:proofs-algebraisation}

\subsection{Proof of \Cref{lem:algebrise-permutation}}\label{sec:proof-of-lem-algebraise-permutation}

Run an odd-even transposition sort on $\sigma$,
recording the transpositions performed in each round.
This takes at most $2k$ rounds,
and the transpositions of one round are pairwise disjoint transpositions of adjacent positions.
Hence the permutation performed in one round is represented by a tensor product of at most $k$ terms $\mathrm{id}$ and $\mathrm{swap}$,
which has DAG-size $O(k)$ and width at most $2k$.
The composite of these at most $2k$ terms is an algebraisation of $\sigma$.
Its DAG-size is $O(k^2)$, and running the sort and writing down the composite takes $O(k^2)$ time, since each of the at most $2k$ rounds takes $O(k)$ time;
its width is at most $2k$, since every sub-term has at most $k$ inputs and at most $k$ outputs. \qed

\subsection{Proof of \Cref{lem:decompose-pure-wiring}}\label{sec:proof-of-lem-decompose-pure-wiring}

Let $X := X_{\func{f}}$, $\mathbf{in} := \mathsf{in}_{\func{f}}$, $\mathbf{out} := \mathsf{out}_{\func{f}}$, $m := |\mathbf{in}|$ and $n := |\mathbf{out}|$.
Fix a total order on $X$.
For a list $\mathbf{l}$ of variables, let $\mathsf{sort}(\mathbf{l})$ be the list obtained by sorting $\mathbf{l}$ with respect to this order (keeping duplicates),
and let $\mathsf{dist}(\mathbf{l})$ be the sorted list of the distinct entries of $\mathbf{l}$.
Choose a permutation $\tau$ of $\{1, \dots, m\}$ such that $\mathbf{in}(\tau(i))$ is the $i$-th entry of $\mathsf{sort}(\mathbf{in})$ for all $i$,
and a permutation $\rho$ of $\{1, \dots, n\}$ such that the $\rho(j)$-th entry of $\mathsf{sort}(\mathbf{out})$ is $\mathbf{out}(j)$ for all $j$.
The five factors have no assignments and are defined as follows.
\begin{itemize}\sloppy
    \item $\sigma_1$ has variable set $\{p_1, \dots, p_m\}$, list of inputs $(p_1, \dots, p_m)$ and list of outputs $(p_{\tau(1)}, \dots, p_{\tau(m)})$.
    \item $\mathsf{e}$ has variable set $X$, list of inputs $\mathsf{sort}(\mathbf{in})$ and list of outputs $\mathsf{dist}(\mathbf{out})$.
    \item $\sigma_2$ is the identity: it has variable set $\{r_1, \dots, r_{n'}\}$, where $n' := |\mathsf{dist}(\mathbf{out})|$, and both its list of inputs and its list of outputs are $(r_1, \dots, r_{n'})$.
    \item $\mathsf{c}$ has variable set $\mathsf{set}(\mathbf{out})$, list of inputs $\mathsf{dist}(\mathbf{out})$ and list of outputs $\mathsf{sort}(\mathbf{out})$.
    \item $\sigma_3$ has variable set $\{q_1, \dots, q_n\}$, list of inputs $(q_1, \dots, q_n)$ and list of outputs $(q_{\rho(1)}, \dots, q_{\rho(n)})$.
\end{itemize}
By construction, $\sigma_1, \sigma_2, \sigma_3$ are of permutation type,
$\mathsf{e}$ is of equate-and-select type, and $\mathsf{c}$ is of copy-and-delete type.

By \Cref{definition:composite-of-dot-diagrams}, the variables of the composite $\sigma_3 \circ \mathsf{c} \circ \sigma_2 \circ \mathsf{e} \circ \sigma_1$
are the classes of the variables of the five factors under the equivalence relation generated by the following identifications:
$p_{\tau(i)}$ with the $i$-th entry of $\mathsf{sort}(\mathbf{in})$, which is $\mathbf{in}(\tau(i))$;
the $j$-th entry of $\mathsf{dist}(\mathbf{out})$ as a variable of $\mathsf{e}$ with $r_j$, and $r_j$ with the same entry as a variable of $\mathsf{c}$;
and the $j$-th entry of $\mathsf{sort}(\mathbf{out})$ with $q_j$.
Hence each variable of $\mathsf{c}$ is identified with the equally named variable of $\mathsf{e}$,
each $p_i$ with $\mathbf{in}(i)$, and each $q_j$ with $\mathsf{sort}(\mathbf{out})(j)$,
and there are no further identifications;
so the variables of the composite are in bijection with $X$.
Under this bijection, the list of inputs $(p_1, \dots, p_m)$ of the composite corresponds to $\mathbf{in}$,
and its list of outputs $(q_{\rho(1)}, \dots, q_{\rho(n)})$ corresponds to $(\mathsf{sort}(\mathbf{out})(\rho(1)), \dots, \mathsf{sort}(\mathbf{out})(\rho(n))) = \mathbf{out}$.
Since neither the composite nor ${\func{f}}$ has assignments, this shows ${\func{f}} \equiv \sigma_3 \circ \mathsf{c} \circ \sigma_2 \circ \mathsf{e} \circ \sigma_1$. \qed

\subsection{Proof of \Cref{cor:algbrise-pure-wiring}}\label{sec:proof-of-algbrise-pure-wiring}

Write ${\func{f}} \equiv \sigma_3 \circ \mathsf{c} \circ \sigma_2 \circ \mathsf{e} \circ \sigma_1$ with the five factors constructed in the proof of \Cref{lem:decompose-pure-wiring},
and algebrise each factor separately.
The permutation-type factors $\sigma_1, \sigma_2, \sigma_3$ have at most $k$ inputs and are algebrised by \Cref{lem:algebrise-permutation}.
The factor $\mathsf{e}$ is algebrised by the tensor product, over the variables $x$ of ${\func{f}}$ in the fixed order,
of a term $E_x$ with $s_x$ inputs and $r_x$ outputs,
where $s_x$ is the number of occurrences of $x$ in $\mathsf{in}_{\func{f}}$, and $r_x \in \{0, 1\}$ indicates whether $x$ occurs in $\mathsf{out}_{\func{f}}$:
if $r_x = 1$, then $E_x$ is an $s_x$-fold equate (which is $\mathrm{new}$ for $s_x = 0$ and $\mathrm{id}$ for $s_x = 1$),
and if $r_x = 0$, then $E_x$ is an $s_x$-fold equate followed by $\mathrm{del}$.
The factor $\mathsf{c}$ is algebrised by the tensor product, over the variables $x$ occurring in $\mathsf{out}_{\func{f}}$, in the fixed order,
of an $r'_x$-fold copy (which is $\mathrm{id}$ for $r'_x = 1$), where $r'_x \geq 1$ is the number of occurrences of $x$ in $\mathsf{out}_{\func{f}}$.
The composite $t$ of the five terms, in the order of the factors, is an algebraisation of ${\func{f}}$.
Since ${\func{f}}$ has no isolated variables, it has at most $2k$ variables,
so that the terms for $\mathsf{e}$ and $\mathsf{c}$ have DAG-size $O(k)$ and are written down in $O(k)$ time;
hence the DAG-size of $t$ and the time needed to compute it are dominated by the permutation-type factors,
and are in $O(k^2)$ by \Cref{lem:algebrise-permutation}.
Finally, every sub-term of $t$ is a sub-term of one of the five terms, or a composite of some of them.
The sub-terms of the permutation-type terms have width at most $2k$ by \Cref{lem:algebrise-permutation};
the terms for $\mathsf{e}$ and $\mathsf{c}$ and their sub-terms have at most $|\mathsf{in}_{\func{f}}| + |\mathsf{dist}(\mathsf{out}_{\func{f}})| \leq 2k$
and $|\mathsf{dist}(\mathsf{out}_{\func{f}})| + |\mathsf{out}_{\func{f}}|\leq 2k$ in- and outputs in total, respectively;
and a composite of some of the five terms has at most $\max\{|\mathsf{in}_{\func{f}}|, |\mathsf{dist}(\mathsf{out}_{\func{f}})|\} \leq k$ inputs and at most $\max\{|\mathsf{dist}(\mathsf{out}_{\func{f}})|, |\mathsf{out}_{\func{f}}|\} \leq k$ outputs.
Hence $\mathsf{width}(t) \leq 2k$. \qed

\subsection{Proof of \Cref{lem:wiring-tensor-factorisation}}\label{sec:proof-of-lem-wiring-tensor-factorisation}

Write $\mathbf{z} = (z_1, \dots, z_u)$ and $A_{\func{f}} = (a_1, \dots, a_n)$.
Both wirings have variable set $X_{\func{f}}$ and no assignments.
The wiring $\mathsf{w}_2$, which is applied first, has list of inputs $\mathsf{in}_{\func{f}}$ and list of outputs
$$ \mathsf{out}_{\mathsf{w}_2} := \mathbf{z}\, \mathbf{x}_{a_1}\mathbf{x}_{a_1}\mathbf{y}_{a_1} \cdots \mathbf{x}_{a_n}\mathbf{x}_{a_n}\mathbf{y}_{a_n},$$
where juxtaposition denotes the concatenation of lists;
the wiring $\mathsf{w}_1$ has list of inputs
$$ \mathsf{in}_{\mathsf{w}_1} := \mathbf{z}\, \mathbf{y}_{a_1}\mathbf{x}_{a_1}\mathbf{y}_{a_1} \cdots \mathbf{y}_{a_n}\mathbf{x}_{a_n}\mathbf{y}_{a_n}$$
and list of outputs $\mathsf{out}_{\func{f}}$.

Let $\mathsf{m}$ denote the middle factor $\mathsf{id}_{\mathbf{U}} \otimes \bigotimes_{a \in A_{\func{f}}} (g_a \otimes \mathsf{id}_{\mathbf{S}_a} \otimes \mathsf{id}_{\mathbf{T}_a})$.
By \Cref{definition:tensor-product-of-dot-diagrams}, its variables are
$r_1, \dots, r_u$ for $\mathsf{id}_{\mathbf{U}}$,
and, for each $a\in A_{\func{f}}$,
the input variables $\mathbf{u}_a = (u_{a,1}, \dots, u_{a, |\mathbf{S}_a|})$ and output variables $\mathbf{v}_a = (v_{a,1}, \dots, v_{a,|\mathbf{T}_a|})$ of the dot diagram $\langle g_a\rangle$,
and the variables $\mathbf{p}_a = (p_{a,1}, \dots, p_{a,|\mathbf{S}_a|})$ and $\mathbf{q}_a = (q_{a,1}, \dots, q_{a, |\mathbf{T}_a|})$ of $\mathsf{id}_{\mathbf{S}_a}$ and $\mathsf{id}_{\mathbf{T}_a}$.
Its lists of inputs and outputs are
\[ (r_1, \dots, r_u)\, \mathbf{u}_{a_1} \mathbf{p}_{a_1} \mathbf{q}_{a_1} \cdots \mathbf{u}_{a_n}\mathbf{p}_{a_n}\mathbf{q}_{a_n} \quad\text{and}\quad (r_1, \dots, r_u)\, \mathbf{v}_{a_1}\mathbf{p}_{a_1}\mathbf{q}_{a_1}\cdots\mathbf{v}_{a_n}\mathbf{p}_{a_n}\mathbf{q}_{a_n}, \]
respectively, and its assignments are $(\lett \mathbf{v}_a = g_a(\mathbf{u}_a))$ for $a \in A_{\func{f}}$, in this order.
Now form the composite $\mathsf{w}_1 \circ \mathsf{m} \circ \mathsf{w}_2$ according to \Cref{definition:composite-of-dot-diagrams},
writing $x^{(2)}$ and $x^{(1)}$ for the copies of a variable $x \in X_{\func{f}}$ in $\mathsf{w}_2$ and $\mathsf{w}_1$, respectively.
Composing $\mathsf{w}_2$ with $\mathsf{m}$ identifies $z_l^{(2)}$ with $r_l$,
and, for each $a$ and all $i, j$,
the variable $x_{a,i}^{(2)}$ with $u_{a,i}$ and with $p_{a,i}$,
and $y_{a,j}^{(2)}$ with $q_{a,j}$.
Composing with $\mathsf{w}_1$ identifies $r_l$ with $z_l^{(1)}$, $v_{a,j}$ with $y^{(1)}_{a,j}$, $p_{a,i}$ with $x^{(1)}_{a,i}$, and $q_{a,j}$ with $y_{a,j}^{(1)}$.
Consequently, $x^{(1)}$ is identified with $x^{(2)}$ for every variable $x$ of ${\func{f}}$:
via $p_{a,i}$ or $q_{a,j}$ if $x$ occurs in some assignment $a$, and via $r_l$ if $x = z_l$ occurs in no assignment.
Moreover, $u_{a,i}$ is identified with $x_{a,i}$ and $v_{a,j}$ with $y_{a,j}$,
and there are no further identifications.
Hence the variables of the composite are in bijection with $X_{\func{f}}$,
and under this bijection its assignments are $(\lett \mathbf{y}_a = g_a(\mathbf{x}_a))$ for $a\in A_{\func{f}}$ in this order,
its list of inputs is $\mathsf{in}_{\func{f}}$, and its list of outputs is $\mathsf{out}_{\func{f}}$.
That is, ${\func{f}} \equiv \mathsf{w}_1 \circ \mathsf{m} \circ \mathsf{w}_2$.

Now suppose that ${\func{f}}$ has no isolated variables.
Then every entry of $\mathbf{z}$ occurs in $\mathsf{in}_{\func{f}}$ or in $\mathsf{out}_{\func{f}}$, so that $u \leq |\mathsf{in}_{\func{f}}| + |\mathsf{out}_{\func{f}}| \leq \mathsf{maxar}({\func{f}})$.
To write down $\mathsf{w}_1$ and $\mathsf{w}_2$, we concatenate $6n + 2$ lists,
each of length at most $\mathsf{maxar}({\func{f}})$,
which takes $O((n + 1)\cdot\mathsf{maxar}({\func{f}}))$ time. \qed

\subsection{Proof of \Cref{cor:algebrise-small-decls}}\label{proof-of-cor-algebrise-small-decls}

Decompose ${\func{f}} \equiv \mathsf{w}_1 \circ \mathsf{m} \circ \mathsf{w}_2$ as in the proof of \Cref{lem:wiring-tensor-factorisation},
and algebrise the three factors separately.
The number $u$ of variables of ${\func{f}}$ occurring in no assignment is at most $a$, as noted in the proof of \Cref{lem:wiring-tensor-factorisation}.
The list $\mathsf{out}_{\mathsf{w}_2}$ is the concatenation of $\mathbf{z}$ with, for each assignment, its list of inputs twice and its list of outputs once,
which together have length at most twice the arity of its function symbol, hence at most $2a$;
the same applies to $\mathsf{in}_{\mathsf{w}_1}$.
Hence $\max \{|\mathsf{in}_{\mathsf{w}_i}|, |\mathsf{out}_{\mathsf{w}_i}|\} \leq 2na + a$ for $i \in \{1, 2\}$.
Moreover, every variable of ${\func{f}}$ occurs in $\mathsf{out}_{\mathsf{w}_2}$ and in $\mathsf{in}_{\mathsf{w}_1}$, so that neither wiring has isolated variables.
By \Cref{cor:algbrise-pure-wiring}, $\mathsf{w}_1$ and $\mathsf{w}_2$ therefore have algebraisations of width at most $4na + 2a$ and DAG-size $O(n^2a^2)$,
which are computed in $O(n^2a^2)$ time.
The middle factor $\mathsf{m}$ is a tensor product of at most $3n + 1$ function symbols and identities;
it is written as a term of DAG-size $O(na)$ in $O(na)$ time,
and its number of inputs plus its number of outputs is $2u$ plus three times the sum of the arities of the function symbols of the assignments, hence at most $2a + 3na$.
Let $t$ be the composite of the three terms.
Its DAG-size is in $O(n^2a^2)$.
Its width is the maximum of the widths of the three terms
and of the numbers of in- and outputs of the composites of two or three of them,
all of which are at most $4na + 2a$.
(The assumption $n \geq 1$ ensures that the bounds $O(n^2a^2)$ do not degenerate;
pure wirings are covered by \Cref{cor:algbrise-pure-wiring}.) \qed

\subsection{Substituting algebraisations}\label{sec:substituting-algebraisations}

The proofs of \Cref{lem:algebrise-hierarchical-function-declarations,thm:algebraisation-general} both combine algebraisations of the dot diagrams appearing in a hierarchical dot diagram
into an algebraisation of its unfolding.
The following lemma isolates this step.

\begin{lemma}\label{lem:substitute-algebraisations}
    Let ${\func{f}}$ be a hierarchical $\Sigma$-diagram, and for every $v \in V({\func{f}})$, let $t_v$ be an algebraisation of ${\func{f}}_v$.
    Then $\mathsf{unfold}({\func{f}})$ has an algebraisation $t$ with
    \[ \mathsf{width}(t) \leq \max_{v \in V({\func{f}})} \mathsf{width}(t_v) \quad\text{and}\quad |t| \leq \sum_{v \in V({\func{f}})} |t_v|. \]
    Moreover, $t$ can be computed from the terms $t_v$ in $O\big(\sum_{v\in V({\func{f}})} |t_v|\big)$ time.
\end{lemma}
\begin{proof}
    For $v\in V({\func{f}})$, define
    $$
    \theta(v) :=
    \begin{cases}
        t_v & \text{if } v \text{ is a leaf,} \\
        t_v[\mathrm{defn}_{{\func{f}}, v}^{-1}(u) \mapsto \theta(u)]_{u \in \mathrm{ch}(v)} & \text{otherwise,}
    \end{cases}
    $$
    and let $t := \theta(R_{\func{f}})$.
    By induction on the height of $v$, the term $\theta(v)$ is an algebraisation of $\mathsf{unfold}'({\func{f}}, v)$:
    for a leaf this is the assumption,
    and otherwise it follows from the induction hypothesis and \Cref{lem:algebraisation-commutes-with-substitution}.
    In particular, $t$ is an algebraisation of $\mathsf{unfold}({\func{f}})$.

    A sub-term of a substitution $s[g_i \mapsto s_i]_{i}$ is either a sub-term of some $s_i$,
    or arises from a sub-term of $s$ by the same substitution, which does not change its numbers of in- and outputs.
    Hence $\mathsf{width}(s[g_i\mapsto s_i]_i) \leq \max\{\mathsf{width}(s), \max_i \mathsf{width}(s_i)\}$,
    and the bound on $\mathsf{width}(t)$ follows by induction.

    To compute $t$, traverse the call graph of ${\func{f}}$ bottom-up,
    and represent $\theta(v)$ by the DAG representing $t_v$,
    in which every leaf labelled by a function symbol $g \in \Sigma_v$ is replaced by a reference to the root of the DAG representing $\theta(\mathrm{defn}_{{\func{f}}, v}(g))$.
    This takes $O(|t_v|)$ time at $v$, and the DAG obtained for $t$ has at most $\sum_v |t_v|$ nodes.
    It need not have maximal sharing;
    canonicalising it at the end merges nodes that represent the same sub-term, and hence does not increase the number of nodes.
\end{proof}

\subsection{Proof of \Cref{lem:algebrise-hierarchical-function-declarations}}\label{sec:proof-of-lem-algebrise-hierarchical-function-declarations}

For every $v \in V({\func{f}})$, let $t_v$ be the algebraisation of ${\func{f}}_v$ from \Cref{cor:algebrise-small-decls} if $|A_{{\func{f}}_v}| \geq 1$,
and from \Cref{cor:algbrise-pure-wiring} if ${\func{f}}_v$ is a pure wiring.
Since ${\func{f}}_v$ has no isolated variables, at most $N$ assignments and maximal arity at most $K$,
the term $t_v$ has width at most $4NK + 2K$ and DAG-size $O(N^2K^2)$, and is computed in $O(N^2K^2)$ time.
By \Cref{lem:substitute-algebraisations}, $\mathsf{unfold}({\func{f}})$ therefore has an algebraisation of width at most $4NK + 2K$ and DAG-size $O(|V({\func{f}})| N^2K^2)$,
which is computed in $O(|V({\func{f}})| N^2K^2)$ time in total. \qed

\subsection{Proof of \Cref{lem:decomp-to-hierarchical-decl}}\label{sec:proof-of-lem-decomp-to-hierarchical-decl}

If ${\func{f}}$ has no assignments, we may take ${\func{f}}^*$ to be the hierarchical $\Sigma$-diagram with a single vertex, whose dot diagram is ${\func{f}}$ itself;
properties (1)--(4) then hold trivially.
Hence, assume that $|A_{\func{f}}|\geq 1$.
Then $\mathcal{B}$ has $|A_{\func{f}}| + 2 \geq 3$ leaves, and hence at least one internal vertex.
We consider $\mathcal{B}$ as a tree rooted at the leaf $\ell$ with $\eta(\ell) = \mathsf{in}$.
Then every internal vertex of $\mathcal{B}$ has exactly two children,
and the root $\ell$ has exactly one child, which is an internal vertex $r$.

For a vertex $v \neq \ell$ of $\mathcal{B}$, let $T_v$ be the subtree rooted at $v$,
let $H_v$ be the set of hyperedges of $\mathrm{Dep}({\func{f}})$ at the leaves of $T_v$,
and let $A_v$ be the sub-list of $A_{\func{f}}$ consisting of the assignments in $H_v$.
Let $I_v$ be the set of variables that are incident with a hyperedge in $H_v$ and with a hyperedge not in $H_v$.
This is the set of variables at the interface of the edge from $v$ to its parent,
so that $|I_v| \leq \mathsf{width}(\mathcal{B})$.
Write $\mathsf{vars}(A_v)$ for the set of variables occurring in some assignment in $A_v$.
For every internal vertex $v \neq r$ of $\mathcal{B}$,
define two lists of pairwise distinct variables:
$\mathsf{in}_v$ is the list of all variables in $\mathsf{vars}(A_v)$ that occur in $\mathsf{in}_{\func{f}}$ or in an assignment not in $A_v$,
and $\mathsf{out}_v$ is the list of all variables in $\mathsf{vars}(A_v)$ that occur in $\mathsf{out}_{\func{f}}$ or in an assignment not in $A_v$,
both in the order of first occurrence in $A_v$;
moreover, let $\mathsf{in}_r := \mathsf{in}_{\func{f}}$ and $\mathsf{out}_r := \mathsf{out}_{\func{f}}$.
A variable of $\mathsf{vars}(A_v)$ occurring in an assignment not in $A_v$ lies in $I_v$;
hence, for $v \neq r$,
$$|\mathsf{in}_v| + |\mathsf{out}_v| \leq 2|I_v| + |\mathsf{set}(\mathsf{in}_{\func{f}})| + |\mathsf{set}(\mathsf{out}_{\func{f}})| \leq 2\,\mathsf{width}(\mathcal{B}) + \mathsf{maxar}({\func{f}}).$$
For $v = r$, this bound holds since $|\mathsf{in}_{\func{f}}| + |\mathsf{out}_{\func{f}}| \leq \mathsf{maxar}({\func{f}})$.

The hierarchical $\Sigma$-diagram ${\func{f}}^*$ has as vertices the internal vertices of $\mathcal{B}$, with root $r$,
and as edges the edges of $\mathcal{B}$ from parents to children that are internal vertices.
The dot diagram ${\func{f}}^*_v$ at an internal vertex $v$ with children $c_1, c_2$ has one assignment for each child $c_i$ that is not the leaf with hyperedge $\mathsf{out}$:
if $c_i$ is a leaf whose hyperedge is the assignment $a \in A_{\func{f}}$, this assignment is $a$ itself;
if $c_i$ is an internal vertex, this assignment is $(\lett \mathsf{out}_{c_i} = g_{c_i}(\mathsf{in}_{c_i}))$
for a fresh function symbol $g_{c_i}\colon \mathsf{sort}_*(\mathsf{in}_{c_i}) \to \mathsf{sort}_*(\mathsf{out}_{c_i})$ of $\Sigma_v$,
with $\mathrm{defn}_{{\func{f}}^*, v}(g_{c_i}) := c_i$.
Its lists of inputs and outputs are $\mathsf{in}_v$ and $\mathsf{out}_v$,
and its variable set consists of the variables occurring in its assignments and in these two lists.
By construction, ${\func{f}}^*$ is a hierarchical $\Sigma$-diagram: every function symbol in $\Sigma_v$ occurs in ${\func{f}}^*_v$.

We verify properties (1)--(4).
For (1), let $X_v$ be the set of variables occurring in $A_v$, in $\mathsf{in}_v$ or in $\mathsf{out}_v$.
Since ${\func{f}}$ has no isolated variables, $X_r = X_{\func{f}}$, and $(X_r, A_r, \mathsf{in}_r, \mathsf{out}_r)$ is ${\func{f}}$ up to the order of the assignments.
We show by induction on the height of an internal vertex $v$ that $\mathsf{unfold}'({\func{f}}^*, v) \equiv (X_v, A_v, \mathsf{in}_v, \mathsf{out}_v)$,
where the list $A_v$ may be reordered;
for $v = r$, this is property (1).
The claim is clear if both children of $v$ are leaves.
In general, substituting the unfoldings of the internal children into ${\func{f}}^*_v$
identifies the input and output variables of each child's unfolding with the corresponding variables of ${\func{f}}^*_v$;
a variable occurring in the assignments of both children occurs in both of their lists of inputs (as it occurs in an assignment outside either child),
so that its two copies are identified,
whereas a variable occurring in the assignments of only one child and in none of the lists $\mathsf{in}_v, \mathsf{out}_v$
occurs in neither list of that child and is not identified with anything.
The result is $(X_v, A_v, \mathsf{in}_v, \mathsf{out}_v)$ up to reordering of $A_v$, as claimed.
For (2), the maximal arity of ${\func{f}}^*_v$ is the maximum of $|\mathsf{in}_v| + |\mathsf{out}_v|$ and of the arities of the function symbols of its assignments.
These function symbols are function symbols of ${\func{f}}$, whose arities are at most $\mathsf{maxar}({\func{f}})$,
and the fresh symbols $g_c$ for internal vertices $c \neq r$, whose arities are $|\mathsf{in}_c| + |\mathsf{out}_c| \leq 2\,\mathsf{width}(\mathcal{B}) + \mathsf{maxar}({\func{f}})$.
Property (3) holds by construction.
For (4), an unrooted binary tree with $|A_{\func{f}}| + 2$ leaves has $|A_{\func{f}}|$ internal vertices;
hence $|V({\func{f}}^*)| = |A_{\func{f}}|$.

Finally, we consider the time complexity.
For every variable $x$, the vertices $v$ with $x \in I_v$ are exactly the vertices $v \neq \ell$ whose parent edge lies on a path between two leaves whose hyperedges contain $x$;
hence the sets $I_v$ can be computed by a traversal of $\mathcal{B}$ in time $O(\sum_v |I_v| + |V(\mathcal{B})| \cdot \mathsf{maxar}({\func{f}})) = O(|V(\mathcal{B})| \cdot (\mathsf{width}(\mathcal{B}) + \mathsf{maxar}({\func{f}})))$,
where the second summand accounts for reading the hyperedges at the leaves, each of which has at most $\mathsf{maxar}({\func{f}})$ incident variables.
The lists $\mathsf{in}_v$ and $\mathsf{out}_v$ are then obtained from $I_v$, $\mathsf{set}(\mathsf{in}_{\func{f}})$ and $\mathsf{set}(\mathsf{out}_{\func{f}})$
in time $O(\mathsf{width}(\mathcal{B}) + \mathsf{maxar}({\func{f}}))$ per vertex,
after which each ${\func{f}}^*_v$ is written down in the same time.
This leads to a total time complexity of $O(|V(\mathcal{B})|\cdot (\mathsf{width}(\mathcal{B}) + \mathsf{maxar}({\func{f}})))$. \qed

\subsection{Proof of \Cref{thm:algebraisation-general}}\label{sec:proof-of-thm-algebraisation-general}

For every $v \in V({\func{f}})$, the dot diagram ${\func{f}}_v$ has no isolated variables,
$\mathsf{tw}({\func{f}}_v') \leq k$, $\mathsf{maxar}({\func{f}}_v) \leq a$ and $1 + |A_{{\func{f}}_v}| \leq N$.
Hence \Cref{thm:algebraisation} yields an algebraisation $t_v$ of ${\func{f}}_v$ with width in $O(k + a)$ and DAG-size $O(N(k + a)^2)$ in $O(2^{O(k)} N a^2)$ time.
By \Cref{lem:substitute-algebraisations}, $\mathsf{unfold}({\func{f}})$ therefore has an algebraisation of width in $O(k + a)$ and DAG-size $O(MN(k+a)^2)$,
which is computed in $O(M \cdot 2^{O(k)} N a^2 + MN(k + a)^2) \subseteq O(LMN\cdot 2^{O(k)} a^2)$ time in total,
using $(k+a)^2 \leq 2^{O(k)} a^2$ and $L \geq 1$. \qed

\section{Proofs for \Cref{sec:applications}}\label{appendix:proofs-applications}

\subsection{{Proof of \Cref{thm:semantics-in-matrix-categories}}}\label{sec:proof-of-thm-semantics-in-matrix-categories}

By \Cref{thm:algebraisation-general}, 
we can compute an algebraisation $t$ of $\mathsf{unfold}({\func{f}})$
in $O(L M N \cdot 2^{O(k)} a^2)$ time such that the width of $t$ is in $O(k + a)$
and the DAG-size of $t$ is in $O(M N (k+a)^2)$. 
Let $\Delta$ be the maximum of the dimensions $\mathcal{I}(S)$ over the sorts $S\in \Sigma_0$ occurring in ${\func{f}}$.
In the DAG-representation of $t$, 
every node corresponds to a sub-term $s$ of $t$, whose semantics $[\![ s ]\!]_{\mathcal{I}}$ is a matrix with at most $\Delta^{\mathsf{width}(t)} = 2^{O(k + a)}$ rows and columns;
every non-leaf node corresponds to a matrix multiplication 
or Kronecker product over $R$ of the matrices at its children.
Any individual matrix multiplication of such matrices requires 
$O((\Delta^{\mathsf{width}(t)})^3)= O(2^{O(k + a)})$ arithmetic operations,  
and each individual Kronecker product can be performed 
using $O((\Delta^{\mathsf{width}(t)})^2) = O(2^{O(k + a)})$ arithmetic operations.  
We construct the arithmetic circuit $C$ by replacing each leaf of the DAG by vertices labelled by the entries of the corresponding matrix $\mathcal{I}(g)$, $\mathsf{id}$, $\mathsf{swap}$, $\mathsf{copy}$, $\mathsf{del}$, $\mathsf{equate}$ or $\mathsf{new}$, 
and each non-leaf node by an arithmetic sub-circuit computing the entries of its matrix from the entries of the matrices at its children:
for a Kronecker product, every entry is the product of two entries of the children's matrices; 
for a matrix multiplication, every entry is a sum of products of two entries of the children's matrices, 
which we compute by a chain of additions, 
so that every vertex of the sub-circuit computes either such a product or a partial sum of such products.
In this way, we obtain an arithmetic circuit $C$ with $O(M N \cdot 2^{O(k + a)})$ vertices
in $O(L M N \cdot 2^{O(k + a)})$ time, whose roots are the vertices computing the entries of $[\![t]\!]_{\mathcal{I}} = [\![{\func{f}}]\!]_{\mathcal{I}}$. \qed

\subsection{Gs-monoidal terms}\label{sec:gs-monoidal-terms}

In this and the following subsections, we establish the properties of gs-monoidal dot diagrams used in \Cref{ssec:inference}.
All of them rest on one fact: the gs-monoidal dot diagrams are exactly the dot diagrams that can be built from function symbols, identities, swaps, copies and deletions.
Call a hypergraph term \emph{gs-monoidal} if it contains no symbol $\mathrm{equate}_{\mathbf{T}}$ and no symbol $\mathrm{new}_{\mathbf{T}}$.
Gs-monoidality of dot diagrams is invariant under equivalence, since renaming variables and reordering assignments changes neither the number of occurrences of a variable in a list nor the dependencies between variables.
We call a variable \emph{defined} by an assignment if it occurs in the list of outputs of the assignment.

\begin{lemma}\label{lem:gs-monoidal-terms}
Let $\Sigma$ be a monoidal signature.
If $t$ is a gs-monoidal hypergraph term over $\Sigma$, then $\underline{\mathsf{dot}}(t)$ is represented by a gs-monoidal $\Sigma$-diagram.
Conversely, every gs-monoidal $\Sigma$-diagram has a gs-monoidal algebraisation.
\end{lemma}
\begin{proof}
The dot diagrams $\langle g \rangle$, $\func{id}$, $\func{swap}$, $\func{copy}$ and $\func{del}$ are gs-monoidal,
since each of their variables occurs exactly once in the list of inputs or, in the case of $\langle g \rangle$, in the list of outputs of its single assignment, and there are no cyclic dependencies;
and composites and tensor products of gs-monoidal $\Sigma$-diagrams are gs-monoidal (\Cref{remark:gs-monoidal}).
Since $\underline{\mathsf{dot}}$ maps composites and tensor products of terms to composites and tensor products of dot diagrams, the first claim follows by induction on the structure of $t$.

For the second claim, let $D$ be the set of all classes $\underline{\mathsf{dot}}(t)$ with $t$ a gs-monoidal hypergraph term over $\Sigma$.
By \Cref{remark:gs-monoidal}, the gs-monoidal $\Sigma$-diagrams modulo equivalence form the free gs-monoidal category $\mathrm{FreeGS}_\Sigma$ of Fritz and Liang, whose objects are the lists of sorts,
and by the first claim, $D$ is a set of morphisms of this category.
It contains the identities, the symmetries and the copy and delete morphisms of all objects, it contains $\langle g \rangle$ for every function symbol $g$ of $\Sigma$, and it is closed under composition and tensor product;
hence $D$ is a strict gs-monoidal subcategory of $\mathrm{FreeGS}_\Sigma$ with the same objects.
By \cite[Corollary 4.7]{fritz2023free}, applied to $D$, the interpretation $g \mapsto \langle g \rangle$ of $\Sigma$ in $D$ extends to a strict gs-monoidal functor $F\colon \mathrm{FreeGS}_\Sigma \to D$.
The composite of $F$ with the inclusion $J\colon D \to \mathrm{FreeGS}_\Sigma$ is a strict gs-monoidal functor from $\mathrm{FreeGS}_\Sigma$ to itself that maps every $\langle g \rangle$ to itself, and so is the identity functor;
by the uniqueness in \cite[Corollary 4.7]{fritz2023free}, applied to $\mathrm{FreeGS}_\Sigma$, the two functors are equal.
Hence $J$ is surjective on morphisms, which is the second claim.
\end{proof}

\begin{corollary}\label{lem:unfold-gs-monoidal}
The unfolding of a hierarchical gs-monoidal $\Sigma$-diagram is a gs-monoidal $\Sigma$-diagram. In particular, this applies to DPPs.
\end{corollary}
\begin{proof}
Let ${\func{f}}$ be a hierarchical gs-monoidal $\Sigma$-diagram.
We show by induction on the height of a vertex $v$ of its call graph that $\mathsf{unfold}'({\func{f}}, v)$ has a gs-monoidal algebraisation;
by \Cref{lem:gs-monoidal-terms}, $\mathsf{unfold}'({\func{f}}, v)$ is then gs-monoidal, and for the root this is the claim.
If $v$ is a leaf, then $\mathsf{unfold}'({\func{f}}, v) = {\func{f}}_v$ is gs-monoidal and has a gs-monoidal algebraisation by \Cref{lem:gs-monoidal-terms}.
Otherwise, choose a gs-monoidal algebraisation $t$ of ${\func{f}}_v$ and, by the induction hypothesis, a gs-monoidal algebraisation $s_u$ of $\mathsf{unfold}'({\func{f}}, u)$ for every child $u$ of $v$.
By \Cref{lem:algebraisation-commutes-with-substitution}, the term $t[\mathrm{defn}_{{\func{f}}, v}^{-1}(u) \mapsto s_u]_{u \in \mathrm{ch}(v)}$ is an algebraisation of $\mathsf{unfold}'({\func{f}}, v)$,
and it is gs-monoidal, since substitution introduces no symbols $\mathrm{equate}_{\mathbf{T}}$ or $\mathrm{new}_{\mathbf{T}}$.
\end{proof}

\subsection{Substochasticity}\label{sec:substochastic}

\begin{lemma}\label{lem:substochastic}
Let $\mathcal{I}\colon \Sigma \to \mathsf{Mat}_{\mathbb{Q}_{\geq 0}}$ be an interpretation such that $\mathcal{I}(g)$ is substochastic for every function symbol $g$ of $\Sigma$.
Then $[\![{\func{f}}]\!]_{\mathcal{I}}$ is substochastic for every gs-monoidal $\Sigma$-diagram ${\func{f}}$.
\end{lemma}
\begin{proof}
Substochastic matrices are closed under products and Kronecker products:
if $A$ and $B$ are substochastic and $AB$ is defined, then the $j$-th column sum of $AB$ is $\sum_k B_{kj} \sum_i A_{ik} \leq \sum_k B_{kj} \leq 1$,
and every column sum of $A \otimes B$ is the product of a column sum of $A$ and a column sum of $B$.
Moreover, the identity matrices and the matrices $\mathsf{swap}_{n,m}$, $\mathsf{copy}_n$ and $\mathsf{del}_n$ of \Cref{sec:semantics} are substochastic, since every column of each of them has exactly one entry $1$ and all other entries $0$.
By \Cref{lem:gs-monoidal-terms}, ${\func{f}}$ has a gs-monoidal algebraisation $t$, and $[\![{\func{f}}]\!]_{\mathcal{I}} = [\![t]\!]_{\mathcal{I}}$ (\Cref{ssec:algebraisations}).
By \Cref{def:hypergraph-terms}, $[\![t]\!]_{\mathcal{I}}$ is obtained from the matrices $\mathcal{I}(g)$ and the matrices just listed by products and Kronecker products, and is therefore substochastic.
\end{proof}

\subsection{Proof of \Cref{lem:dpp-substochastic}}\label{sec:proof-of-lem-dpp-substochastic}

By \Cref{lem:unfold-gs-monoidal}, $\mathsf{unfold}({\func{f}})$ is a gs-monoidal $\mathsf{BoolStochSig}$-diagram, 
and all matrices $\mathcal{S}(g)$, $g \in \mathsf{BoolStochSig}$, are substochastic.
Hence $[\![{\func{f}}]\!]_{\mathcal{S}} = [\![\mathsf{unfold}({\func{f}})]\!]_{\mathcal{S}}$ is substochastic by \Cref{lem:substochastic}. \qed

\subsection{Entries of the semantics of sub-terms}\label{sec:subterm-entries}

A sub-term of an algebraisation of a gs-monoidal dot diagram need not be gs-monoidal:
for instance, $\mathrm{equate}_{(S)} \circ (\mathrm{new}_{(S)} \otimes \mathrm{id}_{(S)})$ is an algebraisation of $\func{id}_{(S)}$.
Such a sub-term is, however, a gs-monoidal dot diagram with some of its wires bent:
$\func{equate}$ arises from $\func{copy}$ by bending one output down into an input, and $\func{new}$ arises from $\func{del}$ by bending its input up into an output.
In $\mathsf{Mat}_R$, bending a wire rearranges the entries of a matrix without changing them.
This is the idea behind the following two lemmas.

For a sort $S$, let $\func{cap}_S := \func{del}_S \circ \func{equate}_S\colon (S, S) \to (\,)$ and $\func{cup}_S := \func{copy}_S \circ \func{new}_S\colon (\,) \to (S, S)$.
Let ${\func{f}}$ be a $\Sigma$-diagram.
If the list of outputs of ${\func{f}}$ is nonempty, let $y$ be its last entry, let $S$ be the sort of $y$, and let $\mathbf{T}$ be the list of sorts of the remaining outputs;
we call
\[ {\func{f}}^{\downarrow} := (\func{id}_{\mathbf{T}} \otimes \func{cap}_S) \circ ({\func{f}} \otimes \func{id}_{(S)}) \]
the result of \emph{bending down} the last output of ${\func{f}}$.
Dually, if the list of inputs of ${\func{f}}$ is nonempty, let $x$ be its last entry, let $S$ be the sort of $x$, and let $\mathbf{U}$ be the list of sorts of the remaining inputs;
we call
\[ {\func{f}}^{\uparrow} := ({\func{f}} \otimes \func{id}_{(S)}) \circ (\func{id}_{\mathbf{U}} \otimes \func{cup}_S) \]
the result of \emph{bending up} the last input of ${\func{f}}$.
Finally, for an interpretation $\mathcal{I}$ of $\Sigma$ in $\mathsf{Mat}_R$, we index the rows of the matrix $[\![{\func{f}}]\!]_{\mathcal{I}}$ by tuples of indices, one for each output of ${\func{f}}$, and its columns by tuples of indices, one for each input,
via the basis vectors $e_{i_1} \otimes \dots \otimes e_{i_n}$ of $R^{\mathcal{I}(S_1) \cdots \mathcal{I}(S_n)}$, where $S_1, \dots, S_n$ are the sorts of the outputs (inputs, respectively).

\begin{lemma}\label{lem:bending}
Let ${\func{f}}$ be a $\Sigma$-diagram, and let $\mathcal{I}$ be an interpretation of $\Sigma$ in $\mathsf{Mat}_R$ for a semiring $R$.
\begin{bracketenumerate}
\item If the list of outputs of ${\func{f}}$ is nonempty with last entry $y$, then ${\func{f}}^{\downarrow}$ is equivalent to the $\Sigma$-diagram with the variables and assignments of ${\func{f}}$ in which $y$ is removed from the list of outputs and appended to the list of inputs,
and the entry of $[\![{\func{f}}^{\downarrow}]\!]_{\mathcal{I}}$ at row $\mathbf{b}$ and column $(\mathbf{c}, j)$ is the entry of $[\![{\func{f}}]\!]_{\mathcal{I}}$ at row $(\mathbf{b}, j)$ and column $\mathbf{c}$.
\item If the list of inputs of ${\func{f}}$ is nonempty with last entry $x$, then ${\func{f}}^{\uparrow}$ is equivalent to the $\Sigma$-diagram with the variables and assignments of ${\func{f}}$ in which $x$ is removed from the list of inputs and appended to the list of outputs,
and the entry of $[\![{\func{f}}^{\uparrow}]\!]_{\mathcal{I}}$ at row $(\mathbf{b}, i)$ and column $\mathbf{c}$ is the entry of $[\![{\func{f}}]\!]_{\mathcal{I}}$ at row $\mathbf{b}$ and column $(\mathbf{c}, i)$.
\item If ${\func{f}}'$ arises from ${\func{f}}$ by transposing two adjacent entries of the list of inputs or of the list of outputs, then ${\func{f}}' \equiv {\func{f}} \circ \sigma$ or ${\func{f}}' \equiv \sigma \circ {\func{f}}$, respectively, for a tensor product $\sigma$ of identities and one swap,
and $[\![{\func{f}}']\!]_{\mathcal{I}}$ arises from $[\![{\func{f}}]\!]_{\mathcal{I}}$ by permuting its columns or its rows, respectively.
\end{bracketenumerate}
In particular, in all three cases every entry of the semantics of the new $\Sigma$-diagram is an entry of $[\![{\func{f}}]\!]_{\mathcal{I}}$.
\end{lemma}
\begin{proof}
For (1), by \Cref{definition:tensor-product-of-dot-diagrams}, ${\func{f}} \otimes \func{id}_{(S)}$ has the variables of ${\func{f}}$ and one further variable $u$, the assignments of ${\func{f}}$, and the lists of ${\func{f}}$ with $u$ appended to both;
the pure wiring $\func{id}_{\mathbf{T}} \otimes \func{cap}_S$ has variables $w_1, \dots, w_k, z$, where $k := |\mathbf{T}|$, the list of inputs $(w_1, \dots, w_k, z, z)$ and the list of outputs $(w_1, \dots, w_k)$.
By \Cref{definition:composite-of-dot-diagrams}, their composite identifies $w_l$ with the $l$-th output of ${\func{f}}$ for $l \leq k$, and $z$ with both $y$ and $u$.
Its variables are therefore in bijection with those of ${\func{f}}$, its assignments are those of ${\func{f}}$, its list of inputs is that of ${\func{f}}$ followed by $u \sim y$, and its list of outputs is $(w_1, \dots, w_k)$, the list of outputs of ${\func{f}}$ without $y$.
For the entries, write $M := [\![{\func{f}}]\!]_{\mathcal{I}}$ and $d := \mathcal{I}(S)$.
Since $[\![\,\cdot\,]\!]_{\mathcal{I}}$ is a strict hypergraph functor, $[\![{\func{f}}^{\downarrow}]\!]_{\mathcal{I}} = (\mathsf{id}_{\mathcal{I}(\mathbf{T})} \otimes (\mathsf{del}_d \circ \mathsf{equate}_d)) \circ (M \otimes \mathsf{id}_d)$.
The entry of $M \otimes \mathsf{id}_d$ at row $(\mathbf{b}', j, j'')$ and column $(\mathbf{c}, j')$ is $M_{(\mathbf{b}', j), \mathbf{c}}\, \delta_{j'' j'}$,
and the entry of $\mathsf{id}_{\mathcal{I}(\mathbf{T})} \otimes (\mathsf{del}_d \circ \mathsf{equate}_d)$ at row $\mathbf{b}$ and column $(\mathbf{b}', j, j'')$ is $\delta_{\mathbf{b} \mathbf{b}'} \delta_{j j''}$.
Summing the product of the two over $\mathbf{b}'$, $j$ and $j''$ gives $M_{(\mathbf{b}, j'), \mathbf{c}}$.

For (2), the composite with $\func{id}_{\mathbf{U}} \otimes \func{cup}_S$ identifies the variable of $\func{cup}_S$ with both $x$ and the further variable $u$ of ${\func{f}} \otimes \func{id}_{(S)}$, which moves $x$ from the list of inputs to the end of the list of outputs;
and the entries are computed as in (1), with $\mathsf{copy}_d \circ \mathsf{new}_d$, whose entry at row $(i, i')$ is $\delta_{i i'}$, in place of $\mathsf{del}_d \circ \mathsf{equate}_d$.
For (3), the composite with $\sigma$ identifies the two variables of the swap crosswise with the two transposed entries, and $[\![\sigma]\!]_{\mathcal{I}}$ is a permutation matrix.
\end{proof}

\begin{lemma}\label{lem:subterm-entries}
Let $\mathcal{I}\colon \Sigma \to \mathsf{Mat}_{\mathbb{Q}_{\geq 0}}$ be an interpretation such that $\mathcal{I}(g)$ is substochastic for every function symbol $g$ of $\Sigma$,
let ${\func{f}}$ be a gs-monoidal $\Sigma$-diagram, and let $t$ be an algebraisation of ${\func{f}}$.
Then all entries of $[\![s]\!]_{\mathcal{I}}$ lie in $[0,1]$ for every sub-term $s$ of $t$.
\end{lemma}
\begin{proof}
Let $\func{s}$ be a $\Sigma$-diagram representing $\underline{\mathsf{dot}}(s)$, so that $[\![s]\!]_{\mathcal{I}} = [\![\func{s}]\!]_{\mathcal{I}}$ (\Cref{ssec:algebraisations}).
We show that $\func{s}$ arises from a gs-monoidal $\Sigma$-diagram by bending wires and transposing adjacent entries of its lists.
By \Cref{lem:bending}, every entry of $[\![\func{s}]\!]_{\mathcal{I}}$ is then an entry of the semantics of a gs-monoidal $\Sigma$-diagram, which is substochastic by \Cref{lem:substochastic}, and the entries of a substochastic matrix lie in $[0,1]$.

First, we record three properties of $\func{s}$ that follow from the gs-monoidality of ${\func{f}}$.
By \Cref{definition:composite-of-dot-diagrams,definition:tensor-product-of-dot-diagrams} and induction on the structure of $t$,
a $\Sigma$-diagram $\func{t}$ representing $\underline{\mathsf{dot}}(t)$ is obtained from $\func{s}$ by repeatedly forming tensor products and composites with other $\Sigma$-diagrams,
and these operations only identify input or output variables.
Hence there are a map $\iota$ from the variables of $\func{s}$ to those of $\func{t}$ and an injective map from the assignments of $\func{s}$ to those of $\func{t}$, compatible with $\iota$,
such that every variable $x$ of $\func{s}$ that occurs in neither list of $\func{s}$ is mapped to a variable $\iota(x)$ that occurs in neither list of $\func{t}$, is the image of $x$ only, and occurs only in the images of the assignments of $\func{s}$.
Since $\func{t} \equiv {\func{f}}$ is gs-monoidal, it follows that
(i) $\func{s}$ has no cyclic dependencies between variables, as a cycle in $\func{s}$ would map to one in $\func{t}$;
(ii) every variable of $\func{s}$ is defined by at most one assignment, in at most one position, as otherwise its image would be defined twice; and
(iii) every variable of $\func{s}$ that occurs in neither list of $\func{s}$ is defined by some assignment of $\func{s}$, since its image does not occur in $\mathsf{in}_{\func{t}}$, hence is defined by an assignment of $\func{t}$, and this assignment is the image of an assignment of $\func{s}$ defining the variable.

Second, we straighten $\func{s}$.
Call a variable of $\func{s}$ \emph{free} if no assignment defines it.
By (iii), every free variable occurs in $\mathsf{in}_{\func{s}}$ or in $\mathsf{out}_{\func{s}}$; for each free variable, choose one such occurrence, in $\mathsf{in}_{\func{s}}$ if there is one.
Let $\func{s}'$ be the $\Sigma$-diagram with the variables and assignments of $\func{s}$
whose list of inputs lists the free variables, each once, those with chosen occurrence in $\mathsf{in}_{\func{s}}$ first,
and whose list of outputs is $\mathsf{out}_{\func{s}}$ followed by $\mathsf{in}_{\func{s}}$, with the chosen occurrences removed from both.
Then $\func{s}'$ is gs-monoidal:
a free variable occurs exactly once in $\mathsf{in}_{\func{s}'}$ and is defined by no assignment,
every other variable does not occur in $\mathsf{in}_{\func{s}'}$ and is defined by exactly one assignment, in one position, by (ii),
and there are no cyclic dependencies by (i).

Third, we bend the wires back.
Bend down the last output of $\func{s}'$ repeatedly, until all outputs that came from $\mathsf{in}_{\func{s}}$ have become inputs;
by \Cref{lem:bending}, the result has the variables and assignments of $\func{s}$, the list of inputs $\mathsf{in}_{\func{s}'}$ followed by these entries in reverse order, and the list of outputs $\mathsf{out}_{\func{s}}$ without the chosen occurrences.
Transpose adjacent inputs until the free variables with chosen occurrence in $\mathsf{out}_{\func{s}}$ are the last inputs, and bend them up one after the other.
The result has the variables and assignments of $\func{s}$, and its lists of inputs and outputs have the same entries, with the same multiplicities, as $\mathsf{in}_{\func{s}}$ and $\mathsf{out}_{\func{s}}$;
transposing adjacent entries once more yields $\func{s}$.
\end{proof}

\subsection{Maximal arity and treewidth for DPPs}\label{sec:dpp-width}

\begin{lemma}\label{lem:dpp-width}
    Let $\func{f}$ be a DPP with root $R$, and let $k := \max_{v \in V(\func{f})} \mathsf{tw}(\func{f}_v')$. 
    Then $\mathsf{maxar}(\func{f}_v) \leq \max\{3, 2k + 2, k + 1 + |\mathsf{out}_{\func{f}_R}|\}$ for all $v \in V(\func{f})$. 
    In particular, if $\func{f}\colon (\,) \to (\mathbb{B})$, then $\mathsf{maxar}(\func{f}_v) \leq 2k + 3$ for all $v$.
\end{lemma}
\begin{proof}
    We use the well-known fact that every clique of a graph is contained in some bag of every tree decomposition of the graph, 
    so that a graph of treewidth $k'$ has no clique with more than $k' + 1$ vertices.
    Let $v \in V(\func{f})$. 
    The input variables of $\func{f}_v$ are pairwise distinct by Condition (1) of \Cref{definition:gs-monoidal}, 
    and pairwise adjacent in $\func{f}_v'$; hence $|\mathsf{in}_{\func{f}_v}| \leq \mathsf{tw}(\func{f}_v') + 1$. 
    Now consider an assignment $(\lett \mathbf{y} = g(\mathbf{x}))$ of $\func{f}_v$. 
    The entries of $\mathbf{y}$ are pairwise distinct, since every variable is defined at most once and in at most one position, 
    and they are pairwise adjacent in $\func{f}_v'$; hence $|\mathbf{y}| \leq \mathsf{tw}(\func{f}_v') + 1$. 
    If $g \in \mathsf{BoolStochSig}$, then $|\mathbf{x}| \leq 2$. 
    Otherwise, $g \in \Sigma_v$, and $|\mathbf{x}| = |\mathsf{in}_{\func{f}_u}| \leq \mathsf{tw}(\func{f}_u') + 1$ for $u := \mathrm{defn}_{\func{f}, v}(g)$ by the first part. 
    Finally, if $v \neq R$, then $v = \mathrm{defn}_{\func{f}, w}(g)$ for some parent $w$ of $v$ and some $g \in \Sigma_w$, 
    and $g$ occurs in some assignment $(\lett \mathbf{y} = g(\mathbf{x}))$ of $\func{f}_w$ by Condition (4) of \Cref{definition:hierarchical-dot-diagram}; 
    hence $|\mathsf{out}_{\func{f}_v}| = |\mathbf{y}| \leq \mathsf{tw}(\func{f}_w') + 1$ by the second part. 
    Together, these bounds prove the lemma: the arity of the function symbol of every assignment of $\func{f}_v$ is at most $\max\{3, 2k + 2\}$, and $|\mathsf{in}_{\func{f}_v}| + |\mathsf{out}_{\func{f}_v}|$ is at most $k + 1 + |\mathsf{out}_{\func{f}_R}|$ if $v = R$, and at most $2k + 2$ otherwise.
\end{proof}

\subsection{A lower bound on the acceptance probability}\label{sec:acceptance-probability-bound}

\begin{lemma}\label{lem:acceptance-probability-bound}
    Let $\func{f}\colon (\,) \to (\mathbb{B})$ be a DPP with acceptance probability $p_{\mathrm{acc}} > 0$, 
    let $D$ be the depth of its call graph, i.e.~the maximum number of edges on a path from the root to a leaf, 
    and let $N$ be as in \Cref{thm:algebraisation-general}. 
    Let $p_{\min}$ be the minimum of the numbers $p$ and $1 - p$ over all $\func{flip}(p)$ with $0 < p < 1$ occurring in some dot diagram appearing in $\func{f}$, 
    and $p_{\min} := 1$ if there is no such $\func{flip}(p)$. 
    Then $p_{\mathrm{acc}} \geq p_{\min}^{N^{D+1}}$.
\end{lemma}
\begin{proof}
    We first show that $\mathsf{unfold}'(\func{f}, v)$ has at most $(N-1)^{h(v) + 1} \leq N^{h(v)+1}$ assignments, 
    where $h(v)$ is the height of $v$ in the call graph, by induction on $h(v)$. 
    For a leaf $v$, $\mathsf{unfold}'(\func{f}, v) = \func{f}_v$ has at most $N - 1$ assignments. 
    Otherwise, every assignment of $\func{f}_v$ is either kept, or replaced by the assignments of $\mathsf{unfold}'(\func{f}, u)$ for a child $u$ of $v$ with $h(u) \leq h(v) - 1$, 
    of which there are at most $(N-1)^{h(v)}$ by the induction hypothesis; 
    hence $\mathsf{unfold}'(\func{f}, v)$ has at most $(N-1)\cdot (N-1)^{h(v)}$ assignments.
    In particular, at most $N^{D+1}$ assignments of $\mathsf{unfold}(\func{f})$ have a function symbol of the form $\func{flip}(p)$.

    Now let $t$ be an algebraisation of $\mathsf{unfold}(\func{f})$, so that $[\![t]\!]_{\mathcal{S}} = [\![\func{f}]\!]_{\mathcal{S}}$ (\Cref{ssec:algebraisations}),
    and for a hypergraph term $s$, let $\varphi(s)$ be the number of occurrences of symbols $\func{flip}(p)$ in $s$.
    We show by induction on the structure of $s$ that every nonzero entry of $[\![s]\!]_{\mathcal{S}}$ is at least $p_{\min}^{\varphi(s)}$, for every sub-term $s$ of $t$.
    If $s$ is a symbol $\func{flip}(p)$, then its nonzero entries are among $p$ and $1 - p$, which are at least $p_{\min}$ if $0 < p < 1$, and equal to $1$ otherwise.
    The matrices of the other function symbols of $\mathsf{BoolStochSig}$ and of the symbols $\mathrm{id}$, $\mathrm{swap}$, $\mathrm{copy}$, $\mathrm{del}$, $\mathrm{equate}$ and $\mathrm{new}$ have entries $0$ and $1$ only.
    Every nonzero entry of a Kronecker product $[\![s_1]\!]_{\mathcal{S}} \otimes [\![s_2]\!]_{\mathcal{S}}$ is a product of nonzero entries of the two factors,
    and every nonzero entry $\sum_k A_{ik} B_{kj}$ of a product of two matrices with non-negative entries is at least one of its nonzero summands $A_{ik} B_{kj}$;
    in both cases, the induction hypothesis gives the bound $p_{\min}^{\varphi(s_1) + \varphi(s_2)} = p_{\min}^{\varphi(s)}$.

    Finally, $\varphi(t)$ is the number of assignments of $\mathsf{unfold}(\func{f})$ whose function symbol is of the form $\func{flip}(p)$:
    by \Cref{definition:composite-of-dot-diagrams,definition:tensor-product-of-dot-diagrams}, the assignments of the dot diagram associated to a term correspond to the occurrences of function symbols in the term,
    and equivalent dot diagrams have the same assignments up to renaming and reordering.
    Since $p_{\mathrm{acc}} = p + q > 0$, one of $p$ and $q$ is a nonzero entry of $[\![t]\!]_{\mathcal{S}}$, so that $p_{\mathrm{acc}} \geq p_{\min}^{\varphi(t)} \geq p_{\min}^{N^{D+1}}$.
\end{proof}

\subsection{Proof of \Cref{thm:fixed-parameter-tractable-inference}}\label{sec:proof-of-thm-fixed-parameter-tractable-inference}

Let $a := \max_{v\in V({\func{f}})} \mathsf{maxar}({\func{f}}_v)$ be the maximal-arity parameter of \Cref{thm:algebraisation-general,thm:semantics-in-matrix-categories}. 
By \Cref{lem:dpp-width}, $a \leq 2k + 3$, so that $2^{O(k + a)} = 2^{O(k)}$.
Every dot diagram appearing in ${\func{f}}$ is gs-monoidal and hence has no isolated variables, so that \Cref{thm:semantics-in-matrix-categories} applies to ${\func{f}}$.
Applying it to ${\func{f}}$ and the substochastic matrix interpretation $\mathcal{S}\colon \mathsf{BoolStochSig} \to \mathsf{Mat}_{\mathbb{Q}_{\geq 0}}$, 
we obtain, in time $O(LMN\cdot 2^{O(k)})$, an arithmetic circuit $C$ over $\mathbb{Q}_{\geq 0}$ such that:
\begin{enumerate}
    \item The circuit $C$ has two roots $R_0, R_1$ and 
        $ [\![{\func{f}}]\!]_\mathcal{S} = \begin{pmatrix}
                [\![C]\!]_{R_0} &
                [\![C]\!]_{R_1}
        \end{pmatrix}^\intercal.$
    \item The size $|V(C)|$ of $C$ is in $O(MN\cdot 2^{O(k)})$.
\end{enumerate}
Let $q:= [\![C]\!]_{R_0}$ and $p := [\![C]\!]_{R_1}$, so that $p_{\mathrm{acc}} = p + q$ and $p_{\func{f}} = p/(p+q)$.

\emph{All values in $C$ lie in $[0,1]$.} 
Recall from the proof of \Cref{thm:semantics-in-matrix-categories} that $C$ is obtained from the DAG-representation of an algebraisation $t$ of $\mathsf{unfold}({\func{f}})$: 
every vertex of $C$ computes either an entry of the matrix $[\![s]\!]_{\mathcal{S}}$ for a sub-term $s$ of $t$, 
or a product of two such entries, or a partial sum of the sum of such products defining an entry of a matrix product. 
By \Cref{lem:unfold-gs-monoidal}, $\mathsf{unfold}({\func{f}})$ is gs-monoidal, and all matrices $\mathcal{S}(g)$ are substochastic; 
hence, by \Cref{lem:subterm-entries}, all entries of all $[\![s]\!]_{\mathcal{S}}$ lie in $[0,1]$. 
Since all entries are non-negative, a product of two entries and a partial sum of a sum of such products are at most the entry they contribute to. 
Consequently, $0 \leq [\![C]\!]_v \leq 1$ for all $v\in V(C)$.

\emph{Evaluating $C$ with truncated binary arithmetic.} 
Fix $b \in \mathbb{N}$ and let $b' := |V(C)| + b$. 
We evaluate $C$ using numbers with $b'$ fractional binary digits: 
each leaf value (an entry of one of the matrices $\mathcal{S}(g)$, i.e.~a rational number in $[0,1]$) is replaced by its binary expansion truncated after $b'$ fractional digits,\footnote{
We assume that the first $b'$ binary digits of every probability $p$ occurring in ${\func{f}}$ can be computed in $O(b'^2)$ time, 
which is the case if $p$ is given as a fraction of integers with at most $b'$ bits. 
In general, a factor of the bit-size of the probabilities occurring in ${\func{f}}$ has to be added to the time bound.
}
and every sum and product is computed exactly and then truncated after $b'$ fractional digits. 
Let $\tilde{c}_v$ be the value computed at $v \in V(C)$ in this way, and let $\delta_v := [\![C]\!]_v - \tilde{c}_v$. 
Since truncation rounds down and all values are non-negative, 
induction along $C$ shows $0 \leq \tilde{c}_v \leq [\![C]\!]_v \leq 1$, so that $\delta_v \geq 0$. 
Moreover, $\delta_v \leq 2^{-b'}$ at every leaf $v$, 
and at a vertex $v$ with children $u, w$, we have $\delta_v \leq \delta_u + \delta_w + 2^{-b'}$: 
for a sum, this is clear, 
and for a product, $[\![C]\!]_u [\![C]\!]_w - \tilde{c}_u \tilde{c}_w = [\![C]\!]_u \delta_w + \tilde{c}_w \delta_u \leq \delta_w + \delta_u$. 
By induction on the height $h(v)$ of $v$ (the maximum number of edges on a path from $v$ to a leaf), 
it follows that $\delta_v \leq (2^{h(v) + 1} - 1)\cdot 2^{-b'} < 2^{h(v)+1-b'} \leq 2^{|V(C)| - b'} = 2^{-b}$, 
using $h(v) \leq |V(C)| - 1$. 
In particular, the computed values $\tilde{p}$ and $\tilde{q}$ at the roots satisfy $0 \leq p - \tilde{p} \leq 2^{-b}$ and $0 \leq q - \tilde{q}\leq 2^{-b}$.
The evaluation takes $O(|V(C)|)$ arithmetic operations on numbers with at most $b' + 1$ binary digits, 
each of which takes $O(b'^2)$ time (using the schoolbook algorithm for multiplication for simplicity).

\emph{Approximating $p_{\func{f}}$.} 
Let $s := p + q$, $\tilde{s} := \tilde{p} + \tilde{q}$, and $\varepsilon := 2^{-b}$, so that $0 \leq s - \tilde{s} \leq 2\varepsilon$. 
If $\varepsilon \leq s/4$, then $\tilde{s} \geq s/2 > 0$, and 
\[ \Big| \frac{\tilde{p}}{\tilde{s}} - \frac{p}{s}\Big| = \frac{|\tilde{p}s - p\tilde{s}|}{s\tilde{s}} = \frac{|(\tilde{p} - p) q + p(q - \tilde{q})|}{s \tilde{s}} \leq \frac{\varepsilon q + \varepsilon p}{s\tilde{s}} = \frac{\varepsilon}{\tilde{s}} \leq \frac{2\varepsilon}{s}, \]
using $\tilde{p}s - p\tilde{s} = \tilde{p}(p + q) - p(\tilde{p} + \tilde{q}) = \tilde{p} q - p\tilde{q}$.

\emph{Choosing $b$.} 
Since $s = p_{\mathrm{acc}}$ is not known in advance, we proceed as follows. 
For $j = 1, 2, 3, \dots$, evaluate $C$ with $b := d + j + 1$ as above, 
obtaining $\tilde{p}_j$, $\tilde{q}_j$ and $\tilde{s}_j := \tilde{p}_j + \tilde{q}_j$ with $0 \leq s - \tilde{s}_j \leq 2\cdot 2^{-(d + j + 1)} \leq 2\cdot 2^{-j}$, 
and stop at the first $j$ with $\tilde{s}_j \geq 2\cdot 2^{-j}$. 
This happens at the latest when $2^{-j} \leq s/4$, since then $\tilde{s}_j \geq s - 2\cdot 2^{-j} \geq s/2 \geq 2\cdot 2^{-j}$; 
hence the procedure stops after at most $\lceil \log_2(s^{-1})\rceil + 2 \leq \ell + 2$ evaluations, where $\ell := \log_2(2p_{\mathrm{acc}}^{-1}) = 1 + \log_2(p_{\mathrm{acc}}^{-1})$, 
and $b \leq d + \ell + 3$ throughout. 
Upon stopping, $s \geq \tilde{s}_j \geq 2\cdot 2^{-j}$, 
so that $\varepsilon = 2^{-(d + j + 1)} = 2^{-(d+2)} \cdot 2 \cdot 2^{-j} \leq 2^{-(d+2)}\, s \leq s/4$, 
and $\tilde{p}_j/\tilde{s}_j$ approximates $p_{\func{f}} = p/s$ with absolute error at most $2\varepsilon/s \leq 2^{-(d+1)}$. 
The division $\tilde{p}_j/\tilde{s}_j$, carried out to $d + 2$ fractional binary digits, adds an error of at most $2^{-(d+2)}$ and takes $O(b'^2)$ time; 
the result approximates $p_{\func{f}}$ with absolute error at most $2^{-(d+1)} + 2^{-(d+2)} < 2^{-d}$. 
To obtain an error of at most $2^{-(d+1)}$ as stated, replace $d$ by $d + 1$ throughout, which does not affect the asymptotic bound.

\emph{Time complexity.} 
In total, we perform $O(\ell)$ evaluations of $C$, each of which takes $O(|V(C)| \cdot (|V(C)| + d + \ell)^2)$ time. 
Using $|V(C)| = O(MN\cdot 2^{O(k)})$ and $(x + y + z)^2 \leq 9 x^2 y^2 z^2$ for $x, y, z \geq 1$, 
this is $O((MN)^3 \cdot d^2 \cdot \ell^3 \cdot 2^{O(k)})$ time. 
Together with the $O(LMN\cdot 2^{O(k)})$ time needed to compute $C$, 
and since $L \geq 1$, we obtain the time bound $O((LMN)^3\cdot d^2 \cdot \log_2^3(2p_{\mathrm{acc}}^{-1}) \cdot 2^{O(k)})$ of the theorem. \qed
\section{Monoidal width and hypergraph terms}\label{sec:relation-to-monoidal-width}

The idea of finding `good' decompositions into composites and tensor products is conceptually very similar to the \emph{monoidal decompositions} and \emph{monoidal width} of \cite{di2023monoidal}. 
However, their notion of width also substantially differs from ours.  
Intuitively, the width of a monoidal decomposition does not penalise tensor products, 
whereas tensor products may generally increase the width of a hypergraph term in our sense.

More precisely,  let $C=\underline{\mathrm{Term}}(\Sigma)$
be the monoidal category of hypergraph terms over a monoidal signature $\Sigma$
(see \Cref{sec:algebraisations})
and let $\mathcal{A}\subseteq \mathrm{Mor}(C)$
be the set of all basic hypergraph terms, i.e.~those hypergraph terms $t$ with $|t| = 1$.
We define a \emph{weight function} $\mathsf{w}$ in the sense of \cite[Definition 2.2]{di2023monoidal} 
as $\mathsf{w}(f) = |\mathrm{dom}(f)| + |\mathrm{codom}(f)|$ for `atomic morphisms' $f\in \mathcal{A}$
and $\mathsf{w}(X) = |X|$ on objects, where $|X|$ is the length of the list $X$.
With these definitions in place, it follows directly from the definition of the width of a monoidal decomposition~\cite[Definition 2.3]{di2023monoidal} that for all hypergraph terms $t \in \mathrm{Term}(\Sigma)$,
$$\mathsf{wd}(\,\underline{t}\,) \leq \mathsf{width}(t), $$ 
where $\mathsf{wd}(\underline{t})$ is the width of the obvious monoidal decomposition of $\underline{t}$ (according to the structure of $t$). 
However, the converse inequality does not generally hold, as can be seen from the term $t= \mathrm{id}_{\mathbb{B}} \otimes \mathrm{id}_{\mathbb{B}}$ which has a width of $4$, 
but the corresponding monoidal decomposition of $\underline{t}$ has a width of only $2$. 

We do not know if a different connection of our ideas 
to those of \cite{di2023monoidal} could be found, 
or if this connection could be used to leverage the theory of monoidal width 
in our algorithmic setting; 
exploring this possibility would require further research.

\section{Database queries as dot diagrams}\label{sec:queries-as-dot-diagrams}

To make the translation from conjunctive queries to dot diagrams precise, 
we first recall some standard definitions from relational database theory;
see, for example, \cite{grohe2001evaluation}.

\begin{definition}
    Let $\sigma$ be a relational signature in the sense of first-order logic, 
    i.e.~a tuple 
    $$\sigma = \left((R_i^\sigma)_{i\in I}, \mathrm{arity}\colon I\to \mathbb{N}\right)$$ 
    consisting of a finite family of relation symbols
    and a function assigning to each relation symbol its arity.
    In the context of databases, $\sigma$ is also called the \emph{schema}.

    A \emph{relational database}, or \emph{instance}, over $\sigma$ is a $\sigma$-structure, 
    i.e.~a pair $\mathcal{A} = (D, (R_i)_{i\in I})$ consisting of a finite set $D$, the \emph{domain}, 
    and for each index $i\in I$, a relation $R_i \subseteq D^{\mathrm{arity}(i)}$.

    A \emph{conjunctive query}, or \emph{select-project-join} query,
    with $k$ free variables is a first-order formula over $\sigma$ of the form
        $$ \phi(x_1,\dots,x_k) =\exists y_1 \dots \exists y_n \;\left(R_{i_1}^\sigma(z^{1}_1, \dots, z^1_{\mathrm{arity}(i_1)}) \,\land \,\dots \,\land\, R_{i_m}^\sigma(z^{m}_1, \dots, z^m_{\mathrm{arity}(i_m)})\right), $$
    for some $n,m\in \mathbb{N}$, where each $z^j_l \in \{x_1,\dots,x_k,y_1,\dots,y_n\}$,
    i.e.~$\phi$ is an existentially quantified conjunction of atoms.
    The \emph{evaluation} of a conjunctive query $\phi(x_1,\dots,x_k)$ on a relational database $\mathcal{A} = (D, (R_i)_{i\in I})$ is the relation
    $$ \phi(\mathcal{A}) := \{(x_1, \dots, x_k) \in D^k\mid \mathcal{A} \models \phi(x_1, \dots, x_k) \}. $$
\end{definition}

\begin{example}\label{example:booking-query}
    Let $\sigma$ be the relational signature with three relation symbols
    $\mathsf{Bookings}$, $\mathsf{Hotels}$, and $\mathsf{Cities}$ of arity $3$, $2$, and $3$, respectively.
    These relation symbols are interpreted as follows:
    \begin{itemize}
        \item $\mathsf{Bookings}(u, h, d)$: user $u$ has a booking at hotel $h$ on date $d$;
        \item $\mathsf{Hotels}(h, p)$: hotel $h$ has postcode $p$;
        \item $\mathsf{Cities}(c, p, w)$: postcode $p$ is located in city $c$ in country $w$.
    \end{itemize}
    In practice, a database instance over $\sigma$ consists of one table for each relation symbol, 
    recording those tuples satisfying the relation.
    The following conjunctive query asks \emph{for which pairs $(u,c)$ of a user and a city does the user have a booking at some hotel located in that city}:
\[
\phi(u,c) \;=\; \exists h \exists p \exists d \exists w \;\bigl(
\mathsf{Bookings}(u,h,d)
\;\land\;
\mathsf{Hotels}(h,p)
\;\land\;
\mathsf{Cities}(c,p,w)
\bigr).
\]
    Equivalently, as a SQL query:
\begin{center}
\begin{minipage}{0.72\linewidth}
\begin{lstlisting}[style=sql]
SELECT u, c
FROM   Bookings, Hotels, Cities
WHERE  Bookings.h = Hotels.h
  AND  Hotels.p = Cities.p
\end{lstlisting}
\end{minipage}
\end{center}
\end{example}

The next step is to associate a dot diagram to every conjunctive query.
In the following, $\mathsf{FinRel}$ denotes the strict hypergraph category of finite sets and relations between them.
For our purposes, $\mathsf{FinRel}$ can be defined as the matrix hypergraph category $\mathsf{Mat}_{\mathbb{B}}$, 
where $(\mathbb{B}, \lor, \land)$ is the Boolean semiring. 
We identify matrices over the Boolean semiring with finite relations in the obvious way.

\begin{definition}
    Let $\sigma = \left((R_i^\sigma)_{i\in I}, \mathrm{arity}\colon I\to \mathbb{N}\right)$ be a relational signature, 
    and let 
    $$ \phi(x_1,\dots,x_k) =\exists y_1 \dots \exists y_n \;\left(R_{i_1}^\sigma(z^{1}_1, \dots, z^1_{\mathrm{arity}(i_1)}) \,\land \,\dots \,\land\, R_{i_m}^\sigma(z^{m}_1, \dots, z^m_{\mathrm{arity}(i_m)})\right) $$
    be a conjunctive query over $\sigma$ with $k$ free variables
    $x_1, \dots, x_k$. 
    
    The \emph{monoidal signature $\Sigma$ associated to $\sigma$}
    has a single sort $\mathbb{D}$ and for each relation symbol $R_i$ 
    one function symbol $\func{R}_i\colon \mathbb{D}^{\mathrm{arity}(i)} \to (\,)$, where $\mathbb{D}^{r}$ denotes the list of $r$ copies of $\mathbb{D}$.
    
    The \emph{dot diagram $\func{\phi}^{\bullet}$ associated to $\phi$} 
    is the following dot diagram:
    \begin{align*}
        &\func{\phi}^{\bullet}(x_1, \dots, x_k) := \\
        &\quad \func{R}_{i_1}(z^{1}_1, \dots, z^1_{\mathrm{arity}(i_1)}); \\
        &\quad \dots \\
        &\quad \func{R}_{i_m}(z^{m}_1, \dots, z^m_{\mathrm{arity}(i_m)});
    \end{align*}
    Finally, the \emph{interpretation $\mathcal{A}^{\bullet}\colon \Sigma \to \mathsf{FinRel}$  
    associated to a relational database $\mathcal{A} = (D, (R_i)_{i\in I})$ over $\sigma$} associates to the single sort $\mathbb{D}$ the domain $D$ and to each function symbol $\func{R}_{i}$ in the monoidal signature $\Sigma$
    the relation $R_i \subseteq D^{\mathrm{arity}(i)}$, which is a morphism $D^{\mathrm{arity}(i)} \to 1$ in $\mathsf{FinRel}$, where $1$ denotes a one-element set, the tensor unit of $\mathsf{FinRel}$.
\end{definition}

\begin{example}
    Continuing \Cref{example:booking-query},
    the monoidal signature $\Sigma$ associated to $\sigma$ has a single sort $\mathbb{D}$ and three function symbols:
    \begin{align*}
        \func{Bookings}&\colon (\mathbb{D}, \mathbb{D}, \mathbb{D}) \to (\,), &
        \func{Hotels}&\colon (\mathbb{D}, \mathbb{D}) \to (\,), &
        \func{Cities}&\colon (\mathbb{D}, \mathbb{D}, \mathbb{D}) \to (\,).
    \end{align*}
    The dot diagram $\func{\phi}^\bullet$ associated to $\phi$ is:
    \begin{align*}
        &\func{\phi}^\bullet(u, c) := \\
        &\quad \func{Bookings}(u, h, d); \\
        &\quad \func{Hotels}(h, p); \\
        &\quad \func{Cities}(c, p, w);
    \end{align*}
    The variable set of $\func{\phi}^\bullet$ is $\{u, c, h, d, p, w\}$, all of sort $\mathbb{D}$.
    The input list is $(u, c)$ (the free variables of $\phi$), the output list is empty,
    and the variables $h, d, p, w$ are internal (existentially quantified).
    Under the interpretation $\mathcal{A}^\bullet$,
    the semantics $[\![\func{\phi}^\bullet]\!]_{\mathcal{A}^\bullet}$
    is the Boolean matrix (relation) $\phi(\mathcal{A}) \subseteq D^2$
    recording precisely those pairs $(u, c)$ satisfying the query.
\end{example}

The following formula for the entries of the semantics of a dot diagram in $\mathsf{Mat}_R$ is a variant of \cite[Theorem 5.1]{kissinger2014finite}; 
we include a proof based on \Cref{lem:wiring-tensor-factorisation}.
For a $\Sigma$-diagram $\func{f}$ and an interpretation $\mathcal{I}\colon \Sigma \to \mathsf{Mat}_R$, 
let $\mathrm{Idx}(\func{f})$ be the set of all maps $\phi$ assigning to each variable $x \in X_{\func{f}}$ an index $\phi(x) \in \{1, \dots, \mathcal{I}(\mathsf{sort}(x))\}$. 
For a list $\mathbf{l} = (l_1, \dots, l_n)$ of variables, $\phi_*(\mathbf{l}) := (\phi(l_1), \dots, \phi(l_n))$ is then a tuple of indices, 
which we identify with the index of the basis vector $e_{\phi(l_1)} \otimes \dots \otimes e_{\phi(l_n)}$ of $\mathcal{I}(\mathsf{sort}(l_1)) \otimes \dots \otimes \mathcal{I}(\mathsf{sort}(l_n))$.

\begin{lemma}\label{lem:entry-formula}
    Let $\Sigma$ be a monoidal signature, let $R$ be a semiring, let $\mathcal{I}\colon \Sigma\to \mathsf{Mat}_R$ be an interpretation, 
    and let $\func{f}$ be a $\Sigma$-diagram. 
    For each $a \in A_{\func{f}}$, write $a = (\lett \mathbf{y}_a = g_a(\mathbf{x}_a))$.
    Then, for all index tuples $\mathbf{b}$ and $\mathbf{c}$ for the lists $\mathsf{out}_{\func{f}}$ and $\mathsf{in}_{\func{f}}$, 
    \[
        \big([\![\func{f}]\!]_{\mathcal{I}}\big)_{\mathbf{b}, \mathbf{c}} 
        = \sum_{\substack{\phi \in \mathrm{Idx}(\func{f}) \\ \phi_*(\mathsf{in}_{\func{f}}) = \mathbf{c},\; \phi_*(\mathsf{out}_{\func{f}}) = \mathbf{b}}} 
        \;\prod_{a\in A_{\func{f}}} \mathcal{I}(g_a)_{\phi_*(\mathbf{y}_a), \phi_*(\mathbf{x}_a)}.
    \]
\end{lemma}
\begin{proof}
    Both sides of the formula depend only on the equivalence class of $\func{f}$:
    the left-hand side by definition, and the right-hand side because renaming the variables and reordering the assignments does not change the sum.
    By \Cref{lem:wiring-tensor-factorisation}, $\func{f}$ is equivalent to a composite of two pure wirings and a tensor product of dot diagrams $\langle g_a \rangle$ and identities,
    and every pure wiring is equivalent to the dot diagram associated to a term built from $\mathrm{id}$, $\mathrm{swap}$, $\mathrm{copy}$, $\mathrm{del}$, $\mathrm{equate}$ and $\mathrm{new}$ by composition and tensor product
    (\Cref{lem:decompose-pure-wiring} and the construction in the proof of \Cref{cor:algbrise-pure-wiring}, in which an isolated variable $x$ contributes the factor $E_x = \mathrm{del} \circ \mathrm{new}$).
    Hence it suffices to show that the formula holds for the dot diagrams $\langle g \rangle$, $g \in \Sigma$, and for the six dot diagrams $\mathsf{id}$, $\mathsf{swap}$, $\mathsf{copy}$, $\mathsf{del}$, $\mathsf{equate}$ and $\mathsf{new}$,
    and that it is preserved by tensor products and composites.

    The dot diagram $\langle g \rangle$ has pairwise distinct input variables, pairwise distinct output variables, and the single assignment $\lett \mathsf{out}_{\langle g \rangle} = g(\mathsf{in}_{\langle g \rangle})$;
    hence exactly one $\phi$ satisfies the constraints, and both sides equal $\mathcal{I}(g)_{\mathbf{b}, \mathbf{c}}$.
    For the six dot diagrams without assignments, the product is empty,
    and the formula states that the entry $(\mathbf{b}, \mathbf{c})$ is the number of $\phi$ with $\phi_*(\mathsf{in}) = \mathbf{c}$ and $\phi_*(\mathsf{out}) = \mathbf{b}$,
    interpreted as a sum of ones in $R$.
    This is checked directly against the definition of $\mathsf{Mat}_R$ 
    (for instance, $\mathsf{copy}$ has one variable $x$, list of inputs $(x)$ and list of outputs $(x,x)$, 
    and the number of $\phi$ with $\phi(x) = c$ and $(\phi(x), \phi(x)) = \mathbf{b}$ is $1$ if $\mathbf{b} = (c, c)$ and $0$ otherwise, in accordance with $\mathsf{copy}\cdot e_c = e_c \otimes e_c$).

    The formula is preserved by tensor products:
    the variable set of $\func{f} \otimes \func{g}$ is the disjoint union of the variable sets,
    so that a map in $\mathrm{Idx}(\func{f} \otimes \func{g})$ is a pair of maps in $\mathrm{Idx}(\func{f})$ and $\mathrm{Idx}(\func{g})$,
    the constraints on the pair are the constraints on its two components,
    and the product over the assignments of $\func{f} \otimes \func{g}$ is the product of the products over the assignments of $\func{f}$ and of $\func{g}$.
    Hence the right-hand side for $\func{f} \otimes \func{g}$ at $(\mathbf{b}\mathbf{b}', \mathbf{c}\mathbf{c}')$ is the product of the right-hand sides for $\func{f}$ at $(\mathbf{b}, \mathbf{c})$ and for $\func{g}$ at $(\mathbf{b}', \mathbf{c}')$,
    in accordance with the Kronecker product.
    The formula is preserved by composites:
    the entry $(\mathbf{b}, \mathbf{c})$ of $[\![\func{g} \circ \func{f}]\!]_{\mathcal{I}} = [\![\func{g}]\!]_{\mathcal{I}} \cdot [\![\func{f}]\!]_{\mathcal{I}}$
    is the sum, over all index tuples $\mathbf{j}$ for $\mathsf{out}_{\func{f}}$, of the product of the entries $(\mathbf{b}, \mathbf{j})$ of $[\![\func{g}]\!]_{\mathcal{I}}$ and $(\mathbf{j}, \mathbf{c})$ of $[\![\func{f}]\!]_{\mathcal{I}}$.
    By the formulas for $\func{f}$ and $\func{g}$, this is the sum over all pairs $(\phi, \phi') \in \mathrm{Idx}(\func{f}) \times \mathrm{Idx}(\func{g})$
    with $\phi_*(\mathsf{in}_{\func{f}}) = \mathbf{c}$, $\phi_*(\mathsf{out}_{\func{f}}) = \phi'_*(\mathsf{in}_{\func{g}})$ and $\phi'_*(\mathsf{out}_{\func{g}}) = \mathbf{b}$
    of the product of the products for $\phi$ over $A_{\func{f}}$ and for $\phi'$ over $A_{\func{g}}$.
    By \Cref{definition:composite-of-dot-diagrams}, a pair $(\phi, \phi')$ defines a map on the classes of variables of $\func{g} \circ \func{f}$
    precisely when $\phi_*(\mathsf{out}_{\func{f}}) = \phi'_*(\mathsf{in}_{\func{g}})$;
    so these pairs are in bijection with the maps $\psi \in \mathrm{Idx}(\func{g}\circ\func{f})$ with $\psi_*(\mathsf{in}_{\func{g}\circ\func{f}}) = \mathbf{c}$ and $\psi_*(\mathsf{out}_{\func{g}\circ\func{f}}) = \mathbf{b}$,
    and the product for $\psi$ over $A_{\func{g} \circ \func{f}}$ is the product of the products for $\phi$ and for $\phi'$.
\end{proof}

We can now show that evaluating a conjunctive query is indeed the same as computing the semantics of the corresponding dot diagram.

\begin{lemma}
    Let $\sigma$ be a relational signature, 
    let $\phi$ be a conjunctive query over $\sigma$, 
    and let $\mathcal{A}$ be a relational database.

    Then the evaluation of $\phi$ on $\mathcal{A}$ is the semantics of the dot diagram $\func{\phi}^{\bullet}$ under the interpretation $\mathcal{A}^{\bullet}$.
        $$ \phi(\mathcal{A}) = [\![\func{\phi}^{\bullet}]\!]_{\mathcal{A}^{\bullet}}. $$
\end{lemma}
\begin{proof}
    This is \Cref{lem:entry-formula} for the Boolean semiring:
    the sum over the internal variables becomes the chain of existential quantifiers,
    and the product over the assignments becomes the conjunction of the atoms of $\phi$.
\end{proof}

\newpage
\thispagestyle{empty}

\subparagraph{Revision Notice} \ \\

\sffamily

\noindent This is a revised version of the eponymous paper that appeared in the proceedings of LICS 2026 (LIPIcs, volume 380, \url{https://www.dagstuhl.de/dagpub/978-3-95977-434-5}, published in July 2026), in which several errors have been corrected.
The two main results (Theorems 32 and 55 of the published version) remain valid up to minor adjustments, 
but several definitions, lemmas and proofs had to be revised more substantially.
The numbering of the published version is used in the following description; the numbering of the revised version differs.

\begin{enumerate}
\item \emph{Definition of DPPs and their semantics.}
The published version claimed that substochastic matrices 
form a hypergraph category $\mathsf{BoolSubStoch}$ (Definition 18).
This is incorrect: substochastic matrices do not form a hypergraph category as stated, 
since the morphism ``$\mathsf{new}$'' (an all-ones column matrix) is not substochastic.
The deeper underlying problem is that Definition 14 allowed arbitrary 
hierarchical dot diagrams over $\mathsf{BoolStochSig}$ as DPPs, 
whereas to allow unnormalised semantics in substochastic matrices, 
DPPs must impose conditions on data flow that make DPPs behave like ordinary straight-line programs.
In the revised version, the semantics of DPPs takes values in $\mathsf{Mat}_{\mathbb{Q}_{\geq 0}}$,
DPPs are defined as hierarchical \emph{gs-monoidal} dot diagrams 
(an added definition: every variable is either an input or is defined exactly once, 
and there are no cyclic dependencies; this holds for all programs in the usual sense),
a new remark relates gs-monoidal dot diagrams to the gs-monoidal string diagrams of Fritz and Liang,
and a new lemma in Section 6.2 shows that the semantics of a DPP is a substochastic matrix, 
so that the acceptance probability (Definition 53) is well defined.
Discrete probabilistic programs in the usual sense, e.g. of the language \texttt{Dice} of Holtzen et al.,
satisfy these newly imposed conditions, so that the intended application is not affected.
Example 20 has been adjusted accordingly.
\item \emph{Dot diagrams (Definition 9).}
The condition that every variable occurs in the lists of inputs or outputs 
or in some assignment has been removed, 
since it is not preserved by composition; 
isolated variables are allowed, as in the original definition of Kissinger.
The definition of the equivalence of dot diagrams compared sets of assignments, 
which does not respect multiplicities; 
it now uses an explicit bijection between the lists of assignments.
\item \emph{The algebraisation algorithm (Section 5).}
The statement of the ``main result 3'' in the introduction has been adjusted to 
require a \emph{finite} monoidal signature (a signature containing only finitely many function symbols).
The bounds of the general precise statements, Theorems 32 and 49, do not hold in terms of the treewidth $k$ of the primal graph alone:
a variable that occurs several times in the same assignment makes the width of every algebraisation
at least the arity of the function symbol of that assignment, without increasing the treewidth.
The revised statements therefore involve, besides $k$, the \emph{maximal arity} $a$ of the dot diagram
(the maximal number of in- and outputs, counted together, of the diagram and of the function symbols occurring in it),
and their bounds are stated in terms of $k + a$; Theorem 52 is stated with $k + a$ accordingly.
The maximal arity takes the place of what was formerly called the 
``width of a dot diagram'' (Definition 37), 
which counted inputs and outputs separately.
This was a confusing term since it suggests a closer relation to the width of hypergraph terms;
the maximum arity of the revised version is conceptually clearer.
Definition 37 has been removed, and Lemmas 38 and 43 and Corollary 39 
are stated in terms of the maximal arity, with adjusted constants.
For a signature of bounded arity, $a$ is a constant; 
for DPPs, $a$ is bounded in terms of the treewidth
(a new lemma in the appendix of the full version), so that Theorem 55 
is essentially unchanged and
stated in terms of the treewidth as before.
Lemma 35 was stated with the copy-and-delete-type factor before the equate-and-select-type factor, 
in which order it does not hold, and its proof was incorrect; 
the order of the factors has been reversed and the proof replaced by an explicit construction.
Lemma 38 did not hold for dot diagrams with variables that occur in no assignment; 
an identity factor for these variables has been added, which changes the constant in Corollary 39.
Lemma 47 and its proof neglected the hyperedges of the dependency hypergraph corresponding 
to the lists of inputs and outputs, and property (2) of Lemma 47 was incorrect; 
it has been replaced by a bound in terms of the width of the branch decomposition and the maximal arity of the dot diagram, 
and property (4) now reads $|V({\func{f}}^*)| \leq |A_{\func{f}}| + 1$. 
This changes the constants in Theorem 48 and in the proof of Theorem 32.
Condition (4) of Definition 12 has been strengthened to exclude unused function symbols, 
and Definition 45 now counts assignments with multiplicity.
Since isolated variables are now allowed, Theorems 32, 48, 49 and 52, Lemmas 34, 43 and 47, 
Corollaries 36 and 39, and the time bound of Lemma 38 are stated for dot diagrams 
without isolated variables 
(for hierarchical dot diagrams: such that no dot diagram appearing in them has isolated variables), 
as all dot diagrams of the published version were. 
This is no loss of generality, as explained at the beginning of Section 5.
The same statements now also exclude the degenerate case of dot diagrams without inputs and outputs
all of whose function symbols have no inputs and outputs, for which the stated bounds would vanish.
\item \emph{Inference (Section 6.2).}
The proof of Theorem 55 assumed that all intermediate results of the computation lie in $[0,1]$, 
because they are entries of substochastic matrices. 
This is not the case for arbitrary hypergraph terms. 
The revised proof, in the appendix of the full version, shows that this property holds for all sub-terms of an algebraisation 
of a gs-monoidal dot diagram.
Theorem 55 now assumes $p_{\mathrm{acc}} > 0$ and is stated with $\log_2(2p_{\mathrm{acc}}^{-1})$ in place of $\log(p_{\mathrm{acc}}^{-1})$ (equivalently: $1 + \log_2(p_{\mathrm{acc}}^{-1})$), 
which is needed to cover the case $p_{\mathrm{acc}} = 1$.
The claim that the inverse acceptance probability is at most exponential in $LMN$ 
when the call graph has bounded depth $D$ was not proved and is incorrect as stated; 
it has been replaced by the bound $\log_2(p_{\mathrm{acc}}^{-1}) \leq N^{D+1}\log_2(p_{\min}^{-1})$, 
which is proved in the appendix of the full version, and footnote 2 has been adjusted accordingly.
Theorem 52 now makes explicit that the constants depend on the interpretation.
\item \emph{Minor corrections.} Typos have been corrected in Example 8 ($\mathsf{BoolSig}$), 
Example 13 (the signatures at the vertices of the call graph, which are now numbered $1$ to $4$, since the name $\Sigma_0$ of the signature at the vertex $0$ clashed with the set of sorts),  
Definition 25 (the semantics of composite terms), 
Definition 27, Definition 29, Definition 31 and Lemma 43, 
and in the caption of Figure 3 (the term shown there has width $4$, not $3$), among others.
\end{enumerate}

\noindent This is the full version of the revised paper, including the appendix with all proofs. The published version is available at \url{https://doi.org/10.4230/LIPIcs.LICS.2026.76}.

\bigskip

\noindent \emph{Dagstuhl Publishing -- September 25, 2026.}

\end{document}